\documentclass[11pt]{article}

\usepackage[utf8]{inputenc}
\usepackage{authblk}

\usepackage{amsfonts,amsmath,amssymb,mathtools,amsthm}
\usepackage{graphicx}
\usepackage{stmaryrd}
\usepackage{thmtools}
\DeclareGraphicsExtensions{.eps}
\usepackage{float}
\usepackage{comment}

\usepackage{xcolor}
\definecolor{linkcolour}{rgb}{0.15,0.15,0.55}
\definecolor{urlcolour}{rgb}{0.15,0.15,0.55}
\definecolor{citecolour}{rgb}{0.15,0.15,0.55}

\usepackage[linktoc=all,backref=page]{hyperref}
	\hypersetup{
		colorlinks = true,
		linkcolor = linkcolour,
		urlcolor = citecolour,
		citecolor = citecolour,
		linktoc = all,
		hypertexnames = false,
		unicode = true,
		bookmarksnumbered = false,
		pdfmenubar = true,
		pdftoolbar = true}

\usepackage[textsize=footnotesize]{todonotes}

\usepackage
	[top=2.5cm,
	bottom=2.5cm,
	left=2.5cm,
	right=2.5cm,
	a4paper]
	{geometry}
		\usepackage[margin=2cm]{caption}
\renewcommand{\theequation}{\thesection.\arabic{equation}}

\newcommand\encadremath[1]{\vbox{\hrule\hbox{\vrule\kern8pt
\vbox{\kern8pt \hbox{$\displaystyle #1$}\kern8pt}
\kern8pt\vrule}\hrule}}
\def\enca#1{\vbox{\hrule\hbox{
\vrule\kern8pt\vbox{\kern8pt \hbox{$\displaystyle #1$}
\kern8pt} \kern8pt\vrule}\hrule}}

\newcommand\framefig[1]{
\begin{figure}[bth]
\hrule\hbox{\vrule\kern8pt
\vbox{\kern8pt \vbox{
\begin{center}
{#1}
\end{center}
}\kern8pt}
\kern8pt\vrule}\hrule
\end{figure}
}

\newcommand\figureframex[3]{
\begin{figure}[bth]
\hrule\hbox{\vrule\kern8pt
\vbox{\kern8pt \vbox{
\begin{center}
{\mbox{\epsfxsize=#1.truecm\epsfbox{#2}}}
\end{center}
\caption{#3}
}\kern8pt}
\kern8pt\vrule}\hrule
\end{figure}
}
\newcommand\figureframey[3]{
\begin{figure}[bth]
\hrule\hbox{\vrule\kern8pt
\vbox{\kern8pt \vbox{
\begin{center}
{\mbox{\epsfysize=#1.truecm\epsfbox{#2}}}
\end{center}
\caption{#3}
}\kern8pt}
\kern8pt\vrule}\hrule
\end{figure}
}
 \makeatother

 \usepackage{mathdots}
 \usepackage{xcolor}
 \usepackage{amsfonts,amsmath,amssymb,mathtools,amsthm}
 \usepackage{graphicx}
 \usepackage{stmaryrd}
 \usepackage{thmtools}
 \DeclareGraphicsExtensions{.eps}
\usepackage{hyperref}

\renewcommand{\thesection}{\arabic{section}}
\renewcommand{\theequation}{\arabic{section}-\arabic{equation}}
\makeatletter
\@addtoreset{equation}{section}
\makeatother
\newtheorem{theorem}{Theorem}[section]
\newtheorem{conjecture}{Conjecture}[section]

\newtheorem{proposition}{Proposition}[section]
\newtheorem{lemma}{Lemma}[section]
\newtheorem{corollary}{Corollary}[section]

\theoremstyle{definition}
\newtheorem{remark}{Remark}[section]
\newtheorem{definition}{Definition}[section]

\def\br{\begin{remark}\rm\small}
\def\er{\end{remark}}
\def\bt{\begin{theorem}}
\def\et{\end{theorem}}
\def\bd{\begin{definition}}
\def\ed{\end{definition}}
\def\bp{\begin{proposition}}
\def\ep{\end{proposition}}
\def\bl{\begin{lemma}}
\def\el{\end{lemma}}
\def\bc{\begin{corollary}}
\def\ec{\end{corollary}}
\def\beaq{\begin{eqnarray}}
\def\eeaq{\end{eqnarray}}

\theoremstyle{definition}
\newtheorem{example}{Example}[section]

\newcommand{\td}{\tilde}

\newcommand{\be}{\begin{equation}}
\newcommand{\ee}{\end{equation}}
\newcommand{\beq}{\begin{equation}}
\newcommand{\eeq}{\end{equation}}
\newcommand{\bea}{\begin{eqnarray}}
\newcommand{\eea}{\end{eqnarray}}
\newcommand{\beqq}{\begin{equation*}}
\newcommand{\eeqq}{\end{equation*}}
\newcommand{\beaa}{\begin{eqnarray*}}
\newcommand{\eeaa}{\end{eqnarray*}}

\newtheorem{assumption}{Assumption}
\newcommand{\Res}{\mathop{\,\rm Res\,}}

\newcommand{\om}{\omega}

\newcommand\blfootnote[1]{%
  \begingroup
  \renewcommand\thefootnote{}\footnote{#1}%
  \addtocounter{footnote}{-1}%
  \endgroup
}

\title{\bf{Hamiltonian representation of isomonodromic deformations of general rational connections on $\mathfrak{gl}_2(\mathbb{C})$}}

\date{\vspace{-5ex}}
 
\author{$_{1}$Alexander Hock\footnote{Universit\'e de Gen\`eve, Section de math\'ematiques, 24 rue du G\'en\'eral Dufour, 1211 Gen\`eve 4, Suisse}\,\,,
$_{2}$Olivier Marchal\footnote{Universit\'{e} Jean Monnet Saint-\'{E}tienne, CNRS, Institut Camille Jordan UMR 5208, Institut Universitaire de France, F-42023, Saint-\'{E}tienne, France}\,\,,
$_{3}$Nicolas Orantin
}

\begin{document}

\begin{center}
\huge{Geometry of Logarithmic Topological Recursion: Dilaton Equations, Free Energies and Variational Formulas}
\end{center}
\vspace{0.5cm}
\begin{center}
$_{1}$Alexander Hock\footnote{Universit\'{e} de Gen\`{e}ve, Section de math\'{e}matiques, 24 rue du G\'{e}n\'{e}ral Dufour, 1211 Gen\`{e}ve 4, Suisse}\,\,,
$_{2}$Olivier Marchal\footnote{Universit\'{e} Jean Monnet Saint-\'{E}tienne, CNRS, Institut Camille Jordan UMR 5208, Institut Universitaire de France, F-42023, Saint-\'{E}tienne, France}\,\,,
$_{3}$Nicolas Orantin
\end{center}

\vspace{1.0cm}

\textbf{Abstract}:
One of the most important applications of topological recursion concerns spectral curves for which the functions $(x,y)$ defining the spectral curve are allowed to have logarithmic singularities. This occurs for instance for Seiberg-Witten curves and mirror curves computing Gromov--Witten invariants of toric Calabi--Yau threefolds. A recently introduced extension of topological recursion, the so-called logarithmic topological recursion, exhibits the correct behavior under certain limits of those spectral curves. 

In this article, we derive the dilaton equations in the setting of logarithmic topological recursion, as well as variational formulas, and provide a definition of the free energies in situations where standard topological recursion was known to fail. We present examples in which the new definition of the free energies \textit{directly} (without any computation) reproduces the full perturbative part of the Nekrasov--Shatashvili partition function of 4d $\mathcal{N}=2$ pure supersymmetric gauge theory, as well as the all-genus free energies of mirror curves of strip geometries, including in particular the topological vertex and the resolved conifold. 

\blfootnote{\textit{Email Addresses:}$_{1}$\textsf{alexander.hock@unige.ch}, $_{2}$\textsf{olivier.marchal@univ-st-etienne.fr}, $_{3}$\textsf{nicolas.orantin@gmail.com}}

\tableofcontents

\newpage

\section{Introduction}
\subsection{From Topological Recursion to Logarithmic Topological Recursion}
The topological recursion (TR) \cite{EO07} is a universal recursive procedure which defines an infinite family of symmetric multi-differentials $\omega_{h,n}^{\text{EO}}$ on $n$ copies of a Riemann surface $\Sigma$. The importance of TR lies in its wide range of applications. One of the simplest, yet most important, applications is the computation of intersection numbers on $\overline{\mathcal{M}}_{g,n}$, the moduli space of complex curves, for instance in relation to Witten's conjecture \cite{Witten}, proved by Kontsevich \cite{Kontsevich}. Other examples include Hurwitz theory \cite{Bouchard:2007hi}, Weil--Petersson volumes of the bordered moduli space \cite{Mirzakhani,EOMirzakhani}, and computations of Gromov--Witten invariants for toric Calabi--Yau threefolds \cite{BKMP,Eynard:2012nj,Fang:2013dna}. As one can see, the range of applications is very broad, and the few examples mentioned above are just the tip of the iceberg, since this list is far from complete. 

Since the original definition of TR \cite{EO07}, several extensions have been introduced, relaxing some of the initial assumptions. One example is the extension to higher order topological recursion \cite{BouchardEynard,HigherRam}, which allows for higher order ramification points, originally assumed to be simple. The correct way to define such an extension comes from the requirement that taking limits of several coalescing simple ramification points should yield the more general higher order TR. This extension has applications in the computation of spin intersection numbers on $\overline{\mathcal{M}}_{g,n}$ \cite{Wittenrspin}.

As mentioned above, another important application of TR is the computation of Gromov--Witten invariants of toric Calabi--Yau threefolds, where the spectral curve corresponds to the mirror curve. One important assumption necessary for applying TR to mirror curves is the so-called framing parameter, which must be chosen to be generic. Indeed, it was observed that taking singular limits of mirror curves leads to spectral curves where TR fails to be applicable \cite{Bouchard:2011ya}. The problem lies in the fact that mirror curves are spectral curves in $\mathbb{C}^*\times \mathbb{C}^*$, where some assumptions of the original definition of TR are violated, but are nevertheless effectively restored by choosing a generic framing parameter. 

More recently, a new extension of TR was introduced that precisely resolves this problem of taking singular limits of mirror curves. This extension is called \textit{logarithmic topological recursion} (LogTR) \cite{Alexandrov:2023tgl}. The remarkable observation that LogTR commutes with the singular limit of mirror curves was, in fact, not the original motivation for its definition. 

The origin of LogTR comes from a completely different direction related to a newly understood duality in the theory of TR. The so-called $x$-$y$ duality \cite{Alexandrov:2022ydc,Hock:2022wer,Hock:2022pbw} relates two infinite families of correlators defined on the same spectral curve, but projected either via $x$ or $y$. This duality is as universal as TR itself and has already reproduced many important explicit results in various enumerative problems \cite{Hock:2023qii,Hock:2025wlm,Alexandrov:2023oov}. 

However, the universal $x$-$y$ duality was observed to fail precisely in cases where TR does not reproduce the correct enumerative invariants \cite{Hock:2023dno}. This led to the proposal of an extension of TR obtained by enforcing the validity of the $x$-$y$ duality. This was the origin of the definition of LogTR. It was only later realized that this definition also resolves the problem of singular limits in mirror curves. 

Therefore, the understanding of the $x$-$y$ duality has reshaped our perspective on topological recursion and has led to new equivalent global definitions of TR \cite{Alexandrov:2024tjo}, which are, however, beyond the scope of this article. We only mention that this new viewpoint has also provided a deeper understanding of integrable systems and has shown that TR generally gives rise to KP integrability \cite{Alexandrov:2023jcj,Alexandrov:2024hgu} and has deep connections with isomonodromic deformations \cite{MO19_hyper,MOsl2,Quantization_2021}.\\

This article aims to provide a new geometric understanding of LogTR. Important formulas relating $\omega^{\text{EO}}_{h,n+1}$ to $\omega^{\text{EO}}_{h,n}$, and thus ultimately leading to a consistent definition of the free energies, are the so-called dilaton equations. The dilaton equations are closely related to variational formulas, which arise when considering families of spectral curves and studying how variations of the spectral curve affect the correlators $\omega^{\text{EO}}_{h,n}$ \cite{BorotGeometry}. 

For LogTR, no dilaton equations were derived in \cite{Alexandrov:2023tgl}, although general variational formulas were provided. We close this gap in the literature by deriving the dilaton equations in LogTR (\autoref{TheoremDilatonEquation}) and thereby obtaining a new definition of the free energies in this setting (\autoref{DefFreeEnergies}). 

Finally, we derive more explicit variational formulas, for instance describing how the correlators defined by LogTR vary with respect to specific logarithmic poles (\autoref{TheoVariationsLogTRpoles}). We present several examples in which our new definition of the free energies in the context of LogTR directly reproduces the full perturbative part of the Nekrasov--Shatashvili partition function of 4d $\mathcal{N}=2$ pure supersymmetric gauge theory, as well as the all-genus free energies of mirror curves of strip geometries, including in particular the topological vertex and the resolved conifold.

\subsection{Structure of the article}

We begin in \autoref{sec.defLogTR} with the set-up for this article, including the allowed geometry of the spectral curves and a review of the definition of LogTR. Then, in \autoref{sec.PropLogTR}, we list properties of LogTR which are known and which essentially follow directly from its definition. 

Next, in \autoref{sec.Dilaton}, we state the dilaton equations, starting in \autoref{sec.dilatonomehn} with the dilaton equation for the differentials $\omega_{h,k}$, and continuing in \autoref{sec.dilatonfreeenergy} with the dilaton equation for $F_h:=\omega_{h,0}$ providing a definition of the free energies.  We conclude this section with examples in \autoref{sec.dilatonfreeenergyexamples}, which already provide nontrivial results for the free energies for some interesting spectral curves.  

In \autoref{sec.parametrizeom01}, we study how $\omega_{0,1}$ parametrizes the space of spectral curves in LogTR. In this context, the local decomposition is studied in \autoref{Sec:localcorrdinate}, and the global decomposition in \autoref{SectionGlobalDec}. The resulting decomposition depends on an additional set of parameters, namely the positions of the so-called LogTR-vital singularities, which allow one to vary in the space of spectral curves. This leads to a parametrization of the spectral curve in terms of the one-form $\omega_{0,1}$ in \autoref{sec.parametrizespectralcurve}, which is further rewritten via the Bergman kernel in \autoref{SubsectionLambda}. 

The parametrization of the space of spectral curves then gives rise to the variational formulas studied in \autoref{sec.variationalformula} for LogTR. First, we provide a list of parameters in \autoref{sec.variationsof}, and then derive the variational formulas with respect to the classical parameters in \autoref{SectionVartdydx}, as well as with respect to the LogTR-vital singularities in \autoref{SectionVarVitalSing}, including also the free energies. These new variational formulas show that the dilaton equations and the definition of the free energies in LogTR are compatible with the variational formulas. This confirms that the definition of the free energies in \autoref{sec.dilatonfreeenergy} is the correct one, which is further supported by the examples.

Most of the proofs are deferred to \autoref{AppendixDilatonProof}--\autoref{AppendixVarF1LogTR} due to their length.

\subsection{Notations and reminders}
In this article, we shall use the following notations:

\begin{itemize}
    \item $\llbracket n\rrbracket=\llbracket 1,n\rrbracket= \{1,\dots,n\}$ for any $n\geq 1$. By convention $\llbracket 0\rrbracket=\emptyset$.
    \item $z_I$ denote $(z_i)_{i\in I}$ for any finite and non-empty set $I \subset \mathbb{N}$.
    \item $[u^d]$ denotes the operator that extracts the $d^{\text{th}}$ coefficient in the formal series expansion in $u$ from the whole expression to the right of it, that is: $[u^d]\left(\underset{k=-\infty}{\overset{+\infty}{\sum}} f_ku^k\right):=f_k$.
    \item $\mathcal{S}(u):=\frac{e^{\frac{u}{2}} -e^{-\frac{u}{2}} }{u}= \frac{2\sinh \frac{u}{2}}{u}$. We shall also define the sequence $(\beta_{2k})_{k\geq 0}$:
    \beq \frac{1}{\mathcal{S}(u)}=\sum_{k=0}^{\infty} \beta_{2k}u^{2k}\,\,,\,\, \beta_{2k}=\frac{(2^{1-2k}-1)B_{2k}}{(2k)!} \,,\,\forall\, k\geq 0 \eeq
    where $(B_k)_{k\geq 0}$ are the Bernoulli numbers.
    \item $\delta_{i,j}$ is the Kronecker symbol.
\end{itemize}
Moreover, for a compact Riemann surface $\Sigma$ of genus $g$, we shall use the following notations:
\begin{itemize}
    \item $\left(\mathcal{A}_i,\mathcal{B}_i\right)_{1\leq i\leq g}$ will denote a basis of homology cycles (i.e. $\mathcal{A}_i\cap \mathcal{B}_j=\delta_{i,j}$). We shall also refer to it as a choice of Torelli marking on $\Sigma$.
    \item For a given choice of Torelli marking, $\left(du_i\right)_{1\leq i\leq g}$ will denote the basis of holomorphic one-forms on $\Sigma$ that is normalized on the $\mathcal{A}$-cycles:
\beq \forall \,(i,j)\in \llbracket 1,g\rrbracket^2\,:\, \oint_{\mathcal{A}_i}du_j=\delta_{i,j}\eeq
and the corresponding Riemann's matrix of periods $\left(\tau_{i,j}\right)_{1\leq i,j\leq g}$ is defined by 
\beq  \forall \,(i,j)\in \llbracket 1,g\rrbracket^2\,:\, \tau_{i,j}:=\oint_{\mathcal{B}_i}du_j\eeq 
\item We shall denote $B$ the Bergman kernel: it is the unique bi-differential with a residueless double pole on the diagonal and that is normalized on the $\mathcal{A}$-cycles:
\bea B(p,q)&\overset{q\to p}{=}&\frac{dz(q)dz(p)}{(z(q)-z(p))^2} +O(1) \text{ in any local coordinate $z$,}\cr
\oint_{\mathcal{A}_i}B&=&0\,\,,\forall\,i \in \llbracket1,g\rrbracket.\eea
\item The modified third kind differential based on $(z_1,z_2)\in \Sigma^2$ shall be denoted $dS_{z_1,z_2}$ and is defined by
\beq dS_{z_1,z_2}(q)=\int_{z_2}^{z_1} B(.,q)\eeq
It is the unique meromorphic differential on $\Sigma$ with only two simple poles at $q=z_1$ with residue $+1$ and $q=z_2$ with residue $-1$ and normalized on the $\mathcal{A}$-cycles: $\oint_{q\in\mathcal{A}_i} dS_{z_1,z_2}(q)=0$ for all $i\in \llbracket 1,g\rrbracket$.
\item The prime form associated to $\Sigma$, denoted $E(p,q)$ is the unique $\left(-\frac{1}{2},-\frac{1}{2}\right)$ form such that it has no pole and only simple zeros on the diagonal $p=q$. In any local coordinates $z$:
\beq E(p,q)\overset{p\to q}{\sim}\frac{z(p)-z(q)}{\sqrt{dz(p)}\sqrt{dz(q)}} \eeq
It is related to the modified third kind differential by $dS_{z_1,z_2}(q)=d_q\ln \frac{E(q,z_1)}{E(q,z_2)}$.
\item  We shall denote $S_B$ the Bergman projective connection on $\Sigma$. It is defined by 
\beq B(p,q)\overset{q\to p}{=}\left(\frac{1}{(z(p)-z(q))^2}+\frac{1}{6} S_B(z(p))+ O(z(q))\right)dz(q)dz(p)\eeq
for any local coordinate $z$ around the diagonal.
\item For a meromorphic one-form $dx$ on $\Sigma$, the ramification points $(p_i)_{1\leq i\leq N}$ are defined as the zeros of $dx$. A simple ramification point $p_i$ defines a local involution, denoted $\sigma_i$, such that $x(\sigma_i(q))=x(q)$ locally around $p_i$. Following the notation of \cite{EO07}, we shall denote
\bea\label{DefVertexPropagator}
dE_{i,q}(p)&:=&\frac{1}{2}\int_q^{\sigma_i(q)} B(.,p)=\frac{1}{2} dS_{q,\sigma_i(q)}(p)\cr
\text{and }\omega_i(p)&:=&(y(p)-y(\sigma_i(p)))dx(p)
\eea
for a given meromorphic one-form $dy$ on $\Sigma$. These quantities are defined only locally around the simple ramification points $p_i$. When $dx$ has only simple ramification points the Bergman tau-function $\tau_B$ is defined by
\beq \forall \,i \in \llbracket1,N\rrbracket\,:\, \partial_{p_i}[\ln \tau_B]=-\frac{1}{12}S_B(p_i)\eeq 
\end{itemize}
Due to its importance in the present work, we also remind Riemann bilinear identity and one of its corollary.

\begin{proposition}[Riemann bilinear identity]\label{RiemannBilinearIdentity}Let $\Sigma$ be a compact Riemann surface of genus $g$ and $\left(\mathcal{A}_i,\mathcal{B}_i\right)_{1\leq i\leq g}$ a basis of homology cycles. Let $\omega_1$ and $\omega_2$ be two meromorphic differentials on $\Sigma$. Let $p_0\in\Sigma$ be an arbitrary basepoint, we consider the function $\Phi_1$ defined on the fundamental domain by
\beq \Phi_1(p):=\int^p_{p_0} \omega_1\eeq
Riemann bilinear identity states that
\beq \sum_{a\in \{\text{all poles }\omega_1 \text{ and }\omega_2\}}\Res_{q\to a} \Phi_1(q)\omega_2(q)= \frac{1}{2i\pi}\sum_{i=1}^g\oint_{\mathcal{A}_i}\omega_1\oint_{\mathcal{B}_i}\omega_2- \oint_{\mathcal{B}_i}\omega_1\oint_{\mathcal{A}_i}\omega_2
\eeq
Consequently for any meromorphic one-form $\omega$ on $\Sigma$:
\beq \label{CorollaryRiemannBilinearIdentity} 
\omega(p)= \sum_{a\in \{\text{all poles of }\omega\}}\Res_{q\to a} dS_{p_0,q}(p) \omega(q) + \sum_{i=1}^g du_i(p) \oint_{\mathcal{A}_i} \omega 
\eeq
\end{proposition}


\section{Definition of LogTR and its known properties}
\subsection{Definition of logarithmic topological recursion}\label{sec.defLogTR}
To define the Topological Recursion (TR) \cite{EO07}, one needs a compact Riemann surface $\Sigma$ with a choice of Torelli marking and two one-forms $dx$ and $dy$ on $\Sigma$. In the original definition of TR \cite{EO07}, the two one-forms are assumed to be meromorphic on $\Sigma$ without residue so that the corresponding function $x$ and $y$ are meromorphic functions on $\Sigma$. However, an important application of TR was found in Gromov--Witten theory for toric Calabi--Yau threefolds \cite{BKMP}. In this setting, the assumption on $x$ and $y$ has to be relaxed so that the differentials $dx$ and $dy$ are meromorphic but with possible residues at their poles. This allows $x$ and $y$ to have locally logarithmic singularities, where $dx$ and $dy$ have non-vanishing residues. Luckily, the original definition of TR provides the correct result in the presence of a so-called generic framing parameter, which ensures that $dx$ has a non-vanishing residue whenever $dy$ has a non-vanishing residue. However, it was observed that if a specific framing is chosen such that this property is not satisfied, TR does not compute the correct Gromov--Witten invariants and does not commute with limits to those specific points \cite{Bouchard:2011ya,Gukov:2011qp}. The notion of framing originates in physics, namely in topological string theory. In other areas where TR finds applications, such a generic framing may not be available. Therefore, one seeks a new definition that behaves correctly under limits and produces the correct Gromov--Witten invariants already for special framings, without assuming genericity.

\medskip

The refinement of TR resolving this problem is Logarithmic Topological Recursion (LogTR). Indeed, it incorporates in a specific way the non-vanishing residues of $dy$ if $dx$ is regular at those points. Moreover, an additional motivation for this refinement is the observation \cite{Hock:2023dno} that LogTR extends TR in the context of  the so-called $x$-$y$ duality \cite{Alexandrov:2022ydc,Hock:2022wer,Hock:2022pbw}. The final definition of LogTR has been introduced in \cite{Alexandrov:2023tgl} but several points are missing, such as the definition of free energies, the dilaton equations, and variational formulas. These relations play a crucial role in standard TR especially regarding quantization and integrable systems. In order to recap the definition and proceed with the proof of these new properties, we need the following initial data:

\begin{definition}[Initial data for LogTR]\label{InitialData} The initial data for LogTR are the following:
\begin{itemize}
    \item A Riemann surface $\Sigma$ of genus $g\geq 0$.
    \item A choice of Torelli marking for $\Sigma$, i.e. of homology cycles $\left(\mathcal{A}_i,\mathcal{B}_i\right)_{1\leq i\leq g}$ such that $\mathcal{A}_i\cap \mathcal{B}_j=\delta_{i,j}$. This choice is equivalent to the choice of a Bergman kernel $B$, i.e. the unique meromorphic bi-differential with a double pole along the diagonal with bi-residue $1$ and no further singularities and whose $\mathcal{A}$-periods vanish.
    \item Two meromorphic differentials $dx$ and $dy$ on $\Sigma$. The zeros of $dx$ are called ramification points and denoted $\text{Ram}:=\{p_1,\dots,p_N\}$. We shall denote $\mathcal{P}_x$ the set of poles of $dx$ and $\mathcal{P}_y$ the set of poles of $dy$. We shall denote $\mathcal{P}=\mathcal{P}_x\cup \mathcal{P}_y$. Finally, we shall denote $\mathcal{S}_x\subset \mathcal{P}_x$ the set of all poles of $dx$ with non-vanishing residues and $\mathcal{S}_y\subset \mathcal{P}_y$ the set of all poles of $dy$ with non-vanishing residues.
\end{itemize}
\end{definition}

In fact, we shall only consider a subset of initial data by requiring the following admissibility conditions:

\begin{assumption}[Admissible initial data for LogTR]\label{MainAssumption}In this paper, we shall always make the following assumptions:
\begin{itemize}
    \item The Riemann surface $\Sigma$ is compact.
    \item The ramification points are always simple zeros of $dx$.
    \item $dy$ is regular at the ramification points.
    \item The zero loci of $dy$ and the zero loci of $dx$ are disjoint.
\end{itemize}
These assumptions will be referred to as ``admissible conditions" for the initial data and we shall refer to $(\Sigma,B,dx,dy)$ as an admissible spectral curve or admissible initial data for LogTR.
\end{assumption}

\begin{remark}The first three assumptions are the standard assumptions of the original version of TR. There exist generalizations of TR \cite{HigherRam,Bouchard:2023yau,DN1} that deal with some of these assumptions but we let the generalizations of these cases to LogTR for future works.
\end{remark}

In the standard TR, $x$ and $y$ are meromorphic functions on $\Sigma$ and thus $dx$ and $dy$ cannot have poles with non-vanishing residue (in particular they cannot have simple poles). This originates from the fact that $x$ and $y$ are usually associated to an algebraic equation $P(x,y)=0$. On the contrary, in topological string theory, the spectral curve is given by an algebraic equation on $\mathbb{C}^*$ of the form $P(e^x,e^y)=0$ so that $X=e^x$ and $Y=e^y$ are meromorphic functions but $dx=\frac{dX}{X}$ and $dy=\frac{dY}{Y}$ have simple poles with non-vanishing residues at every zero or pole of $X$ or $Y$. Thus, in the LogTR setting, $dx$ and $dy$ are arbitrary meromorphic one-forms on $\Sigma$ and thus may have simple poles or poles with non-vanishing residues. When this happens, $x$ and $y$ are no longer meromorphic functions on $\Sigma$ because they locally exhibit logarithmic singularities that we will describe in the following definition.

\begin{definition}[Logarithmic cuts for poles of $dx$ and $dy$ with non-vanishing residues]\label{Def22} Let $o$ be a given basepoint that is assumed to be away from ramification points, poles/zeros of $dx$, poles/zeros of $dy$ and the representative $(\mathcal{A}_i,\mathcal{B}_i)$ cycles. Moreover, we choose the basepoint $o$ so that $x(o)\neq x(a)$ for any pole $a$ of $dx$ or $dy$ and such that $x(o)\neq x(p_i)$ for any $i\in \llbracket 1,N\rrbracket$. We shall use this basepoint to define the logarithmic cuts and the domain of $x$ and $y$.
\begin{itemize}
\item For any $s\in \mathcal{S}_x\cup \mathcal{S}_y$ (i.e. a pole of $dx$ or $dy$ with non-vanishing residue), we select a logarithmic cut connecting $s$ to $o$, i.e. a smooth contractible oriented curve connecting $s$ to $o$ such that it avoids all other poles/zeros of $dx$, poles/zeros of $dy$ and the holonomy cycles. We shall denote $\mathcal{C}_{o\to s}$ this logarithmic cut. Moreover, we shall take the logarithmic cuts such that they only intersect at $o$.
\item The differential $dx$ admits an antiderivative $x$ defined on $\Sigma\setminus\left(\mathcal{P}_x\cup \underset{s\in \mathcal{S}_x}{\bigcup} \mathcal{C}_{o\to s}\right)$.
\item The differential $dy$ admits an antiderivative $y$ defined on $\Sigma\setminus\left(\mathcal{P}_y \cup\underset{s\in \mathcal{S}_y}{\bigcup} \mathcal{C}_{o\to s}\right)$.
\item This choice of logarithmic cuts is equivalent to the choice of logarithmic branch of $\ln \frac{E(q,s)}{E(q,o)}$ for any point $s\in \mathcal{S}_x\cup \mathcal{S}_y$. 
\end{itemize}

\end{definition}

\begin{remark}Our choice of logarithmic cuts defines a contractible star-shaped domain based at $o$ (where $dx$ and $dy$ are regular) but other choices could be made. In particular, it is important to notice that all interesting quantities (correlation functions, free energies, etc.) should and will be independent of the choice of logarithmic cuts. 
\end{remark}

In the following definition of LogTR, only specific logarithmic singularities, known as LogTR-vital singularities, will play a role.

\begin{definition}[LogTR-vital singularities] A singularity $s\in \mathcal{S}_y$ is called a LogTR-vital singularity if $s$ is a simple pole of $dy$ and $dx$ is regular at $s$. We shall denote $\mathcal{V}=\{a_1,\dots,a_M\}\subset \mathcal{S}_y$ the finite set of all LogTR-vital singularities and $y_{a_s}:=\underset{q\to a_s}{\Res} dy$ the corresponding residue.
\end{definition}

We are now ready to define the LogTR procedure.

\begin{definition}[Definition of LogTR]\label{DefLogTR}Let $(\Sigma,B,dx,dy)$ be an admissible spectral curve. Denote $(p_i)_{1\leq i\leq N}$ the ramification points, $(a_k)_{1\leq k\leq M}$ the LogTR-vital singularities and $(y_{a_k})_{1\leq k\leq M}$ their corresponding residue. Finally, let $\sigma_i$ be the deck transformations of $x$ near $p_i$ for any $i\in \llbracket 1,N\rrbracket$.
We define the LogTR correlations functions $\left(\omega_{h,n}\right)_{h\geq 0,n\geq 1}$ by the following induction \cite{Alexandrov:2023tgl}:
\begin{itemize}
    \item $\omega_{0,1}:=ydx$ and $\omega_{0,2}:=B$.
    \item For any $h\geq 0$ and $m\geq 1$ such that $2h+m-2>0$ we define by induction 
    \bea \label{LogTRDef} \omega_{h,m}(z_1,\dots,z_m)&:=&\frac{1}{2}\sum_{i=1}^N\Res_{z\to p_i}\frac{\int_{z}^{\sigma_i(z)} \omega_{0,2}(z_1,.)}{\omega_{0,1}(\sigma_i(z))-\omega_{0,1}(z)}\Big(\omega_{h-1,m+1}(z,\sigma_i(z),z_2,\dots,z_m)\cr
    &&+\sum_{\substack{h_1+h_2=h\\I_1\sqcup I_2=\{2,\dots, m\} \\(h_i,|I_i|)\neq (0,0)}} \omega_{h_1,|I_1|+1}(z,z_{I_1}) \omega_{h_2,|I_2|+1}(\sigma_i(z),z_{I_2}) \Big)\cr
    &&-\delta_{m,1}\sum_{s=1}^M\Res_{z\to a_s}\left(\int_{a_s}^z\omega_{0,2}(z_1,.)\right)dx(z)[\hbar^{2h}]\left(\frac{y_{a_s}}{\mathcal{S}(y_{a_s}^{-1}\hbar \partial_x)}\ln(z-a_s) \right)\cr&&
    \eea   
\end{itemize}
\end{definition}

Note that the last line in \eqref{LogTRDef} is the new contribution arising in LogTR and it only contributes for $m=1$. Thus the correlations functions $(\omega_{0,n})_{n\geq 3}$ are the same as those obtained from standard TR. On the contrary, $\omega_{1,1}$ will have an additional contribution and this new contribution will propagate in the induction and contributes to $\omega_{1,2}$, $\omega_{1,3}$, etc. Note also that we have used $dy$ in the initial data. We will have to choose logarithmic cuts for $\omega_{0,1}$, as will be discussed in \autoref{sec.parametrizeom01} and accordingly to \autoref{Def22}. However, for all $\omega_{g,n}$ it is sufficient to work with $dy$, since by recursion only derivatives of $y$ appear for $2g+n-2>0$.

\begin{remark} If there are no LogTR-vital singularities, i.e. $\mathcal{V}=\emptyset$, then the definition of LogTR is the same as the standard TR. In particular, this is exactly the situation if a generic framing is chosen for the mirror curve of toric Calabi-Yau threefolds. From the perspective of TR, a framing is the choice of an integer $f\in \mathbb{Z}$ such that the two ingredients for the spectral curve are chosen to be $(x_f,y)=(x+fy,y)$. For generic $f$, $d(x+fy)$ has residues at all points where $dy$ has residues, thus the set of LogTR-vital points is empty.
\end{remark}

\begin{remark}
The additional term appearing in LogTR involves taking residues at the \emph{LogTR-vital singularities} $\mathcal{V}=\{a_1,\dots,a_M\}\subset \mathcal{S}_y$. In principle, one could include all points of $\mathcal{S}_y$ in the definition of LogTR. However, only the LogTR-vital points $\mathcal{V}\subset \mathcal{S}_y$ yield non-vanishing contributions after performing the residue computation.
\end{remark}

\begin{definition}[Alternative definition of LogTR]\label{DefLogTRAlternative} Using the same setup as in \autoref{DefLogTR} LogTR  is equivalent to
\bea \label{LogTRDefAlter} \omega_{h,m}(z_1,\dots,z_m)&=&\frac{1}{2}\sum_{i=1}^N\Res_{z\to p_i}\frac{\int_{z}^{\sigma_i(z)} \omega_{0,2}(z_1,.)}{\omega_{0,1}(\sigma_i(z))-\omega_{0,1}(z)}\Big(\omega_{h-1,m+1}(z,\sigma_i(z),z_2,\dots,z_m)\cr
    &&+\sum_{\substack{h_1+h_2=h\\I_1\sqcup I_2=\{2,\dots, m\} \\(h_i,|I_i|)\neq (0,0)}} \omega_{h_1,|I_1|+1}(z,z_{I_1}) \omega_{h_2,|I_2|+1}(\sigma_i(z),z_{I_2}) \Big)\cr
    &&+\delta_{m,1}\sum_{s=1}^M\left(\frac{\partial^{2h-2}}{\partial x(q)^{2h-2}}\frac{\omega_{0,2}(z_1,q)}{dx(q)}\right)_{|\, q=a_s} [\hbar^{2h}]\left(\frac{y_{a_s}}{\mathcal{S}(y_{a_s}^{-1}\hbar)}\right)\cr
    &=&\frac{1}{2}\sum_{i=1}^N\Res_{z\to p_i}\frac{\int_{z}^{\sigma_i(z)} \omega_{0,2}(z_1,.)}{\omega_{0,1}(\sigma_i(z))-\omega_{0,1}(z)}\Big(\omega_{h-1,m+1}(z,\sigma_i(z),z_2,\dots,z_m)\cr
    &&+\sum_{\substack{h_1+h_2=h\\I_1\sqcup I_2=\{2,\dots, m\} \\(h_i,|I_i|)\neq (0,0)}} \omega_{h_1,|I_1|+1}(z,z_{I_1}) \omega_{h_2,|I_2|+1}(\sigma_i(z),z_{I_2}) \Big)\cr
    &&+\delta_{m,1}\sum_{s=1}^M\left(\frac{\partial^{2h-2}}{\partial x(q)^{2h-2}}\frac{\omega_{0,2}(z_1,q)}{dx(q)}\right)_{|\, q=a_s} y_{a_s}^{1-2h}\frac{(2^{1-2h}-1)B_{2h}}{(2h)!}
    \eea   
for any $h\geq 0$ and $m\geq 1$ such that $2h+m-2>0$. In this formula, $(B_k)_{k\geq 0}$ stands for the Bernoulli numbers.
\end{definition}

\begin{proof}By integration by parts with local meromorphic differentials $\frac{dz}{z-a_s}$ and $\omega_{0,2}(z_1,z)$ we can rewrite the LogTR special term in the induction as
\bea &&-\delta_{m,1}\sum_{s=1}^M\Res_{z\to a_s}\left(\int_{a_s}^z\omega_{0,2}(z_1,.)\right)dx(z)[\hbar^{2h}]\left(\frac{y_{a_s}}{\mathcal{S}(y_{a_s}^{-1}\hbar \partial_x)}\ln(z-a_s) \right)\cr&&
=\delta_{m,1}\sum_{s=1}^M\left(\frac{\partial^{2h-2}}{\partial x(q)^{2h-2}}\frac{\omega_{0,2}(z_1,q)}{dx(q)}\right)_{|\, q=a_s}[\hbar^{2h}]\left(\frac{y_{a_s}}{\mathcal{S}(y_{a_s}^{-1}\hbar)}\right)
\eea
\end{proof}

Formulations provided by \autoref{DefLogTRAlternative} are convenient for practical computations but are less practical in the theoretical proofs contrary to the residue formulation of \autoref{DefLogTR}.

\subsection{Basic properties of LogTR}\label{sec.PropLogTR}
The correlators $\left(\omega_{h,n}\right)_{h\geq 0,n\geq 1}$ constructed from LogTR  in \autoref{DefLogTR} satisfy many beautiful properties \cite{Alexandrov:2023tgl,Alexandrov:2024tjo} that we summarize in the following proposition.

\begin{proposition}[LogTR basic properties]\label{PropLogTRBasicProperties} The correlators $\left(\omega_{h,n}\right)_{h\geq 0,n\geq 1}$ constructed from LogTR in \autoref{DefLogTR} satisfy the following properties:
\begin{itemize}
    \item The correlators $\omega_{h,n}$ are meromorphic $n$-forms on $\Sigma$ except possibly for $(h,n)=(0,1)$.
    \item The correlators $\omega_{h,n}$ are symmetric forms for $n\geq 2$ and $h\geq 0$.
     \item For $2h+n-2>0$ and $n\geq 2$, the correlators $\omega_{h,n}$ have poles only at the ramification points $(p_1,\dots,p_N)$.
     \item For $h\geq 1$, the correlators $\omega_{h,1}$ have poles only at the ramification points $(p_1,\dots,p_N)$ and at the LogTR-vital singularities $(a_1,\dots,a_M)$. 
    \item For $2h+n-2\geq 0$, any $\omega_{h,n}$ is residue-free.
    \item For any $(h,n)\neq (0,1)$, the correlators $\omega_{h,n}$ have vanishing periods along the $\mathcal{A}$-cycles:
    \beq \label{VanishingAperiods}\forall\, (h,n)\neq (0,1)\,,\, \forall\, i\in \llbracket 1,g\rrbracket\,:\, \oint_{\mathcal{A}_i} \omega_{h,n}=0\eeq
    \item Linear loop equations (LLE): For any $i\in \llbracket 1,N\rrbracket$ and any $(h,m)\in \mathbb{N}^2$ such that $2h+m-1>0$, the quantity
    \beq \label{LinearLoopEquations}\omega_{h,m+1}(z_1,\dots,z_m,z)+\omega_{h,m+1}(z_1,\dots,z_m,\sigma_i(z)) \eeq 
    is holomorphic in a neighborhood of $p_i$ and has at least a simple zero in $z$ at $z = p_i$.
    \item Quadratic loop equations (QLE): For any $i\in \llbracket 1,N\rrbracket$ and any $(h,m)\in \mathbb{N}^2$ such that $2h-2+m>0$, the quadratic differential in $z$
\beq \label{QuadraticLoopEquations}\omega_{h-1,m+2}(z_1,\dots,z_m,z,\sigma_i(z))
    +\sum_{\substack{h_1+h_2=h\\I_1\sqcup I_2=\{1,\dots, m\}}} \omega_{h_1,|I_1|+1}(z_{I_1},z)\, \omega_{h_2,|I_2|+1}(z_{I_2},\sigma_i(z))\eeq
is holomorphic in a neighborhood of $p_i$ and has at least a double zero at $z = p_i$.
    \item Logarithmic projection property (LPP): For any $h\geq 0$, $m\geq 1$ such that $2h-2+m>0$: 
\beq\label{LogProjectionProperty} \omega_{h,m}(z_1,\dots,z_m)=\sum_{p\in \{p_i\}_{1\leq i\leq N}\sqcup \{a_j\}_{1\leq j\leq M}}\Res_{z'\to p}\left(\int_{p}^{z'}B(.,z_1)\right)\omega_{h,m}(z',z_2,\dots,z_m)\eeq
for $m>1$, the set $\{a_j\}_{1\leq j\leq M}$ does not contribute and can be removed.
\end{itemize}
\end{proposition}

Note that the LPP immediately follows from Riemann bilinear identity \eqref{RiemannBilinearIdentity} applied to $\omega_1=B$ and $\omega_2=\omega_{h,m}$. Indeed, both differentials are normalized on the $\mathcal{A}$-cycles and the only missing pole is the double pole of $B$ at $z'=z_1$.

\begin{remark}\label{remarklooplog}
    The projection of $\omega_{h,1}$ to the singular part at the LogTR-vital singularity $a_i$ has several equivalent formulations which are very explicit. Those can be understood as the \textit{loop equations at the LogTR-vital singularity}. However, this terminology is rather misleading since it is neither a linear nor a quadratic equation. In fact, we have the following representations of $\omega_{h,1}$ projected to the LogTR-vital singularity $a_s$:
    \begin{align*}
        \Res_{z'\to a_s}\left(\int_{a_s}^{z'}B(.,z)\right)\omega_{h,1}(z')
        =&-\Res_{z\to a_s}\left(\int_{a_s}^{z'}B(.,z)\right)dx(z')[\hbar^{2h}]\left(\frac{y_{a_s}}{\mathcal{S}(y_{a_s}^{-1}\hbar \partial_{x(z')})}\ln('z-a_s) \right)\\
=&\left(\frac{\partial^{2h-2}}{\partial x(z')^{2h-2}}\frac{B(z,z')}{dx(z')}\right)_{|\, z'=a_s}[\hbar^{2h}]\left(\frac{y_{a_s}}{\mathcal{S}(y_{a_s}^{-1}\hbar)}\right)\cr
    \end{align*}
\end{remark}
\begin{remark}
    In the standard TR setting, the linear loop equations, the quadratic loop equations together with the projection property to the ramification points is equivalent to the definition of standard TR \cite{Borot:2013lpa}. Analogously, \autoref{DefLogTR} of LogTR is the unique solution to the linear loop equations, the quadratic loop equations, the loop equations at the LogTR-vital singularity of \autoref{remarklooplog} and the logarithmic projection properties.
\end{remark}

\section{Dilaton equations and definition of the free energies}\label{sec.Dilaton}
In standard TR, the dilaton equation is a simple formula that provides $\omega^{\text{EO}}_{h,n}$ in terms of $\omega^{\text{EO}}_{h,n+1}$ by multiplying with a primitive of $\Phi(z)=\int_o^z\omega_{0,1}$ and taking residues at all ramification points
\begin{equation}\label{standarddilaton}
    (2-2h-n)\omega^{\text{EO}}_{h,n}(z_1,...,z_n)=\sum_{p\in \text{Ram}}\Res_{z'\to p}\Phi(z')\omega^{\text{EO}}_{h,n+1}(z',z_1,...,z_n).
\end{equation}
At first glance, such a formula may not seem very surprising, since $\omega^{\text{EO}}_{h,n+1}$ is recursively defined from $\omega^{\text{EO}}_{h,n}$. However, the importance of the dilaton equations lies in the fact that they allow one to define the set of free energies $(F_h^{\text{EO}}:=\omega^{\text{EO}}_{h,0})_{h\geq 2}$, by the extension of \eqref{standarddilaton}:
\begin{equation}
    (2-2h)\omega^{\text{EO}}_{h,0}=\sum_{p\in \text{Ram}}\Res_{z'\to p}\Phi(z')\omega^{\text{EO}}_{h,1}(z').
\end{equation}
Thus, the dilaton equations are central to define the free energies that are non-trivial functions of the moduli of the spectral curve. 

\medskip

The aim of this section is to find the dilaton equations for LogTR and thereby to define the free energies $(F_h)_{h\geq 2}$ in the presence of LogTR-vital singularities. In \cite{Alexandrov:2023tgl}, where LogTR was defined, no dilaton equations were found and no definition of the free energies was proposed. A naive guess for a possible dilaton equation would be a direct extension of the standard TR dilaton equations, possibly including residues at the LogTR-vital singularities. This, however, does not work for several reasons. First, logarithmic singularities appear when integrating $ydx$, so residues at those points are not defined. Secondly, if one only takes residues at the ramification points, as in the standard TR setting then one cannot obtain poles of $\omega_{h,1}$ at the logTR-vital singularities and thus the formula cannot hold.

Some further motivations can be gained from topological string theory, where TR computes Gromov--Witten invariants of toric Calabi--Yau threefolds. For this class of examples (as mentioned before), standard TR already computes the correct invariants $\omega_{h,n}^{\text{EO}}$, but only in the presence of a generic framing. Thus, all differentials $\omega_{h,n}^{\text{EO}}$ associated with a mirror curve depend on the framing. 
One important observation in the Gromov--Witten theory of toric Calabi--Yau threefolds is that the \textit{closed invariants} (i.e.\ the free energies) $F_h=\omega_{h,0}$ are independent of the framing. This remarkable property suggests that, for any framing (even in the presence of LogTR-vital singularities), one should be able to compute the free energies. An expected definition of the free energies in LogTR should therefore reproduce the same free energies as that obtained from standard TR in a generic framing of the mirror curve.

\subsection{Dilaton equations in LogTR}\label{sec.dilatonomehn}
The derivation of the dilaton equations relies on one fundamental property of the structure of the poles $(\omega_{h,1})_{h\geq 1}$ around the LogTR-vital singularities. Then the derivation essentially follows similar (but slightly extended) lines as the derivation of the dilaton equation in standard TR and is performed by induction. The important behavior of $(\omega_{h,1})_{h\geq 1}$ at logTR-vital singularities is formulated in the following lemma.

\begin{lemma}\label{LemmaIntW02Wg1}Let $s\in \llbracket 1,M\rrbracket$. For any function $F$ that is holomorphic in a neighborhood of $a_s$, we have that for any $h\geq 1$:
\beq \label{lemma1eq}\Res_{z\to a_s} F(z) d_z\Big[(x(z)-x(a_s)) \frac{\omega_{h,1}(z)}{dx(z)}\Big]= (1-2h)\Res_{z\to a_s} F(z)\omega_{h,1}(z)\eeq   
or equivalently
\beq \label{lemma2eq} \Res_{z\to a_s} \frac{\omega_{h,1}(z)}{dx(z)}  d_{z}[F(z)]= 2h\Res_{z\to a_s} F(z)\omega_{h,1}(z).\eeq  
In particular for $F(z)=\int_o^z \omega_{0,2}(.,z_1)$, we obtain (for any distinct points $(z_1,o)$ away from the ramification points and the LogTR-vital singularities)
\beq\label{EqLemmeIntW02WG1} (2h-1)\Res_{z\to a_s} \left(\int_{o}^{z} \omega_{0,2}(., z_1)\right)\omega_{h,1}(z)=\Res_{z\to a_s}\frac{x(z)-x(a_s)}{dx(z)}\omega_{h,1} (z)\omega_{0,2} (z,z_1).\eeq
\begin{proof}
    From the fact that $(x(z)-x(a_s))d_z\log(z-a_s)$ is a holomorphic form at $z\to a_s$, one can conclude for that for any $k\geq1$:
   \footnotesize{\begin{align*}
        &F(z)\bigg(d_z\frac{1}{dx(z)}\bigg)^k\bigg((x(z)-x(a_s))d_z\log(z-a_s)\bigg)\\
        &=(k-1) F(z) \bigg(d_z\frac{1}{dx(z)}\bigg)^{k-1}d_z\log(z-a_s)+F(z)\bigg(d_z\frac{1}{dx(z)}\bigg)\bigg[(x(z)-x(a_s))\bigg(d_z\frac{1}{dx(z)}\bigg)^{k-1 }d_z\log(z-a_s)\bigg]\\
        &=O\left((z-a_s)^0dz\right)
    \end{align*}}
    \normalsize{is} a regular one-form at $z\to a_s$ and thus has a vanishing residue. Taking $k=2h$ and multiplying with $[\hbar^{2h}]\frac{y_{a_s}}{\mathcal{S}(\hbar y_{a_s}^{-1})}$, the meromorphic form $\omega_{h,1}$ can be reconstructed locally around $z\to a_s$ in the integrand. Thus, we get the following identity:
    \begin{align*}
        0=\Res_{z\to a_s} \bigg((2h-1)F(z)\omega_{h,1}(z)+F(z) d_z\Big[(x(z)-x(a_s)) \frac{\omega_{h,1}(z)}{dx(z)}\Big]\bigg),
    \end{align*}
    which is equivalent to \eqref{lemma1eq}. It is then straightforward to see that \eqref{lemma2eq} follows from the action of $d_z$ and integration by parts.
\end{proof}
\end{lemma}

We now have all the tools to prove the dilaton equations for LogTR.

\begin{theorem}[Dilaton equations in LogTR]\label{TheoremDilatonEquation}Let $\Phi(q)=\int^q_o ydx$ be any local antiderivative of $ydx$ in the neighborhoods of the ramification points $(p_i)_{1\leq i\leq N}$. We have:
\beq \label{EqSpecialDilaton}\Res_{z\to z'} \omega_{0,2}(z,z')\Phi(z)=d\Phi(z')=\omega_{0,1}(z')\eeq
and for any $h\geq 0$, $k\geq 1$ such that $(h,k)\neq (0,1)$:
\bea\label{DilatonEquation} (2-2h-k)\omega_{h,k}(z_1,\dots,z_k)&=&\sum_{i=1}^N\Res_{z\to p_i} \Phi(z)\omega_{h,k+1}(z,z_1,\dots,z_k) \cr&&
-\sum_{j=1}^M\Res_{z\to a_j}\frac{x(z)-x(a_j)}{dx(z)}
        \overset{h}{\underset{h_1=1}{\sum}}\omega_{h_1,1} (z)\omega_{h-h_1,k+1} (z,z_1,\dots,z_k)\cr&&
\eea
where $(z_1,\dots, z_k)\in \Sigma^k$ are points away from the ramification points and the LogTR-vital singularities.
\end{theorem}

\begin{proof}The proof is trivial for \eqref{EqSpecialDilaton} from the property of the Bergman kernel on the diagonal. Then, the proof is done by induction on $2h+k$ and is done in \autoref{AppendixDilatonProof}. The idea of the proof is a generalization of the proof by induction in standard TR that includes the new contributions at the LogTR-vital singularities using Lemma  \ref{LemmaIntW02Wg1}.
\end{proof}

The general form of the dilaton equations clearly extends the standard TR dilaton equation \eqref{standarddilaton} in the presence of LogTR-vital singularities. Indeed, the second line in \eqref{DilatonEquation} is new and takes all $(\omega_{h',1})_{1\leq h'\leq h}$ into account.

Note that $\Phi(z)$ is not well defined in a neighborhood of a LogTR-vital singularity because of the presence of a logarithmic cut emanating from $a_s$. This explains why, in the first line of \eqref{DilatonEquation}, no residue at $a_s$ can be taken and rather that the second line can be thought of the corresponding analogue.

\begin{remark}[Alternative formulation of the dilaton equations] For any $h\geq 0$, $k\geq 1$ such that $(h,k)\neq (0,1)$, the dilaton equations may also be rewritten into 
   \bea\label{DilatonLogTRAlt}
        &&(2-2h-k)\omega_{h,k}(z_{\llbracket k \rrbracket})=\sum_{i=1}^N\Res_{z\to p_i} \Phi(z)\omega_{h,k+1}(z,z_{\llbracket k \rrbracket})\cr
        &&-\frac{1}{2}\sum_{s=1}^M\Res_{z\to a_s}\frac{x(z)-x(a_s)}{dx(z)}\Big(\omega_{h-1,k+2}(z,z,z_{\llbracket k \rrbracket})+\underset{\substack{h_1+h_2=h\\ I_1\sqcup I_2={\llbracket k \rrbracket}\\(h_i,I_i)\neq (0,\emptyset)}}{\sum}\omega_{h_1,|I_1|+1} (z,z_{I_1})\omega_{h_2,|I_2|+1} (z,z_{I_2})\Big)\cr&&
    \eea
Indeed, only $\left(\om_{h,1}\right)_{h\geq 1}$ are singular at the LogTR-vital singularities so the previous formula obviously reduces to \eqref{DilatonEquation}.
\end{remark}

\subsection{Definition of the free energies in LogTR}\label{sec.dilatonfreeenergy}
As explained above, the dilaton equations may be used as a guideline to define the free energies. However, it is not possible to simply set $k=0$ in the dilaton equations \eqref{DilatonEquation}, since a term involving $\omega_{0,1}$ would appear in the second line, producing a logarithmic singularity at $a_s$. The solution to this problem comes from local integration by parts. Indeed, the exceptional term $\frac{\omega_{0,1}\,\omega_{h,1}}{dx}$ can be rewritten as $-dy \int \omega_{h,1}$, which avoids the logarithmic singularity at $a_s$ and is well-defined because $(\omega_{h,1})_{h\geq 1}$ have vanishing residues at their poles.

\begin{definition}[Definition of the free energies $(F_h:=\omega_{h,0})_{h\geq 1}$]\label{DefFreeEnergies} We shall define for any $h\geq 2$:
\bea \label{FormulasDefFg}(2-2h)\omega_{h,0}&:=&\sum_{i=1}^N\Res_{z\to p_i} \Phi(z)\omega_{h,1}(z)
\cr&&
-\sum_{s=1}^M\Res_{z\to a_s}\big(x(z)-x(a_s)\big)\left(\frac{1}{2}
        \overset{h-1}{\underset{h_1=1}{\sum}}\frac{\omega_{h_1,1} (z)\omega_{h-h_1,1} (z) }{dx(z)} -dy(z)\int_o^z\omega_{h,1}\right).\cr&&
\eea   
The definition is independent of the basepoint $o$. For $h=1$, we define:
\beq\label{F1formula}  F_1:=\omega_{1,0}:= -\frac{1}{2}\ln \tau_B -\frac{1}{24}\ln\left(\prod_{i=1}^N y'(p_i)\right) -\frac{1}{24}\sum_{s=1}^M\bigg(\frac{y(z)}{y_{a_s} }-\log(x(z)-x(a_s))\bigg)_{\vert_{z=a_s}}\eeq
where $\tau_B$ is the Bergman tau-function and $y'(p_i)=\frac{dy(p_i)}{dz_i(p_i)}$ with $z_i(q)=\sqrt{x(q)-x(p_i)}$.
\end{definition}

\begin{remark}The definition is independent of the basepoint $o$ because $\omega_{h,1}$ is residueless at any ramification point (thus the choice of lower bound in $\Phi$ does not modify the free energy). Moreover, $(x(z)-x(a_s))dy(z)$ is regular at the LogTR-vital singularity $a_s$ (by definition of a logTR-vital singularity) thus the choice of lower bound in the second term of \eqref{FormulasDefFg} does not modify the free energy.
\end{remark}

\begin{remark}The last term in \eqref{FormulasDefFg} involving $dy$ is an exceptional contribution that was necessary in \cite{Banerjee:2025qgx} to derive the correct free energies of so-called strip geometries which are special mirror curves.  It does not appear in \autoref{TheoremDilatonEquation} because it vanishes for $k\geq 1$. Indeed, it would correspond to add the term
\beq \sum_{j=1}^M\Res_{z\to a_j}\big(x(z)-x(a_j)\big)dy(z)\int_o^z\omega_{h,k+1}(.,z_1,\dots,z_k)\eeq
but for $k\geq 1$, $\omega_{h,k+1}$ is regular at the LogTR-vital singularities so that the integrand is regular and thus the residues are trivially vanishing.
\end{remark}

\begin{remark}
 Note that we do not provide a definition of the free energy $F_0:=\omega_{0,0}$. In the original article \cite{EO07}, the definition of $F_0$ comes from the extension of the homogeneity property satisfied by the other free energies and correlation functions. However, this property is not verified in the LogTR setting  because of the presence of LogTR-vital singularities and thus no formula can be used to extend the definition to $(h,n)=(0,0)$. To the best of our understanding, the definition of $F_0$ should depend on the actual geometric origin of the spectral curve whether it comes from toric geometry or Seiberg--Witten theory.
\end{remark}

\subsection{Examples}\label{sec.dilatonfreeenergyexamples}
\begin{example}
    Take the following genus $0$ spectral curve:
    \beq
        x(z)=z,\qquad y(z)=\Lambda+\sum_{s=1}^M y_{a_s} \log(z-a_s).
    \eeq
where $z\in \mathbb{P}^1$ is a global coordinate and such that $y_\infty+\underset{s=1}{\overset{M}{\sum}} y_{a_s}=0$, where $y_\infty$ is the residue of $dy$ at infinity. This curve has no ramification points so that all $(\omega_{h,n})_{h\geq 0,n\geq 2}$ with $(h,n)\neq (0,2)$  vanish. Moreover, $(\omega_{h,1})_{h\geq 1}$ are given explicitly by 
\bea \omega_{h,1}(z)&=&-dx(z)[\hbar^{2h}]\sum_{s=1}^M\left(\frac{y_{a_s}}{\mathcal{S}(y_{a_s}^{-1}\hbar \partial_x)}\ln(z-a_s) \right)=-[\hbar^{2h}]\frac{1}{\mathcal{S}(\hbar)}\sum_{s=1}^M y_{a_s}^{1-2h}\frac{(2h-1)!}{(z-a_s)^{2h}}dz\cr
&=&-\sum_{s=1}^M y_{a_s}^{1-2h}\frac{(2^{1-2h}-1)B_{2h}}{2h}\frac{dz}{(z-a_s)^{2h}}.
\eea
Note that there is no contribution from $z=\infty$, since $x$ is singular at infinity ensuring that infinity is not a LogTR-vital singularity.

Inserting this into \eqref{FormulasDefFg} and evaluating the residue, it is easy to derive
    \bea\label{F2example}
       &&\forall\, h\geq 2\,:\,  F_{h}=\frac{1}{2(2h-2)  }\sum_{s=1}^M\sum_{r\neq s}\frac{ (2h-2)!}{(a_s-a_r)^{2h-2}}[\hbar^{2h}]\frac{y_{a_s}y_{a_r}}{\mathcal{S}(y_{a_s}^{-1}\hbar)\mathcal{S}(y_{a_r}^{-1}\hbar)}
   \eea
    and 
\beq F_1=-\frac{1}{24}\sum_{s=1}^M\sum_{r\neq s}\frac{y_{a_r}}{y_{a_s}}\log(a_r-a_s)-\frac{\Lambda}{24}\sum_{s=1}^M\frac{1}{y_{a_s}}.
\eeq
For the special values $y_{a_s}=1$ for all $s\in \llbracket1,M\rrbracket$, the free energies \eqref{F2example} simplify, using $
\frac{1}{\mathcal{S}(t)^2}
= 1 - \underset{h=1}{\overset{\infty}{\sum}} \frac{B_{2h} t^{2h}}{2h (2h-2)!}$,
 to the perturbative part of the $4d$ $\mathcal{N}=2$ pure supersymmetric gauge theory \cite{Nekrasov:2003rj}, computed from the $x$-$y$ dual of the so-called half Seiberg--Witten spectral curve \cite{Borot:2021btb,Borot:2024uos,Hock:2025wlm}.

\medskip

Even more interesting is the generic case with $y_{a_r} \neq y_{a_s}$. The expansion of$
\frac{y_{a_s} y_{a_r}}{\mathcal{S}(y_{a_s}^{-1}\hbar)\mathcal{S}(y_{a_r}^{-1}\hbar)}$
in \eqref{F2example} has coefficients given by so-called double Bernoulli numbers, which appear in various refined settings. For instance, they arise in a quantized Riemann--Hilbert problem \cite{Barbieri:2019yya} and also in refined topological recursion \cite{Kidwai:2023fxs}. Note, however, that allowing generic $y_{a_s}$ should maybe be understood as working with a refined spectral curve on which (Log)TR is performed, rather than refined topological recursion \cite{Osuga:2023kgw} applied to the original spectral curve. This interplay will be further investigated in the near future.
\end{example}

\begin{example}
    Consider the following genus 0 spectral curve 
    \beq
        x(z)=\log z,\qquad y(z)=\sum_{s=1}^M y_{a_s}\log\left(1- \frac{z}{a_s}\right),
    \eeq
    with $y_\infty+\underset{s=1}{\overset{M}{\sum}} y_{a_s}=0$, where $y_\infty$ is the residue of $dy$ at infinity. This curve can be understood via the parametrization $X=e^x$ and $Y=e^y$ as a curve in $\mathbb{C}^*\times \mathbb{C}^*$.

    This example has no ramification points, so that all $(\omega_{h,n})_{h\geq 0,n\geq 2}$ with $(h,n)\neq (0,2)$ vanish. However, $(\omega_{h,1})_{h\geq 1}$ are given explicitly 
\bea \omega_{h,1}(z)&=&-\sum_{s=1}^M\Res_{z\to a_s}\left(\int_{a_s}^z\omega_{0,2}(z_1,.)\right) dx(z)[\hbar^{2h}]\left(\frac{y_{a_s}}{\mathcal{S}(y_{a_s}^{-1}\hbar \partial_x)}\log\left(1- \frac{z}{a_s}\right) \right)\cr
&=&dx(z)[\hbar^{2h}]\sum_{s=1}^M \frac{y_{a_s}}{\mathcal{S}(y_{a_s}^{-1}\hbar )}\mathrm{Li}_{1-2h}\left( \frac{z}{a_s}\right)\cr
&=&dx(z)\sum_{s=1}^M y_{a_s}^{1-2h}\frac{(2^{1-2h}-1)B_{2h}}{(2h)!}\mathrm{Li}_{1-2h}\left( \frac{z}{a_s}\right)
\eea
    Note that there is no contribution from infinity, since $x$ is singular at infinity, ensuring that infinity is not a LogTR-vital singularity.
    
    Inserting this into \eqref{FormulasDefFg}, one can follow one-to-one the computation of \cite[Appendix A]{Banerjee:2025qgx} and obtain, after evaluating the residues,
    \small{\bea\label{F22example}
       &&\forall\, h\geq 2\,:\,  F_{h}=-\frac{B_{2h-2}}{2 (2h-2)}[\hbar^{2h}]\sum_{r=1}^M\left(\frac{y_{a_r}}{\mathcal{S}(y^{-1}_{a_r}\hbar )}\right)^2+\frac{1}{2}\sum_{\substack{r,s=1\\ r\neq s}}^M\mathrm{Li}_{3-2h}\left(\frac{a_s}{a_r}\right)[\hbar^{2h}]\frac{y_{a_r}y_{a_s}}{\mathcal{S}(y^{-1}_{a_r}\hbar )\mathcal{S}(y^{-1}_{a_s}\hbar )}.\cr&&
   \eea}
    \normalsize{and} 
\beq 
F_1=-\frac{1}{24}\sum_{s=1}^M\sum_{r\neq s}\frac{y_{a_r}}{y_{a_s}}\log\left(1-\frac{a_s}{a_r}\right)=\frac{1}{24}\sum_{s=1}^M\sum_{r\neq s}\frac{y_{a_r}}{y_{a_s}}\mathrm{Li}_1\left(\frac{a_r}{a_s}\right)
\eeq
    up to a constant for $F_1$ which is independent of $a_s$, arising from the limit $\log\left(1-\frac{z}{a_s}\right)-\log(x(z)-x(a_s))\vert_{z=a_s}=\text{const}.$

    In the physics literature, it is more natural to express the full asymptotic series in $\hbar$ by expanding the polylogarithm and formally interchanging the series in $\hbar$ with the series expansion of the polylogarithm. For this, different $a_r$ must be compared to ensure convergence of the polylogarithms. This is a common, but non-rigorous, computation. Let $|a_r|<|a_{r+1}|$ for all $r$, then this would lead to
\beq
\sum_{h=0}^\infty \hbar^{2h-2}F_h=\sum_{r<s}\sum_{n=1}^\infty\frac{\left(\frac{a_r}{a_s}\right)^n}{n\left(e^{\frac{n\hbar}{2y_{a_r}}}-e^{-\frac{n\hbar}{2y_{a_r}}}\right)\left(e^{\frac{n\hbar}{2y_{a_s}}}-e^{-\frac{n\hbar}{2y_{a_s}}}\right)}+\sum_{s=1}^M\sum_{k=1}^{\infty} k \log\left(1-e^{\frac{\hbar k}{y_{a_s}}}\right)
\eeq
    The reason for expressing the free energies in this form is to capture the effect of $y_{a_r}$. The usually quadratic denominator in the first line is split into two factors, as observed, for instance, in refined topological string theory. However, the McMahon-type function (last term) is not refined as the refined topological vertex \cite{Iqbal:2007ii}. The last term can be included in the first by adding $r=s$ in the sum. 
    
    For $y_{a_s}=1$, the curve is the $x$-$y$ dual of the mirror curve of strip geometries of toric Calabi--Yau threefolds \cite{Iqbal:2004ne,BKMP}, and our result recovers the known free energies, corresponding to closed string amplitudes or Gromov--Witten invariants \cite{Iqbal:2004ne}.

    In this second example, we observe again that the introduction of $y_{a_r}$ into the spectral curve has the structure of refining the spectral curve, but still using (Log)TR instead of applying a refined version of TR to the original curve. This is consistent with \cite{Eynard:2011vs}.
\end{example}

As both examples with coefficients $y_{a_r}=1$ recover known free energies associated with the corresponding $x$--$y$ dual spectral curves in the literature, we expect that this holds in general:

\begin{conjecture}
For a given genus $g=0$ admissible spectral curve on $\mathbb{P}^1$, with $dx$ and $dy$ having at most simple poles, the free energy defined in \autoref{DefFreeEnergies} is invariant under $x$--$y$ duality.
\end{conjecture}

\section{Parametrization of the spectral curve by a decomposition of~$ydx$}\label{sec.parametrizeom01}

In this section, we shall describe a parametrization of the spectral curve by choosing a decomposition of the one-form $ydx$. This decomposition is not unique and follows from the singularity structure associated to $dx$ and $dy$ upon choosing a base point, logarithmic cuts and a fundamental domain. This will provide the set of parameters for which the variations of the correlators obtained from LogTR can be computed.

\begin{definition}[Residues of $dy$ and logarithmic part of $y$]\label{Proptdy}For any $a \in \mathcal{S}_y=\{a_1,\dots,a_M\}$, we have 
\beq \label{Vanishsumya} y_a=\underset{q\to a}{\Res}dy \in \mathbb{C}^* \,\,\text{with} \,\, \underset{a\in \mathcal{S}_y}{\sum} y_a=0.\eeq
The meromorphic one-form $\td{w}$:
\beq \td{\omega}:=dy-\sum_{a\in \mathcal{S}_y} y_a dS_{a,o}(q)  \eeq
defines a residue-free meromorphic one-form upon a choice of $\mathcal{A}$- and $\mathcal{B}$-cycles and a base point $o$.
From $\td{\omega}$, one can define by integration $\td{y}$
\beq \label{Vanishsumya2} \td{y}(q):=\int_o^q\td{\omega} = y(q) -\sum_{a\in \mathcal{S}_y} y_a  \ln \frac{E(q,a)}{E(q,o)}\eeq
upon a choice of $\mathcal{A}$- and $\mathcal{B}$-cycles, and logarithmic cuts in the fundamental domain emerging at the base point $o$.
\end{definition}

\begin{remark} Note that the pole structure of  $\td{\omega}$ is independent of the basepoint $o$ since $\underset{a\in \mathcal{S}_y}{\sum} y_a=0$. Moreover note that, by definition, the function $\underset{a\in \mathcal{S}_y}{\sum} y_a \ln \frac{E(q,a)}{E(q,o)}$ is defined on $\Sigma\setminus \left(\underset{a\in \mathcal{S}_y}{\bigcup}\mathcal{C}_{o\to a}\right)$, i.e. the Riemann surface $\Sigma$ minus a set of cuts connecting $o$ to $a$ which depends on where the cuts are placed, and  everything furthermore depending on a choice of $\mathcal{A}$- and $\mathcal{B}$-cycles (Torelli marking). The choice of cuts is implicitly encoded in the definition of the logarithm function and consists of contractible paths connecting $o$ to $a$ that do not intersect (or self-intersect) except at $o$, and that avoids all poles of $dx$ and $dy$. 
\end{remark}

\subsection{Local expansion of $ydx$}\label{Sec:localcorrdinate}
In this section, we shall define local coordinates around ramification and singular points of $ydx$. These local coordinates are the $x$-projected coordinates which will later provide canonical parameters that will be used in \autoref{SectionGlobalDec} for a parametrized decomposition of $ydx$.

\subsubsection*{Local coordinates around ramification points}
Let us first define local coordinates around a ramification point. We recall that from \autoref{MainAssumption}, the ramification points are away from the singularities of $dx$ and $dy$ and away from the logarithmic cuts. To be more specific, there exists an open neighborhood around the ramification points that do not intersect with $\mathcal{P}_x\cup \mathcal{P}_y\cup \underset{s\in \mathcal{S}_x}{\bigcup} \mathcal{C}_{o\to s}\cup \underset{s\in \mathcal{S}_y}{\bigcup} \mathcal{C}_{o\to s}$. In particular, when writing residues at ramification points, it implicitly means that the circle integral is taken within this neighborhood.

\begin{definition}[Local coordinates around a ramification point]\label{DefLocalCoordinateBranchPoint} Let $p\in \text{Ram}$ be a (simple from \autoref{MainAssumption}) ramification point. We shall define the local coordinate around $p$ as
\beq z_{p}(q):=\sqrt{x(q)-x(p)} \,\,\Rightarrow\,\, dx(q)=2z_{p}(q) dz_{p}(q)\eeq
\end{definition}

Note that $z_{p}(p)=0$ and $ydx$ is holomorphic in $z_p$ in a neighborhood of $p$.

\subsubsection*{Local coordinates around a singular point of $ydx$}

We shall now focus on local coordinates around a singular point $a$ of $ydx$. There are several cases that are described in the following definition.

\begin{definition}[Local coordinates at singularities of $ydx$]\label{DefLocalCoord}For any singularity $a$ of $ydx$, we shall define the local coordinates $z_a$ and local degree $d_a\in \mathbb{Z}$:
\begin{itemize}
    \item If $a$ is a pole of $dy$ but $x$ is regular at $a$ we shall define 
    \beq z_a(q):=
    x(q)-x(a)\eeq
    \item If $a$ is a pole of $dx$ of order $d_a+1\geq 2$ with possibly non-vanishing residue (i.e. $x(a)=\infty$ such that $x(q)\overset{q\to a}{=}\frac{\mu_a}{(q-a)^{d_a}} +o\left((q-a)^{-d_a}\right)$ for some $\mu_a\neq 0$) we shall define:
    \beq z_a(q):=
    \left(\frac{1}{x(q)}\right)^{\frac{1}{d_a}}
    \eeq
    where the $d_a$-branch is defined locally around $\infty=x(a)\in \mathbb{P}^1$. 
    \item If $a$ is a simple pole of $dx$ with non-vanishing residue $x_a$ we define $d_a=0$ and
    \beq z_a(q):=
    \exp\left(\frac{x(q)}{x_{a}}\right)\eeq
\end{itemize}
In all cases, the local coordinates $z_a$ define a map from a neighborhood of $a$ into $\mathbb{P}^1$. 
\end{definition}

\begin{remark}Note that from \autoref{MainAssumption}, \autoref{DefLocalCoordinateBranchPoint} and \autoref{DefLocalCoord} are consistent (because all cases are disjoint) and exhaust all possible types of singularities of $ydx$ or ramification points. 
\end{remark}

Note that in all cases we have $z_a(a)=0$. 

\medskip

The previous local coordinates at singular points of $ydx$ allow to define the local expansion of $ydx$.

\begin{proposition}[Local decomposition of $ydx$ and local potential $V_a$]\label{PropLocalCoord} Let $a\in \mathcal{P}_x\cup\mathcal{P}_y$ be a singular point of $ydx$ and denote $y_a:=\underset{q\to a}{\Res} dy \in \mathbb{C}$. Take $z_a$ the local coordinate at $a$ defined in \autoref{DefLocalCoord} and the associated value of $d_a\in \mathbb{Z}$. Then, the one-form $ydx$ can be locally written as:
\bea ydx(q)&\overset{q\to a}{=}&
 y_a \ln z_a(q) dx(q)-\sum_{k=1}^{R_a} k\td{t}_{a,k}z_a(q)^{-k-1}dz_a(q)
+\td{t}_{a,0}\frac{dz_a(q)}{z_a(q)}+ d\td{v}_a(q) 
\eea
where $d\td{v}_a$ is a local holomorphic one-form at $a$ and $R_a\in \mathbb{N}$ represents the order of singularity of $ydx$ at $a$. In other words, we have:
\beq \td{y}dx(q)\overset{q\to a}{=}-\sum_{k=1}^{R_a} k\td{t}_{a,k}z_a(q)^{-k-1}dz_a(q)
+\td{t}_{a,0}\frac{dz_a(q)}{z_a(q)}+ d\td{v}_a(q)\eeq
\end{proposition}

\begin{proof}The proof is similar to the one in \cite{EO07} after allowing logarithmic singularities of $y$. By Definition \ref{Proptdy}, one can subtract the logarithmic contribution 
\beq \td{y}dx(q)= ydx(q)-\sum_{b\in \mathcal{P}}y_b \ln \frac{E(q,b)}{E(q,o)} dx(q)\eeq
where $\td{y}$ is meromorphic in a neighborhood of $a$ from \autoref{Proptdy}. Note that if $a$ is both a pole of $dy$ with non-vanishing residue and a pole of $dx$ with nonvanishing residue, then $\log z_a(q)dx(q)\sim \log z_a(q)\frac{dz_a(q)}{z_a(q)}\sim x(q)dx(q).$
\end{proof}

\begin{definition}[Irregular times and monodromies]\label{DefTerminologyParameters}Similarly to the terminology of the standard topological recursion \cite{EO07}, we shall call the coefficients $\left(\td{t}_{a,k}\right)_{1\leq k\leq R_a}$ the \textit{irregular times} of $\td{y}dx$ at the singular point $a$, while $\td{t}_{a,0}$ is called the \textit{monodromy} of $\td{y}dx$ at $a$. Note that the sum of monodromies is vanishing:
\beq \label{VanishingMono} \underset{\alpha\in \mathcal{P}}{\sum}\td{t}_{\alpha,0}=0\eeq
since $\td{y}dx$ is a meromorphic one-form on $\Sigma$. The new coefficient $y_a$ appearing when $a\in \mathcal{S}_y$ shall be referred to as the \textit{log-time} at the singular point $a$.
\end{definition}

Similarly to the standard topological recursion, it is useful to observe that the irregular times and monodromies $\left(\td{t}_{a,k}\right)_{0\leq k\leq R_a}$ can be obtained as local residues at $a$.

\begin{proposition}[Irregular times and monodromies as local residues]\label{PropCoeffIntegralsSimple} For any singular point $a$ of $ydx$, we have:
\beq \forall \, k\in \llbracket 0,R_a\rrbracket\,:\, \Res_{q\to a}z_a(q)^k
\td{y}dx(q)= -k\,\td{t}_{a,k}
+ \td{t}_{a,0}\delta_{k=0}\eeq
where local coordinates $z_a$ and the coefficients $\left(\td{t}_{a,k}\right)_{0\leq k\leq R_a}$ are defined from \autoref{DefLocalCoord}.
\end{proposition}

\begin{proof}The proof follows from the local decomposition of $\td{y}$ at $a$. For details, see \cite{EO07}.
\end{proof}

\begin{remark}Note that if $a$ is a pole of $dy$ with a non-vanishing residue $y_a\in \mathbb{C}^*$, then residues like $-\frac{1}{k}\underset{q\to a}{\Res} z_a(q)^k ydx$ are not well-defined since $ydx$ has a logarithmic cut at $a$ so a local circle integral is not well-defined. To avoid the confusion, we added a tilde on the irregular times and monodromies. Moreover note that
$-\frac{1}{k}\underset{q\to a}{\Res}z_a(q)^k\td{y}(q)dx(q)$ does not provide $\td{t}_{a,k}$ if $a$ is at the same time a pole of $dy$ with a non-vanishing residue $y_a\in \mathbb{C}^*$ and a pole of $dx$. 
\end{remark}


\subsection{Global decomposition of $ydx$}\label{SectionGlobalDec}
Using the local decomposition of $ydx$ at its singularities, we shall propose a global decomposition of $ydx$ that will provide the canonical parameters for variations. Following \cite{EO07}, we shall first introduce the following quantities.

\begin{definition}[Definition of $(B_{a,k})_{k\geq 1}$ and $B_{a,0,o'}$]\label{DefBalphak}Let $a\in \mathcal{P}$ be a singular point of $ydx$ with local behavior given by \autoref{PropLocalCoord}. We define the following one-forms on $\Sigma$:
\beq \forall\,k\geq 1 \,:\, B_{a,k}(q):=\Res_{s\to a} B(q,s)z_a(s)^{-k}\,\,\text{ and }\, 
B_{a,0;o'}(q):=\int_{o'}^a B(.,q):=dS_{a,o'}(q)
\eeq
where $o'\in \Sigma$ is a fixed generic basepoint and the path integral is a contractible smooth path connecting $o' \to a$ avoiding all special points (ramification points, poles of $dx$ or $dy$) or special contour (logarithmic cuts).
\end{definition}

The previous quantities are meromorphic one-forms on $\Sigma$ and normalized on the $\mathcal{A}$-cycles (from the normalization of the Bergman kernel). Moreover, they have singularities that are similar to those of $ydx$ at $a$. More precisely:
\begin{itemize}
    \item For $k\geq 1$, $B_{a,k}$ is a meromorphic one-form on $\Sigma$ with only one pole at $a$ (with vanishing residue) and that locally behaves like
    \beq \forall\,k\geq 1 \,:\, B_{a,k}(q)\overset{q\to a}{=} -kz_a(q)^{-k-1}dz_a(q) + O(1)\eeq
    \item $B_{\alpha,0;o'}$ is a meromorphic one-form on $\Sigma$ with only two simple poles at $a$ and $o'$ with opposite residues. Locally around $a$ we have:
    \beq B_{a,0;o'}(q)\overset{q\to a}{=}\frac{d z_a(q)}{z_a(q)}+O(1)\eeq
\end{itemize}

We now have all the ingredients to write a global decomposition of $ydx$.

\begin{theorem}[Global decomposition of $ydx$]\label{TheoremGlobalDecompositionydx}Let $\left(\td{t}_{a,k}\right)_{a\in \mathcal{P},0\leq k\leq R_a}$ be the irregular times and monodromies defined from \autoref{PropCoeffIntegralsSimple}. Fixing a choice of logarithmic cuts for $ydx$, one can write a decomposition upon this choice as:
\bea ydx(q)&=&\sum_{a\in\mathcal{S}_y} y_a \ln \frac{E(q,a)}{E(q,o)} dx(q)+ \sum_{a\in \mathcal{P}}\left(\sum_{k=1}^{R_a}\td{t}_{a,k} B_{a,k}(q) +\td{t}_{a,0}B_{a,0,o'}(q)\right) +\sum_{i=1}^g \td{\epsilon}_i du_i(q)\cr
\text{i.e. }\td{y}dx(q)&=&\sum_{a\in \mathcal{P}}\left(\sum_{k=1}^{R_a}\td{t}_{a,k} B_{a,k}(q) +\td{t}_{a,0}B_{a,0,o'}(q)\right) +\sum_{i=1}^g \td{\epsilon}_i du_i(q)
\eea
where the parameters $(\td{\epsilon}_i)_{1\leq i\leq g}$ are defined by
$\forall\, i\in \llbracket 1,g\rrbracket\,:\, \td{\epsilon}_i:=\oint_{\mathcal{A}_i} \td{y}dx$.
\end{theorem}

\begin{proof}
Let us first mention that the decomposition depends very explicitly on the choice of $\mathcal{A}$- and $\mathcal{B}$-cycles and by fixing the logarithmic cuts of $y$. From Definition~\ref{Proptdy}, the one form $\td{y}dx$ has a unique decomposition upon those choices. The global definition and normalization of $B_{a,k}$ and $B_{a,0,o'}$ provides the usual decomposition of a one-form on a compact Riemann surface, where $du_i$ are the canonical holomorphic one-forms. This provides a parametrization of $\td{y}dx$ in terms of the times $\td{t}_{a,k}, \td{\epsilon}_i$. 

The additional new term for $ydx$ takes care of all logarithmic cuts after enforcing a choice on $y$ by integration from $dy$. As the decomposition of a meromorphic one-form depends on the $\mathcal{A}$- and $\mathcal{B}$-cycles, the decomposition $ydx$ depends on the choice of how the logarithmic cuts are placed.

\end{proof}

\begin{remark}Note that the point $o'$ in the definition of $B_{\alpha,0,o'}$ is irrelevant for the global decomposition of $ydx$ since the sum of monodromies is vanishing from \eqref{VanishingMono}. Similarly, the point $o$ appearing in the prime form is irrelevant in the previous decomposition from \eqref{Vanishsumya2}.
\end{remark}


\subsection{Parametrization of the spectral curve}\label{sec.parametrizespectralcurve}
\autoref{TheoremGlobalDecompositionydx} provides the set of parameters that we can use to characterize the spectral curve. 

\begin{corollary}[Parametrization of the spectral curve]\label{RemarkParametrization} The one-form  $ydx$ is parametrized by  
\begin{itemize}
    \item The positions $\left(a\right)_{a\in \mathcal{P}}$ of the singularities of $ydx$.
    \item The irregular times $\left(\td{t}_{a,k}\right)_{a\in \mathcal{P},1\leq k \leq R_a}$ of $\td{y}dx$.
    \item The monodromies $\left(\td{t}_{a,0}\right)_{a\in \mathcal{P}}$ of $\td{y}dx$.
    \item The ``filling fractions" $(\td{\epsilon}_i)_{1\leq i\leq g}$ of $\td{y}dx$.
    \item The log-times $\left(y_a\right)_{a\in \mathcal{S}_y}$ that correspond to the residues of $dy$ at its poles.
\end{itemize}
\end{corollary}

In the rest of the article, we will \textit{not} consider variations with respect to log-times. Moreover, in standard topological recursion, variations with respect to the position of poles of $ydx$ are not independent of the variations with respect to irregular times and monodromies and have therefore been omitted in \cite{EO07}. Equivalently in integrable systems, it is standard to only consider isomonodromic deformations, i.e. not to vary the monodromies but rather consider variations of the location of the poles instead. However, in the LogTR setting, it turns out that it is very natural to consider variations with respect to LogTR-vital singularities. This will provide new features that are specific to LogTR.

\subsection{Deformations and rewriting using Bergman kernel's integrals}\label{SubsectionLambda}
For any irregular time, monodromy and filling fraction of $\td{y}dx$ associate an integral of the Bergman kernel analogously to \cite{EO07}, see also more recently \cite{Eynard:2023dha}. This formalism will be particularly convenient when dealing with variations since we will be able to regroup these cases under one notation. More precisely, to any of these parameters, denoted generically $t$,  we shall associate a one-form $\Omega_t$, a contour $\partial_{\Omega_t}$ and a function $\Lambda_t$ on $\Sigma$ such that
\beq \partial_t [\td{y}dx(q)]:=\Omega_t(q):=\int_{\partial_{\Omega_t}} B(q,s) \Lambda_t(s).\eeq
More precisely, we have:
\begin{itemize}
    \item For an irregular time $\left(\td{t}_{a,k}\right)_{a\in \mathcal{P}, k\geq 1}$, we have
    \beq \Omega_{\td{t}_{a,k}}(q):=B_{a,k}(q)=\Res_{s\to a} B(q,s)z_a(s)^{-k}=\oint_{\partial_{\Omega_{\td{t}_{a,k}}}}B(q,s)\Lambda_{\td{t}_{a,k}}(s)\eeq
    with $\Lambda_{\td{t}_{a,k}}(s):=\frac{1}{2i\pi} z_a(s)^{-k}$ and $\partial_{\Omega_{\td{t}_{a,k}}}:= \mathcal{C}_a$ a local loop around $a$.
    \item Deformations with respect to monodromies are constrained by the conditions $\underset{a \in \mathcal{P}}{\sum} \td{t}_{a,0}=0$. Thus, we may only consider $\partial_{\td{t}_{a,0}}-\partial_{\td{t}_{b,0}} $ with $a\neq b$ as deformations. In this case, we have
    \beq \Omega_{\td{t}_{a,0} -\td{t}_{b,0} }(q):= B_{a,0;o'}(q)-B_{b,0;o'}(q)=\int_{b}^{a} B(q,s)=\oint_{\partial_{\Omega_{\td{t}_{a,0}}-\td{t}_{b,0}}}B(q,s)\Lambda_{\td{t}_{a,0} -\td{t}_{b,0}}(s)\eeq
    with $\Lambda_{\td{t}_{a,0}-\td{t}_{b,0}}(s):=1$ and $\partial_{\Omega_{\td{t}_{a,0}-\td{t}_{b,0}}}:=[b,a]$ (or any oriented contractible smooth path connecting $b$ to $a$).    
    \item For a filling fraction $(\td{\epsilon}_i)_{1\leq i\leq g}$ we have 
    \beq \Omega_{\td{\epsilon}_{i}}(q):= du_i(q)=\frac{1}{2i\pi}\oint_{\mathcal{B}_i} B(s,q)=\oint_{\partial_{\Omega_{\td{\epsilon}_{i}}}}B(q,s)\Lambda_{\td{\epsilon}_{i}}(s)\eeq
    so $\Lambda_{\td{\epsilon}_{i}}(s):=\frac{1}{2i\pi}$ and $\partial_{\Omega_{\td{\epsilon}_{i}}}:=\mathcal{B}_i$.
\end{itemize}

\section{Variational formulas}\label{sec.variationalformula}
Variational formulas play an important role in the theory of standard TR. Instead of looking at a specific spectral curve, one can consider an infinitesimal deformation of the spectral curve and derive the corresponding variation of all correlators $\omega_{h,n}^{\text{EO}}$ in standard TR. The original variational formulas in \cite{EO07} are derived for variations corresponding to the parameters of $\td{y}dx$, consisting of irregular times, monodromies, and filling fractions. In the standard setting, the one-form $\omega_{0,1}^{\text{EO}}$ is assumed to be meromorphic, which implies that there is a unique expansion once a choice of $\mathcal{A}$- and $\mathcal{B}$-cycles on the compact Riemann surface is fixed. Roughly speaking, the variational formula in the standard setting has the form
\begin{align*}
    \delta_\Omega[\omega^{\text{EO}}_{h,m}(z_1,\dots,z_m)]
    =
    \int_{\partial_\Omega} \Lambda_\Omega(s)\, \omega^{\text{EO}}_{h,m+1}(z_1,\dots,z_m,s). 
\end{align*}

There is a striking compatibility between the dilaton equations and the variational formulas, even though the latter are derived from the recursive definition of standard TR. Taking the dilaton equations in the standard setting \eqref{standarddilaton}, 
\begin{equation}\label{standarddilaton2}
    (2-2h-n)\omega^{\text{EO}}_{h,n}(z_1,...,z_n)
    =
    \sum_{p\in \text{Ram}}
    \Res_{z'\to p}
    \Phi(z')\omega^{\text{EO}}_{h,n+1}(z',z_1,...,z_n),
\end{equation}
one can act on both sides with the variation $\delta_\Omega$. The variation commutes with the residue and acts by the Leibniz rule on $\Phi(z')=\int^{z'}\omega_{0,1}^{\text{EO}}$ and $\omega^{\text{EO}}_{h,n+1}(z',z_1,...,z_n)$ separately. The action on $\Phi(z')$ extracts a specific part of the decomposition of $\omega_{0,1}^{\text{EO}}$, whereas the action on $\omega^{\text{EO}}_{h,n+1}(z',z_1,...,z_n)$ produces $\delta_\Omega[(1-2h-n)\omega^{\text{EO}}_{h,n}]$. 

In summary, the action of $\delta_\Omega$ on the dilaton equation reduces to 
\begin{align*}
    \delta_\Omega[\omega^{\text{EO}}_{h,n}(z_1,...,z_n)]
    =
    \sum_{p\in \text{Ram}}
    \Res_{z'\to p}
    \delta_\Omega[\Phi(z')]\,
    \omega^{\text{EO}}_{h,n+1}(z',z_1,...,z_n),
\end{align*}
where the factor $(2-2h-n)$ on the left disappears due to the variation of $\omega^{\text{EO}}_{h,n+1}$ on the right, and only the variation of $\Phi(z')$ remains. Performing this variation and deforming the contour on a compact Riemann surface yields precisely the variational formula \eqref{standarddilaton2}. 

We mention this observation—that the variation of $\omega^{EO}_{h,n+1}$ in the dilaton equation cancels the factor $(2-2h-n)$, leaving only the variation of $\Phi(z')$—because the same structure persists in LogTR, not only for variations with respect to the classical times (irregular, monodromy, and filling fraction), but also for variations of the position of the LogTR-vital singularity.

\subsection{Variations in LogTR}\label{sec.variationsof}
In the previous section, we parametrized the spectral curve $ydx$ using several parameters described in \autoref{RemarkParametrization}. The next step is now to compute the variation of the correlators produced by LogTR with respect to these parameters.

In this paper, we shall \textbf{only consider variations at fixed map} $\mathbf{x}$ (this implies, for instance, that the ramification points are not modified by the variation). Thus, we have by definition
\beq\label{NoXvariation} \delta_\Omega[x]=0\eeq
As mentioned in \autoref{RemarkParametrization}, there are different types of parameters: location of the poles, irregular times, monodromies and filling fractions associated to $\td{y}dx$ and log-times. Note that since we only consider variations at fixed $x$, deformations with respect to the location of poles reduce only to deformations with respect to the location of poles of $dy$ where $x$ is regular. In the end, variations at fixed $x$ are provided by

\begin{itemize}
    \item Variations with respect to the location of poles of $dy$ such that $x$ is regular and $dy$ has vanishing residues.
    \item Variations with respect to the irregular times $(\td{t}_{\alpha,k})_{\alpha\in \mathcal{P},k\geq 1}$ associated to $\td{y}dx$.
     \item Variations with respect to the monodromies $(\td{t}_{\alpha,0})_{\alpha\in \mathcal{P}}$ associated to $\td{y}dx$.
     \item  Variations with respect to the filling fractions $(\td{\epsilon}_i)_{1\leq i\leq g}$ of $\td{y}dx$.
     \item Variations with respect to the location of poles of $dy$ with non-vanishing residues where $x$ is regular.
     \item Variations with respect to the log-times $y_\alpha$.
\end{itemize}

\subsection{Variational formulas with respect to parameters of $\td{y}dx$ for correlators}\label{SectionVartdydx}
In this section, we will only consider deformations generated by the first four types of parameters, i.e. location of poles of $dy$ such that $x$ is regular and $dy$ has vanishing residues, irregular times and monodromies and filling fractions associated to $\td{y}dx$, keeping the locations of poles of $dy$ with non-vanishing residues where $x$ is regular and the log-times fixed. The case of location of poles where $dy$ has residue will be done in \autoref{SectionVarVitalSing}.  

In fact, as in standard TR, these four types of deformations are not all independent. Indeed, from the global decomposition of $ydx$ given by \autoref{TheoremGlobalDecompositionydx}, variations with respect to locations of poles of $dy$ such that $x$ is regular and $dy$ has vanishing residues are given by a linear combination of the other three variations considered in this section. More precisely:
\beq \forall\, a\in \mathcal{P}_y \,/\, \Res_{q\to a}dy=0 \text{ and } x(a)\in \mathbb{C}:  \Omega_a=-x'(a)\sum_{j=2}^{R_a+1}(j-1)\td{t}_{a,j-1}\Omega_{\td{t}_{a,j}} +\td{t}_{a,0}B(a,q)  \eeq
Thus, as in standard TR \cite{EO07}, we will only consider variations with respect to irregular times, monodromies and filling fractions.

\smallskip

To simplify notations, we shall denote generically $\delta_\Omega$ the infinitesimal deformation that we consider and the deformations studied in this section are generated by linear combinations (as elements of the tangent space) of the following fundamental infinitesimal deformations:
\begin{itemize}
    \item $\delta_{\Omega}=\partial_{\td{t}_{\alpha,k}}$ for $k\in \llbracket 1,R_\alpha\rrbracket$, i.e. deformations with respect to an irregular time of $\td{y}dx$.
    \item $\delta_\Omega=\partial_{\td{t}_{\alpha_i,0}}- \partial_{\td{t}_{\alpha_j,0}}$ for $i\neq j$, i.e. deformations with respect to monodromies of $\td{y}dx$. Note that we must keep the relation $\underset{\alpha\in \mathcal{P}}{\sum}\td{t}_{\alpha,0}=0$ while performing deformations so that $\partial_{\td{t}_{\alpha,0}}$ is not possible.
    \item $\delta_\Omega=\partial_{\td{\epsilon}_i}$ for $i\in \llbracket 1,g\rrbracket$, i.e. deformations with respect to the filling fractions of $\td{y}dx$.
\end{itemize}
In other words, \textbf{we shall only consider variations with inside the Hurwitz space associated to $\td{y}dx$}. The important observation which we are going to prove is that the original variational formulas corresponding to the meromorphic one-form $\td{y}dx$ which hold in the case of standard TR perfectly extend to the LogTR correlators. Thus, LogTR is compatible with the standard variational formula.

In other words, we shall prove the following variational formulas:

\begin{theorem}[Variational formulas of the correlators for standard times]\label{TheoVariationalFormulas} For any variation $\delta_\Omega$ with respect to $(\td{t}_{a,k})_{a \in \mathcal{P},k\geq 0}$ and $(\td{\epsilon}_i)_{1\leq i\leq g}$ (i.e. fixed log-times and positions of the poles of $dy$ with non-vanishing residues and $x$ regular) with the associated $\Lambda_\Omega$ given in \autoref{SubsectionLambda}, then
\beq \forall \, (h,m)\in \mathbb{N}\times \mathbb{N}^*\setminus\{(0,1)\} \,:\, \delta_\Omega[\omega_{h,m}(z_1,\dots,z_m)]=\int_{\partial_\Omega} \Lambda_\Omega(s) \omega_{h,m+1}(z_1,\dots,z_m,s),\eeq
where $\omega_{h,n}$ is defined by LogTR.
\end{theorem}

\begin{proof}The proof is done by induction and is detailed in \autoref{AppendixVariationProof}. 
\end{proof}

Note that we can check that the variational formulas of \autoref{TheoVariationalFormulas} are compatible with the dilaton equations of \autoref{TheoremDilatonEquation}, even if the dilaton equations differ in LogTR while the variational formula stays the same for variations with respect to the parameters of $\td{y}dx$. For completeness, this proof is presented in \autoref{AppendixCompatibilityDilatonVar}.

In addition to extending the dilaton equations, the definition of the free energies also extends the variational formulas. Indeed, we have the following theorems.

\begin{theorem}[Variational formulas of the free energies for standard times]\label{TheoVarFreenergiesStandardtimes} For any variation $\delta_\Omega$ with respect to $(\td{t}_{a,k})_{a \in \mathcal{P},k\geq 0}$ and $(\td{\epsilon}_i)_{1\leq i\leq g}$ (i.e. fixed log-times and positions of the poles of $dy$ with non-vanishing residues and $x$ regular) with the associated $\Lambda_\Omega$ given in \autoref{SubsectionLambda}, then
    \beq \forall\, h\geq 2\,:\, \delta_{\Omega}[\omega_{h,0}]=\int_{\partial_{\Omega}}\Lambda_{\Omega}(s)\omega_{h,1}(s)\eeq
\end{theorem}

\begin{proof}The proof is given in Appendix \ref{AppendixCompatibilityDilatonVarFreEnergies}.  
\end{proof}

Finally, we also need to state the variational formulas for the special free energy $\omega_{1,0}=F_1$:

\begin{theorem}[Variational formulas for $F_1$ for standard times]\label{ThVarFormulaF1StandardTimes}For any variation $\delta_\Omega$ with respect to $(\td{t}_{a,k})_{a \in \mathcal{P},k\geq 0}$ and $(\td{\epsilon}_i)_{1\leq i\leq g}$ (i.e. fixed log-times and positions of the poles of $dy$ with non-vanishing residues and $x$ regular) with the associated $\Lambda_\Omega$ given in \autoref{SubsectionLambda}, then
    \beq \delta_\Omega[\omega_{1,0}]= \int_{\partial_\Omega} \omega_{1,1}(q) \Lambda_\Omega(q)\eeq
\end{theorem}

\begin{proof}\autoref{ThVarFormulaF1StandardTimes} is proved in Appendix \ref{AppendixVarF1}.   
\end{proof}

\subsection{Variational formulas with respect to LogTR-vital singularities}\label{SectionVarVitalSing}
In this section, we shall now consider variations with respect to the positions of poles of $dy$ with non-vanishing residues and where $x$ is regular. This implicitly means that all other parameters are fixed (including the log-times $(y_\alpha)_{\alpha \in \mathcal{S}_y}$). For clarity let us define:
\beq \mathcal{S}_{y,0}=\{a\in \mathcal{P}_y\,/\, \Res_{q\to a} dy=y_a\neq 0\, \text{ and } x(a)\in \mathbb{C}\}.\eeq
Note that from \autoref{MainAssumption}, for any $a\in \mathcal{S}_{y,0}$ we have $dx(a)\neq 0$ because ramification points are assumed to differ from the poles of $dy$.

Performing variations with respect to the positions of poles $a\in \mathcal{S}_{y,0}$ generates two different contributions: the first one from the variations of the logarithmic part of $ydx$ at $a$ (i.e. variations of $y_a\ln \frac{E(q,a)}{E(q,o)} dx$) and the second one from variations of $(B_{a,k})_{k\geq 1}$ and $B_{a,0,o'}$ appearing in the global decomposition of $ydx$ (\autoref{TheoremGlobalDecompositionydx}). However, similarly to the position of poles of $dy$ with vanishing residues, this second contribution corresponds to a linear combination of variations with respect to the irregular times $(t_{a,k})_{k\geq 1}$ and monodromy $t_{a,0}$ of $\td{y}dx$ at $a$. Therefore, this second contribution is already encoded by the results of \autoref{TheoVariationalFormulas} and of \autoref{SectionVartdydx} and we shall ignore it in this section.
\smallskip

Thus, in this section, we will focus only on the variations of the position of a pole $a\in \mathcal{S}_{y,0}$ which is a simple pole of $dy$, i.e. \textbf{only variations with respect to a LogTR-vital singularity $\left(a_s\right)_{1\leq s\leq M}$} and we shall denote it $\delta_\Omega=\partial_{a_s}$ for $s\in \llbracket 1,M\rrbracket$ in the rest of this section. 

The main theorem is the following:

\begin{theorem}[Variations of the correlators with respect to LogTR-vital singularities]\label{TheoVariationsLogTRpoles}Let $r\in \llbracket 1,M\rrbracket$ then we have for any $(h,n)\in \mathbb{N}\times\mathbb{N}^*\setminus \{(0,1)\}$:
\bea d_{a_r}[\omega_{h,n}(z_1,\dots,z_n)] 
&=& \sum_{i=1}^N \Res_{q\to p_i} d_{a_r}[\Phi_{p_i}(q)] \omega_{h,n+1}(z_1,\dots,z_n,q)\cr\label{NewVariation}
&&+  \Res_{q\to a_r} \frac{dx(a_r)}{dx(q)}\sum_{h_1=1}^h \omega_{h_1,1}(q)\omega_{h-h_1,n+1}(q,z_1,\dots,z_n).
\cr&&
\eea
where $d_{a_r}[\Phi_{p_i}(q)]=\int_{p_i}^q \Omega_{a_r}$ is the local anti-derivative of $\Omega_{a_r}(q)=y_{a_r}d_{a_r}[\ln E(a_r,q)] dx(q)$ defined locally around $p_i$ and we have denoted $d_{a_r}[f]:=da_r\, \partial_{a_r}[f]$. 
\end{theorem}

\begin{proof}The proof is done in \autoref{AppendixVariationProofLogPoles}.
\end{proof}



For completeness, one can check that the variational formulas \autoref{TheoVariationsLogTRpoles} are compatible with the dilaton equations of \autoref{TheoremDilatonEquation}. This is carried out in \autoref{AppendixCompatVarLogTR}.

Returning to the discussion at the beginning of the section on how the variational formula can be compared with the dilaton equation, we would like to emphasize the following structure. Recall the dilaton equation in LogTR (\autoref{TheoremDilatonEquation}):
\bea\label{DilatonEquation2}
(2-2h-k)\omega_{h,k}(z_1,\dots,z_k)
&=&
\sum_{i=1}^N \Res_{z\to p_i} \Phi(z)\omega_{h,k+1}(z,z_1,\dots,z_k)
\cr&&
-\sum_{j=1}^M \Res_{z\to a_j}\frac{x(z)-x(a_j)}{dx(z)}
\overset{h}{\underset{h_1=1}{\sum}}
\omega_{h_1,1}(z)\omega_{h-h_1,k+1}(z,z_1,\dots,z_k).
\cr&&
\eea

Acting on this equation with a variation with respect to $a_r$ yields a remaining contribution on $\Phi(z)$ and on $\frac{x(z)-x(a_j)}{dx(z)}$, with $j=r$. The action on the correlators produces precisely a cancellation of the prefactor $(2-2h-k)$ on the left-hand side. It is striking that this structural feature persists in the setting of LogTR, in particular for variations of LogTR-vital singularities.\\

Finally, we derive the variational formulas for the free energies defined in the setting of LogTR (see \autoref{DefFreeEnergies}). We have the following variational formulas

\begin{theorem}[Variations of the free energies with respect to LogTR-vital singularities]\label{VariationalFormulaLogTRPole}\sloppy{For any $r\in \llbracket 1,M\rrbracket$ and any $h\geq 2$, we have:
\bea \label{EqToProve}&& d_{a_r}[\omega_{h,0}]=\sum_{i=1}^N\Res_{z\to p_i} d_{a_r}[\Phi_{p_i}(z)]\omega_{h,1}(z) +\frac{1}{2}\Res_{z\to a_r}\frac{dx(a_r)}{dx(z)}\overset{h-1}{\underset{h_1=1}{\sum}}\omega_{h_1,1} (z)\omega_{h-h_1,1} (z)\cr&&
-\Res_{z\to a_r} dx(a_r) dy(z) \int_o^z\omega_{h,1}-\sum_{j=1}^M\Res_{z\to a_j} d_{a_r}[y(z)] dx(z)\int_o^z\omega_{h,1}
\eea
where $d_{a_r}[\Phi_{p_i}(z)]=y_{a_r}\int_{q=p_i}^{q=z} d_{a_r}[\ln (E(a_r,q))] dx(q)$ that is locally holomorphic near the ramification points. The last formula is equivalent to 
\bea \label{EqToProve2}&& d_{a_r}[\omega_{h,0}]=\Res_{z\to \{p_i\}\cup \{a_j\}} d_{a_r}[\Phi(z)]\omega_{h,1}(z) +\frac{1}{2}\Res_{z\to a_r}\frac{dx(a_r)}{dx(z)}\overset{h}{\underset{h_1=0}{\sum}}\omega_{h_1,1} (z)\omega_{h-h_1,1} (z)
\eea
provided that we regroup $\underset{z\to a_r}{\Res} d_{a_r}[\Phi(z)] \omega_{h,1}(z)$ with $\underset{z\to a_r}{\Res} \frac{dx(a_r)}{dx(z)} \omega_{h,1}(z) ydx(z)=dx(a_r)\underset{z\to a_r}{\Res} \omega_{h,1}(z) y(z)$ to have something well-defined in the second formulation.}
\end{theorem}

\begin{proof}The proof is done in \autoref{AppendixProofVarationalFormulasLogR}.
\end{proof}

We also state separately the variational formula for the special free energy $\omega_{1,0}=F_1$:

\begin{theorem}[Variational of $F_1$ with respect to LogTR-vital singularities]\label{ThVarFormulaF1LogTRTimes}
For any $r\in \llbracket 1,M\rrbracket$:
\beq d_{a_r}[\omega_{1,0}] =\Res_{z\to \{p_i\}\cup \{a_j\}} d_{a_r}[\Phi(z)]\omega_{1,1}(z) +\Res_{z\to a_r}dx(a_r)y(z)\omega_{1,1} (z).\eeq
\end{theorem}

\begin{proof}\autoref{ThVarFormulaF1LogTRTimes} is proved in \autoref{AppendixVarF1LogTR}.
\end{proof}

Let us compare, for the free energies, the dilaton equations with the variational formulas, as mentioned at the beginning of the section for the standard TR setting. The definition of the free energies was rather involved \autoref{DefFreeEnergies}:
\bea \label{FormulasDefFg2}
(2-2h)\omega_{h,0}
&=&
\sum_{i=1}^N \Res_{z\to p_i} \Phi(z)\omega_{h,1}(z)
\cr&&
-\sum_{s=1}^M \Res_{z\to a_s}\big(x(z)-x(a_s)\big)\left(\frac{1}{2}
\overset{h-1}{\underset{h_1=1}{\sum}}
\frac{\omega_{h_1,1}(z)\omega_{h-h_1,1}(z)}{dx(z)}
- dy(z)\int_o^z \omega_{h,1}\right).
\cr&&
\eea

However, taking the derivative with respect to $a_r$ again yields precisely the remaining action in \eqref{EqToProve} on $\Phi(z)$, $\big(x(z)-x(a_s)\big)$, and $dy(z)$, after canceling the factor $(2-2h)$ coming from the variation of the correlators. Even more interesting is the formulation of the variational formula in \eqref{EqToProve2}, which can be written in this form precisely because there is a corresponding cancellation arising from the quadratic term in the correlators. 

Even though $\Phi(z)$ and $\omega_{0,1}$ have logarithmic singularities where the residue is not defined, our derivation shows how this is regularized in the correct manner. This type of regularization is not needed for algebraic curves (standard TR), i.e.\ when $x$ and $y$ are meromorphic.

Finally, let us emphasize that all computations regarding the dilaton equation and the variational formulas are derived purely from the recursive definition of LogTR. However, the striking compatibility and similarity between the dilaton equation and the variational formula constitute a nontrivial consequence.

\section{Conclusion and future directions}

In this article, we have studied global and local geometric properties arising in LogTR, an extension of the standard Eynard--Orantin topological recursion obtained by allowing $dy$ to have residues at points where $dx$ is regular. It is known that standard TR yields incorrect enumerative invariants in this situation, whereas LogTR produces the expected ones. This can also be understood via limiting procedures, where TR converges to LogTR precisely in this regime.

We derived the dilaton equation for LogTR for all $\omega_{h,n}$, which naturally leads to a definition of the free energies in this framework. We have verified that this definition reproduces the expected free energies in two examples: for the half Seiberg--Witten curve, corresponding to the perturbative sector of $4d$ $\mathcal{N}=2$ pure supersymmetric gauge theory, and for the mirror curve of strip geometries, reproducing the Gromov--Witten invariants of the corresponding Calabi--Yau threefold.

We then studied, under suitable assumptions, a decomposition of $\omega_{0,1}$ (analogous to the standard TR setting), allowing for additional logarithmic contributions. Based on this decomposition, we analyzed the variation of the LogTR differentials with respect to the associated parameters. In particular, we investigated variations with respect to the LogTR-vital singularities, namely those points responsible for the difference between TR and LogTR. All such variations are shown to be compatible with the dilaton equations, providing further support for the proposed definition of the free energies.

We emphasize that the variational formulas with respect to LogTR-vital singularities differ significantly from the usual variational formulas. This suggests that these parameters do not belong to the same class as the standard times used in TR. It would therefore be interesting to study variations of both TR and LogTR with respect to alternative sets of parameters, such as the positions of the ramification points. This cannot be done while keeping $x$ fixed. It is natural to expect that LogTR-vital singularities play a role analogous to ramification points, since in certain limits of standard TR, the latter converge to the former. We plan to investigate this in future work. We hope that this will lead to a better understanding of variational formulas, in particular for mirror curves, where such results have not yet been rigorously established.

There are several further research directions arising from the geometric understanding of LogTR:
\begin{itemize}
    \item \textbf{Quantum curves in $\mathbb{C}^*$:} One of the main motivations for studying variational formulas in LogTR is their application to quantum curve constructions, where such formulas are essential. It is known that, under suitable assumptions, standard TR with meromorphic $x,y$ leads to quantum curves constructed from the Eynard--Orantin differentials \cite{Norbury_survey,BouchardEynard_QC,Iwaki-P1,MO19_hyper,EGF19,Quantization_2021}. In higher genus, non-perturbative contributions related to filling fractions arise through variational formulas. However, for curves in $\mathbb{C}^*$, no general results on quantum curves are known; only case-by-case studies exist \cite{Marchal:2017ntz,ALS,Banerjee:2025shz}. From our perspective, this is largely due to the lack of a systematic understanding of variational formulas in this setting. The present work constitutes an important step in this direction. In particular it should be essential to write the corresponding KZ equations and obtain the difference equations that are expected to replace standard differential equations as derived in \cite{Quantization_2021}.

    \item \textbf{Applications to knot theory:} One expected application of standard TR (and hence also LogTR) is to knot theory \cite{Gukov:2011qp,Brini:2011wi,BEApol}. Certain spectral curves give rise to knot invariants via the associated differentials. While no general framework is currently available, these curves typically lie in $\mathbb{C}^* \times \mathbb{C}^*$, suggesting that LogTR is the appropriate formalism, especially in situations involving limits, where it behaves well. Examples include the $A$-polynomial and related constructions \cite{Brini:2011wi}.

\item \textbf{Augmentation varieties:} Another class of spectral curves of interest in knot theory, motivated by large $N$ duality and knot contact homology, is given by augmentation varieties \cite{Aganagic:2012jb}. These curves also lie in $\mathbb{C}^*\times\mathbb{C}^*$. A modified topological-recursion framework for augmentation varieties was proposed in \cite{Gu:2014yba}, where a calibrated annulus kernel is introduced and checked for torus knots. Moreover, the planar free energy extracted from augmentation varieties was shown to agree with the quantum part of the planar free energy of the resolved conifold in several examples. This suggests that it would be very interesting to revisit augmentation varieties from the viewpoint of LogTR, in particular to understand whether logarithmic effects and the corresponding variational formulas provide a more natural and possibly more general framework for knot-theoretic spectral curves.

    \item \textbf{Refined structures:} In our examples for the free energies, we observed that residues of $dy$ at LogTR-vital singularities generate structures reminiscent of refined topological string theory. In particular, the usual Bernoulli numbers are replaced by double Bernoulli numbers when allowing general residues $y_{a_s}$. Such structures also appear in various refined settings. It remains unclear how these different refinements are related. It is conceivable that they arise from higher-dimensional theories, such as Calabi--Yau fivefolds recently studied in \cite{Brini:2024gtn,schuler2026gromovwitteninvariantsmembraneindices}, which reduce to known refinements under suitable specializations.
\end{itemize}

\renewcommand{\theequation}{\thesection-\arabic{equation}}
\appendix

\section{Proof of \autoref{TheoremDilatonEquation}}\label{AppendixDilatonProof}
The proof is done by induction on $2h+k$. Initialization is a direct consequence of the fact that the dilaton equations holds for any $(h=0,k)$ with $k\geq 2$. Indeed, correlation functions $\omega_{0,k}(z_1,\dots,z_k)$ are the same as in the standard topological recursion (in its local version since $(x,y)$ are not necessarily meromorphic function on $\Sigma$) and thus satisfying the standard topological recursion dilaton equations:
\beq -(k-2)\omega_{0,k}(z_1,\dots,z_k)=\sum_{i=1}^N\Res_{z\to p_i} \Phi(z)\omega_{0,k+1}(z,z_1,\dots,z_k)\eeq
Finally, for $h=0$, the sum with respect to the LogTR-vital singularities in the r.h.s. of the dilaton equation is empty hence proving that \eqref{DilatonEquation} holds for $h=0$. In particular, the induction is initialized as the subcase $(h,k)=(0,2)$.

\medskip

Let $k\geq 1$ and $h\geq 0$ with $(h,k)\notin\{(0,1)\}$ and assume that the property  \eqref{DilatonEquation} holds for any $(h',m)\neq (0,1)$ such that $2h'+m<2h+k$. Consider the r.h.s. of \eqref{DilatonEquation} and insert the definition of LogTR:
\bea &&\text{RHS}_{h,k}(z,z_1,\dots,z_{k-1}):=\sum_{i=1}^N\Res_{z_k\to p_i} \Phi(z_k)\omega_{h,k+1}(z,z_1,\dots z_k)\cr&&
-\sum_{j=1}^M\Res_{z'\to a_j}\frac{x(z')-x(a_j)}{dx(z')}
        \overset{h}{\underset{h_1=1}{\sum}}\omega_{h_1,1} (z')\omega_{h-h_1,k+1} (z',z,z_1,\dots,z_{k-1})\cr
&&\overset{\eqref{LogTRDef}}{=}\sum_{i=1}^N\sum_{j=1}^N\Res_{z_k\to p_i}\Res_{q\to p_j} \Phi(z_k)\bigg( \frac{1}{2}\frac{\int_{q}^{\sigma_j(q)} \omega_{0,2}(z,.)}{\omega_{0,1}(\sigma_j(q))-\omega_{0,1}(q)}\Big(\omega_{h-1,k+2}(q,\sigma_j(q),z_1,\dots,z_{k})\cr
    &&+\sum_{\substack{h_1+h_2=h\\I_1\sqcup I_2=\{1,\dots, k\} \\(h_i,|I_i|)\neq (0,0)}} \omega_{h_1,|I_1|+1}(\sigma_j(q),z_{I_1}) \omega_{h_2,|I_2|+1}(q,z_{I_2}) \Big)\bigg) \cr
    && -\sum_{j=1}^M\Res_{z'\to a_j}\frac{(x(z')-x(a_j))}{dx(z')}
        \overset{h}{\underset{h_1=1}{\sum}}\omega_{h_1,1} (z')\omega_{h-h_1,k+1} (z',z,z_1,\dots,z_{k-1})
\eea
Note that in the use of the definition of LogTR used above, there is no need for the $\delta_{m,1}$ special LogTR term since we have $k+1\geq 2$. We now put aside the term involving $\omega_{0,2}(q,z_k)$ and $\omega_{0,2}(\sigma_j(q),z_k)$. We get
\bea \text{RHS}_{h,k}&=&\frac{1}{2}\sum_{i=1}^N \sum_{j=1}^N \Res_{z_k\to p_i}\Res_{q\to p_j} \Phi(z_k)\frac{\int_{q}^{\sigma_j(q)} \omega_{0,2}(z,.)}{\omega_{0,1}(\sigma_j(q))-\omega_{0,1}(q)}\cr&&
\bigg(\omega_{0,2}(\sigma_j(q),z_k)\omega_{h,k}(q,z_1,\dots,z_{k-1}) + \omega_{0,2}(q,z_k)\omega_{h,k}(\sigma_j(q),z_1,\dots,z_{k-1}) \bigg)
\cr&&
+\sum_{i=1}^N\sum_{j=1}^N\Res_{z_k\to p_i}\Res_{q\to p_j} \Phi(z_k)\bigg( \frac{1}{2}\frac{\int_{q}^{\sigma_j(q)} \omega_{0,2}(z,.)}{\omega_{0,1}(\sigma_j(q))-\omega_{0,1}(q)}\Big(\omega_{h-1,k+2}(q,\sigma_j(q),z_1,\dots,z_{k})\cr
    &&+\sum_{\substack{h_1+h_2=h\\I_1\sqcup I_2=\{1,\dots, k\} \\(h_i,|I_i|)\neq (0,0)\\(h_i,I_i)\neq (0,\{z_k\})}} \omega_{h_1,|I_1|+1}(\sigma_j(q),z_{I_1}) \omega_{h_2,|I_2|+1}(q,z_{I_2}) \Big)\bigg) \cr
    && -\sum_{j=1}^M\Res_{z'\to a_j}\frac{x(z')-x(a_j)}{dx(z')}
        \overset{h}{\underset{h_1=1}{\sum}}\omega_{h_1,1} (z')\omega_{h-h_1,k+1} (z',z,z_1,\dots,z_{k-1})
\eea

The next step is to exchange the residues. At every ramification point, we have:
\bea \label{ResidueExchange}\Res_{z \to p_i}\Res_{q\to p_j}&=&\Res_{q\to p_j} \Res_{z \to p_i} +\delta_{i=j}\Res_{q \to p_i} \Res_{z \to q, \sigma_i(q)} \cr
\Res_{z \to p_i}\Res_{q\to p_j}&=&\Res_{q\to p_j} \Res_{z \to p_i} -\delta_{i=j}\Res_{z\to p_i}\Res_{q \to z, \sigma_i(z)} 
\eea
Note that this formula is valid locally around the ramification points even if $\Phi(z)$ has other singularities far away from the ramification points.  Thus we get:
\beq \label{RHSgk}\text{RHS}_{h,k}(z,z_1,\dots,z_{k-1})= \text{(I)} + \text{(II)} + \text{(III)} + \text{(IV)} \eeq
with
\bea \label{TermI} \text{(I)}&:=&\frac{1}{2}\sum_{i=1}^N \sum_{j=1}^{N} \Res_{z_k\to p_i}\Res_{q\to p_j} \Phi(z_k)\frac{\int_{q}^{\sigma_j(q)} \omega_{0,2}(z,.)}{\omega_{0,1}(\sigma_j(q))-\omega_{0,1}(q)}\cr&&
\bigg(\omega_{0,2}(\sigma_j(q),z_k)\omega_{h,k}(q,z_1,\dots,z_{k-1}) + \omega_{0,2}(q,z_k)\omega_{h,k}(\sigma_j(q),z_1,\dots,z_{k-1}) \bigg)
\eea
\bea\label{TermII} \text{(II)}&:=&\sum_{i=1}^N\sum_{j=1}^N\Res_{q\to p_j}\Res_{z_k\to p_i} \Phi(z_k)\bigg( \frac{1}{2}\frac{\int_{q}^{\sigma_j(q)} \omega_{0,2}(z,.)}{\omega_{0,1}(\sigma_j(q))-\omega_{0,1}(q)}\Big(\omega_{h-1,k+2}(q,\sigma_j(q),z_1,\dots,z_{k})\cr
    &&+\sum_{\substack{h_1+h_2=h\\I_1\sqcup I_2=\{1,\dots, k\} \\(h_i,|I_i|)\neq (0,0)\\(h_i,I_i)\neq (0,\{z_k\}) }} \omega_{h_1,|I_1|+1}(\sigma_j(q),z_{I_1}) \omega_{h_2,|I_2|+1}(q,z_{I_2}) \Big)\bigg)
\eea
\bea\label{TermIII} \text{(III)}&:=&\sum_{i=1}^N\Res_{q \to p_i} \Res_{z_k \to q, \sigma_i(q)}  \Phi(z_k)\bigg( \frac{1}{2}\frac{\int_{q}^{\sigma_i(q)} \omega_{0,2}(z,.)}{\omega_{0,1}(\sigma_i(q))-\omega_{0,1}(q)}\Big(\omega_{h-1,k+2}(q,\sigma_i(q),z_1,\dots,z_k)\cr
    &&+\sum_{\substack{h_1+h_2=h\\I_1\sqcup I_2=\{1,\dots, k\} \\(h_i,|I_i|)\neq (0,0)\\(h_i,I_i)\neq (0,\{z_k\})}} \omega_{h_1,|I_1|+1}(\sigma_i(q),z_{I_1}) \omega_{h_2,|I_2|+1}(q,z_{I_2}) \Big)\bigg)
\eea
\bea \label{TermIV} \text{(IV)}&:=&-\sum_{j=1}^M\Res_{z'\to a_j}\frac{x(z')-x(a_j)}{dx(z')}\overset{h}{\underset{h_1=1}{\sum}}\omega_{h_1,1} (z')\omega_{h-h_1,k+1} (z',z,z_1,\dots,z_{k-1})
\eea

We shall then discuss each of the four contributions $\text{(I)}$, $\text{(II)}$, $\text{(III)}$ and $\text{(IV)}$ separately.

\subsection{Contribution from $\text{(III)}$}
Term $\text{(III)}$ is the simplest one. In the integrand, we do not have term like $\omega_{0,2}(q,z_k)$ nor $\omega_{0,2}(\sigma_j(q),z_k)$ because they have been put aside inside $\text{(I)}$. Thus the integrand is regular at $z_k=q$ and $z_k=\sigma_{i}(q)$ for any $i\in \llbracket 1,N\rrbracket$. Hence the residue is vanishing:
\beq \label{ContributionIII}\text{(III)}=0\eeq

\subsection{Contribution from $\text{(I)}$}
Let us now look at the contribution of $\text{(I)}.$  We exchange the contour of integration using \eqref{ResidueExchange} 

\bea &&\text{(I)}:=\frac{1}{2}\sum_{i=1}^N \sum_{j=1}^N \Res_{z_k\to p_i}\Res_{q\to p_j} \Phi(z_k)\frac{\int_{q}^{\sigma_j(q)} \omega_{0,2}(z,.)}{\omega_{0,1}(\sigma_j(q))-\omega_{0,1}(q)}\cr&&
\bigg(\omega_{0,2}(\sigma_j(q),z_k)\omega_{h,k}(q,z_1,\dots,z_{k-1}) + \omega_{0,2}(q,z_k)\omega_{h,k}(\sigma_j(q),z_1,\dots,z_{k-1}) \bigg)\cr
&&\overset{q\to \sigma_j(q)}{=}\sum_{i=1}^N \sum_{j=1}^N \Res_{z_k\to p_i}\Res_{q\to p_j} \Phi(z_k)\frac{\int_{q}^{\sigma_j(q)} \omega_{0,2}(z,.)}{\omega_{0,1}(\sigma_j(q))-\omega_{0,1}(q)}\omega_{0,2}(q,z_k)\omega_{h,k}(\sigma_j(q),z_1,\dots,z_{k-1})\cr
&&\overset{\eqref{ResidueExchange}}{=}\sum_{j=1}^N \Res_{q\to p_j} \frac{\int_{q}^{\sigma_j(q)} \omega_{0,2}(z,.)}{\omega_{0,1}(\sigma_j(q))-\omega_{0,1}(q)}\omega_{h,k}(\sigma_j(q),z_1,\dots,z_{k-1})\sum_{i=1}\Res_{z_k\to p_i}\Phi(z_k)
 \omega_{0,2}(q,z_k) \cr&&
- \sum_{j=1}^N\Res_{q\to p_j} \frac{\int_{q}^{\sigma_j(q)} \omega_{0,2}(z,.)}{\omega_{0,1}(\sigma_j(q))-\omega_{0,1}(q)}\omega_{h,k}(\sigma_j(q),z_1,\dots,z_{k-1})\Res_{z_k\to q} \Phi(z_k)
\omega_{0,2}(q,z_k)\cr&&
\overset{\eqref{EqSpecialDilaton}}{=} 
- \sum_{j=1}^N\Res_{q\to p_j}\frac{\int_{q}^{\sigma_j(q)} \omega_{0,2}(z,.)}{\omega_{0,1}(\sigma_j(q))-\omega_{0,1}(q)}\omega_{h,k}(\sigma_j(q),z_1,\dots,z_{k-1}) \omega_{0,1}(q)\cr&&
\overset{\text{LLE}}{=} 
- \frac{1}{2}\sum_{j=1}^N\Res_{q\to p_j}\frac{\int_{q}^{\sigma_j(q)} \omega_{0,2}(z,.)}{\omega_{0,1}(\sigma_j(q))-\omega_{0,1}(q)}\omega_{h,k}(\sigma_j(q),z_1,\dots,z_{k-1}) (\omega_{0,1}(q)-\omega_{0,1}(\sigma_j(q)))\cr&&
= 
+ \frac{1}{2}\sum_{j=1}^N\Res_{q\to p_j}\left(\int_{q}^{\sigma_j(q)} \omega_{0,2}(z,.)\right)\omega_{h,k}(\sigma_j(q),z_1,\dots,z_{k-1}),
\eea
where we have used in the third step
\beq \forall \,i \in \llbracket 1,N\rrbracket\,:\, \sum_{i=1}^N\Res_{z_k\to p_i}\Phi(z_k)\omega_{0,2}(q,z_k)=0\eeq
since the integrand is regular at the ramification points. Thus,
\beq \label{ContributionI}\text{(I)}=\frac{1}{2}\sum_{j=1}^N\Res_{q\to p_j}\left(\int_{q}^{\sigma_j(q)} \omega_{0,2}(z,.)\right)\omega_{h,k}(\sigma_j(q),z_1,\dots,z_{k-1}).\eeq

\subsection{Contribution from $\text{(II)}$}
Let us now compute the contribution of $\text{(II)}$. First separating $z_k$, the double sum can be written as
\bea &&\sum_{\substack{h_1+h_2=h\\I_1\sqcup I_2=\{1,\dots, k\} \\(h_i,|I_i|)\neq (0,0)\\(h_i,I_i)\neq (0,\{z_k\}) }} \omega_{h_1,|I_1|+1}(\sigma_j(q),z_{I_1}) \omega_{h_2,|I_2|+1}(q,z_{I_2}) \cr&&
     =\sum_{h_1=0}^h\sum_{\substack{I_1\sqcup I_2=\{1,\dots, k-1\} \\(h_1,I_1)\neq (0,\emptyset)\\(h_1,I_2)\neq (h,\emptyset) }}\omega_{h_1,|I_1|+2}(\sigma_j(q),z_{I_1},z_k) \omega_{h-h_1,|I_2|+1}(q,z_{I_2})\cr&&
    + \sum_{h_1=0}^h\sum_{\substack{I_1\sqcup I_2=\{1,\dots, k-1\} \\(h_1,I_2)\neq (h,\emptyset)\\(h_1,I_1)\neq (0,\emptyset) }}\omega_{h_1,|I_1|+2}(q,z_{I_1},z_k)\omega_{h-h_1,|I_2|+1}(\sigma_j(q),z_{I_2})
\eea
We may now apply the induction on each term since we have excluded in the double sums the cases  $\underset{i=1}{\overset{N}{\sum}}\underset{z_k\to p_i}{\Res}\Phi(z_k)\omega_{0,2}(q,z_k) $ and $\underset{i=1}{\overset{N}{\sum}}\underset{z_k\to p_i}{\Res}\Phi(z_k)\omega_{0,2}(\sigma_j(q),z_k)$.  We find:

\small{\bea  \text{(II)}& \overset{\text{ind.}}{=}&
-(2h+k-3)\sum_{j=1}^N\Res_{q\to p_j} \bigg( \frac{1}{2}\frac{\int_{q}^{\sigma_j(q)} \omega_{0,2}(z,.)}{\omega_{0,1}(\sigma_j(q))-\omega_{0,1}(q)}\Big(\omega_{h-1,k+1}(q,\sigma_j(q),z_1,\dots,z_{k-1})\cr
    &&+\sum_{\substack{h_1+h_2=h\\I_1\sqcup I_2=\{1,\dots, k-1\} \\(h_i,|I_i|)\neq (0,0)}} \omega_{h_1,|I_1|+1}(\sigma_j(q),z_{I_1}) \omega_{h_2,|I_2|+1}(q,z_{I_2}) \Big)\bigg) \cr  
 &&+\sum_{j=1}^N\sum_{s=1}^M\Res_{q\to p_j}\Res_{z'\to a_s}\bigg( \frac{1}{2}\frac{\int_{q}^{\sigma_j(q)} \omega_{0,2}(z,.)}{\omega_{0,1}(\sigma_j(q))-\omega_{0,1}(q)}\frac{(x(z')-x(a_s))}{dx(z')} \cr
    &&\Big(\sum_{h'=1}^{h-1}\omega_{h',1}(z')\omega_{h-1-h',k+2}(z',q,\sigma_j(q),z_1,\dots,z_{k-1})+\sum_{h_1=0}^h\sum_{\substack{I_1\sqcup I_2=\{1,\dots, k-1\} \\(h_1,I_1)\neq (0,\emptyset)\\(h_1,I_2)\neq (h,\emptyset) }}\sum_{h'=1}^{h_1}\cr&&
     \omega_{h',1}(z')[ \omega_{h-h',|I_1|+2}(z',\sigma_j(q),z_{I_1})  \omega_{h-h_1,|I_2|+1}(q,z_{I_2}) + \omega_{h_1-h',|I_1|+2}(z',q,z_{I_1})\omega_{h-h_1,|I_2|+1}(\sigma_j(q),z_{I_2})] \Big)\cr&&
\eea}
\normalsize{Note} that in the last sums, we need to rule out $(h_1,I_1)=(0,\emptyset)$ and $(h_1,I_2)=(0,\emptyset)$ because it would have corresponded to a term $\Phi(z_k)\omega_{0,2}(\sigma_j(q),z_k)$ or $\Phi(z_k)\omega_{0,2}(q,z_k)$ that were put aside in $\text{(I)}$. We observe now that the first contribution corresponds to $\omega_{h,k}(z,z_1,\dots,z_{k-1})$  up to the new LogTR terms for $k=1$ in \eqref{LogTRDef}. Thus, we get:
\bea  &&\text{(II)}=
-(2h+k-3)\omega_{h,k}(z,z_1,\dots,z_{k-1})
\cr&&
-\delta_{k,1}(2h-2)\sum_{s=1}^M\Res_{q\to a_s}\left(\int_{a_s}^q\omega_{0,2}(z,.)\right)[\hbar^{2h}]\left(\frac{1}{\alpha_s\mathcal{S}(\alpha_s\hbar \partial_x)}\ln(q-a_s) \right)dx(q)\cr
 &&+\sum_{j=1}^N\sum_{s=1}^M\Res_{q\to p_j}\Res_{z'\to a_s}\bigg( \frac{1}{2}\frac{\int_{q}^{\sigma_j(q)} \omega_{0,2}(z,.)}{\omega_{0,1}(\sigma_j(q))-\omega_{0,1}(q)}\frac{(x(z')-x(a_s))}{dx(z')} \cr
    &&\Big(\sum_{h'=1}^{h-1}\omega_{h',1}(z')\omega_{h-1-h',k+2}(z',q,\sigma_j(q),z_1,\dots,z_{k-1})+\sum_{h_1=0}^h\sum_{\substack{I_1\sqcup I_2=\{1,\dots, k-1\} \\(h_1,I_1)\neq (0,\emptyset)\\(h_1,I_2)\neq (h,\emptyset) }}\sum_{h'=1}^{h_1}\cr&&
     \omega_{h',1}(z')\left[ \omega_{h-h',|I_1|+2}(z',\sigma_j(q),z_{I_1})  \omega_{h-h_1,|I_2|+1}(q,z_{I_2}) + \omega_{h_1-h',|I_1|+2}(z',q,z_{I_1})\omega_{h-h_1,|I_2|+1}(\sigma_j(q),z_{I_2})\right]\Big)\cr&&
\eea

Using the logarithmic projection property \eqref{LogProjectionProperty} we end up with the fact that for all $k\geq 1$ and $(h,k)\neq (0,1)$:
\bea \label{ContributionII}&&\text{(II)}
    = -(2h+k-3)\omega_{h,k}(z,z_1,\dots,z_{k-1})\delta_{k\neq 1}-(2h-2)\delta_{k,1}\sum_{i=1}^N \Res_{z'\to p_i} \int_{o}^{z'} \om_{0,2}(., z)\omega_{h,1}(z')\cr
 &&+\sum_{j=1}^N\sum_{s=1}^M\Res_{q\to p_j}\Res_{z'\to a_s}\bigg( \frac{1}{2}\frac{\int_{q}^{\sigma_j(q)} \omega_{0,2}(z,.)}{\omega_{0,1}(\sigma_j(q))-\omega_{0,1}(q)}\frac{(x(z')-x(a_s))}{dx(z')} \cr
    &&\Big(\sum_{h'=1}^{h-1}\omega_{h',1}(z')\omega_{h-1-h',k+2}(z',q,\sigma_j(q),z_1,\dots,z_{k-1})+\sum_{h_1=0}^h\sum_{\substack{I_1\sqcup I_2=\{1,\dots, k-1\} \\(h_1,I_1)\neq (0,\emptyset)\\(h_1,I_2)\neq (h,\emptyset) }}\sum_{h'=1}^{h_1}\cr&&
     \omega_{h',1}(z')\left[ \omega_{h-h',|I_1|+2}(z',\sigma_j(q),z_{I_1})  \omega_{h-h_1,|I_2|+1}(q,z_{I_2}) + \omega_{h_1-h',|I_1|+2}(z',q,z_{I_1})\omega_{h-h_1,|I_2|+1}(\sigma_j(q),z_{I_2})\right]\Big)\cr&&
\eea

\subsection{Collecting terms}
Let us now collect \eqref{ContributionI}, \eqref{ContributionII} and \eqref{ContributionIII} and insert them into \eqref{RHSgk}. We find:
\bea \label{InductionStep2} &&\text{RHS}_{h,k}(z,z_1,\dots,z_{k-1})=
\frac{1}{2}\sum_{j=1}^N\Res_{q\to p_j}\left(\int_{q}^{\sigma_j(q)} \omega_{0,2}(z,.)\right)\omega_{h,k}(\sigma_j(q),z_1,\dots,z_{k-1})\cr&&
-(2h+k-3)\omega_{h,k}(z,z_1,\dots,z_{k-1})\delta_{k\neq 1}-(2h-2)\delta_{k,1}\sum_{i=1}^N \Res_{z'\to p_i} \int_{o}^{z'} \om_{0,2}(., z)\omega_{h,1}(z')\cr
&&+ \text{ (A)}
\eea
where 
\small{\bea &&\label{DefTermA}\text{(A)}:=\sum_{j=1}^N\sum_{s=1}^M\Res_{q\to p_j}\Res_{z'\to a_s}\bigg( \frac{1}{2}\frac{\int_{q}^{\sigma_j(q)} \omega_{0,2}(z,.)}{\omega_{0,1}(\sigma_j(q))-\omega_{0,1}(q)}\frac{(x(z')-x(a_s))}{dx(z')} \cr
    &&\Big(\sum_{h'=1}^{h-1}\omega_{h',1}(z')\omega_{h-1-h',k+2}(z',q,\sigma_j(q),z_1,\dots,z_{k-1})+\sum_{h_1=0}^h\sum_{\substack{I_1\sqcup I_2=\{1,\dots, k-1\} \\(h_1,I_1)\neq (0,\emptyset)\\(h_1,I_2)\neq (h,\emptyset) }}\sum_{h'=1}^{h_1}\cr&&
     \omega_{h',1}(z')\left[ \omega_{h-h',|I_1|+2}(z',\sigma_j(q),z_{I_1})  \omega_{h-h_1,|I_2|+1}(q,z_{I_2}) + \omega_{h_1-h',|I_1|+2}(z',q,z_{I_1})\omega_{h-h_1,|I_2|+1}(\sigma_j(q),z_{I_2})\right]\Big)\cr&&
    -\sum_{s=1}^M\Res_{z'\to a_s}\frac{x(z')-x(a_s)}{dx(z')}
        \overset{h}{\underset{h_1=1}{\sum}}\omega_{h_1,1} (z')\omega_{h-h_1,k+1} (z',z,z_1,\dots,z_{k-1})
\eea}
\normalsize{}

\subsection{Computation of the term $\text{(A)}$}
Let us evaluate the contribution from $\text{(A)}$.

\bea&&\label{ContributionA}\text{(A)}
=\sum_{s=1}^M\Res_{z'\to a_s}\sum_{h'=1}^{h-1}\omega_{h',1}(z')\frac{x(z')-x(a_s)}{dx(z')}\sum_{j=1}^N\Res_{q\to p_j} \frac{1}{2}\frac{\int_{q}^{\sigma_j(q)} \omega_{0,2}(z,.)}{\omega_{0,1}(\sigma_j(q))-\omega_{0,1}(q)} \cr
    &&\Big(\omega_{h-1-h',k+2}(z',q,\sigma_j(q),z_1,\dots,z_{k-1})+\sum_{h_1=0}^{h-h'}\sum_{\substack{I_1\sqcup I_2=\{1,\dots, k-1\} \\(h_1,I_2)\neq (h-h',\emptyset) }}\cr&&
     [ \omega_{h-h',|I_1|+2}(z',\sigma_j(q),z_{I_1})  \omega_{h-h'-h_1,|I_2|+1}(q,z_{I_2}) + \omega_{h_1,|I_1|+2}(z',q,z_{I_1})\omega_{h-h'-h_1,|I_2|+1}(\sigma_j(q),z_{I_2})]\Big)\cr&&  
    +\sum_{s=1}^M\Res_{z'\to a_s}\omega_{h,1}(z')\frac{(x(z')-x(a_s))}{dx(z')}\sum_{j=1}^N\Res_{q\to p_j} \frac{1}{2}\frac{\int_{q}^{\sigma_j(q)} \omega_{0,2}(z,.)}{\omega_{0,1}(\sigma_j(q))-\omega_{0,1}(q)}\sum_{\substack{I_1\sqcup I_2=\{1,\dots, k-1\} \\I_2\neq \emptyset }}\cr&&
     [ \omega_{0,|I_1|+2}(z',\sigma_j(q),z_{I_1})  \omega_{0,|I_2|+1}(q,z_{I_2}) + \omega_{0,|I_1|+2}(z',q,z_{I_1})\omega_{0,|I_2|+1}(\sigma_j(q),z_{I_2})]\cr&&
    -\sum_{j=1}^M\Res_{z'\to a_j}\frac{x(z')-x(a_j)}{dx(z')}
        \overset{h}{\underset{h_1=1}{\sum}}\omega_{h_1,1} (z')\omega_{h-h_1,k+1} (z',z,z_1,\dots,z_{k-1})\cr&&
=\sum_{s=1}^M\Res_{z'\to a_s}\sum_{h'=1}^{h-1}\omega_{h',1}(z')\frac{x(z')-x(a_s)}{dx(z')}\omega_{h-h',k+1}(z',z,z_1,\dots,z_{k-1})\cr&& 
    +\sum_{s=1}^M\Res_{z'\to a_s}\omega_{h,1}(z')\frac{x(z')-x(a_s)}{dx(z')}\omega_{0,k+1}(z,z',z_1,\dots,z_{k-1})\delta_{k\geq 2}\cr&&
    -\sum_{j=1}^M\Res_{z'\to a_j}\frac{x(z')-x(a_j)}{dx(z')}\bigg(
        \overset{h}{\underset{h'=1}{\sum}}\omega_{h',1} (z')\omega_{h-h',k+1} (z',z,z_1,\dots,z_{k-1})\bigg)\cr&& 
=-\delta_{k,1} \sum_{s=1}^M\Res_{z'\to a_s}\frac{x(z')-x(a_s)}{dx(z')}
        \omega_{h,1} (z')\omega_{0,2} (z',z)
\eea

\subsection{End of the proof}
Inserting \eqref{ContributionA} into \eqref{InductionStep2} gives
\bea \label{InductionStep4} &&\text{RHS}_{h,k}=
\frac{1}{2}\sum_{j=1}^N\Res_{q\to p_j}\left(\int_{q}^{\sigma_j(q)} \omega_{0,2}(z,.)\right)\omega_{h,k}(\sigma_j(q),z_1,\dots,z_{k-1})\cr&&
-(2h+k-3)\omega_{h,k}(z,z_1,\dots,z_{k-1})\delta_{k\neq 1}-(2h-2)\delta_{k,1}\sum_{i=1}^N \Res_{z'\to p_i} \int_{o}^{z'} \om_{0,2}(., z)\omega_{h,1}(z')\cr&&
-\delta_{k,1} \sum_{s=1}^M\Res_{z'\to a_s}\frac{x(z')-x(a_s)}{dx(z')}\omega_{h,1} (z')\omega_{0,2} (z',z)
\eea

Let us now split cases depending on the value of $k$. 

\medskip

\textbf{First Case: $k\geq 2$}

When $k\geq 2$, we have
\bea \forall \,k\geq 2&:&
\frac{1}{2}\sum_{j=1}^N\Res_{q\to p_j}\left(\int_{q}^{\sigma_j(q)} \omega_{0,2}(z,.)\right)\omega_{h,k}(\sigma_j(q),z_1,\dots,z_{k-1})
 =-\omega_{h,k}(z,z_1,\dots,z_{k-1})\cr&&
 \eea
 because we can use the linear loop equation and move the contour of the residue and only pick poles at $q=z$. Inserting into \eqref{InductionStep4} we eventually get:
\beq\forall \, k\geq 2\,:\, \text{RHS}_{h,k}(z,z_1,\dots,z_{k-1})=-(2h+k-2)\omega_{h,k}(z,z_1,\dots,z_{k-1})
\eeq 
ending the induction procedure.

\medskip

\textbf{Second case: $k=1$}

When $k=1$, we also get pole at the LogTR-vital singularities:
\bea &&\frac{1}{2}\sum_{j=1}^N\Res_{q\to p_j}\left(\int_{q}^{\sigma_j(q)} \omega_{0,2}(z,.)\right)\omega_{h,1}(\sigma_j(q))
 =-\omega_{h,1}(z)-\sum_{s=1}^M\Res_{q\to a_s}\left(\int_{o}^{q} \omega_{0,2}(z,.)\right)\omega_{h,1}(q)\cr&&
 \eea
 which following from the linear loop equation and the logarithmic projection property.
Inserting into \eqref{InductionStep4} we eventually get:
\bea  \text{RHS}_{h,1}(z)
&=&-(2h-1)\omega_{h,1}(z)  +(2h-1)\sum_{s=1}^M \Res_{z'\to a_s} \int_{o}^{z'} \omega_{0,2}(., z)\omega_{h,1}(z')\cr&&
- \sum_{s=1}^M\Res_{z'\to a_s}\frac{x(z')-x(a_s)}{dx(z')}\omega_{h,1} (z')\omega_{0,2} (z',z)
\eea

Finally, using \autoref{LemmaIntW02Wg1}, we observe that the last two terms cancel and we eventually get:
\beq \text{RHS}_{h,1}(z)=-(2h-1)\omega_{h,1}(z)\eeq
ending the induction procedure.

\section{Proof of some lemmas: variations of the kernels}\label{AppendixLemmaVariations}
We have by definition
\beq \label{FixedLogTimes}\delta_{\Omega}[ydx]= \delta_{\Omega}[\td{y}dx].\eeq
Following the notations of \cite{EO07} and as explained in \autoref{SubsectionLambda} and \autoref{sec.variationsof}, we shall denote $\Omega$ the one-form such that
\beq\label{DefOmegaVariations} \delta_{\Omega}[y dx]\overset{\eqref{FixedLogTimes}}{=}\delta_{\Omega}[\td{y}dx]\overset{\eqref{NoXvariation}}{=}\delta_{\Omega} [\td{y}]dx= \Omega\eeq

Finally, we observe the following property:
\begin{lemma}[Local symmetry of $\Omega$ around ramification points]\label{PropLocalSymmetryOmegaBranchpoints} For any variation considered in this section, we have locally around a ramification point $p_j$
\beq \forall\, j\in \llbracket 1,N \rrbracket\,:\, \Omega(\sigma_j(r))=-\Omega(r)\eeq    
\end{lemma}

\begin{proof}The proof simply follows from the fact that variations considered in this article are done at fixed $x$. Hence, $\Omega=\delta_\Omega[y] dx$ and $y\circ \sigma_j=-y$ around a ramification point $p_j$ and so is $\Omega$.  
\end{proof}

The proof of the variational formulas for the correlators generated by LogTR will be done by induction. To perform it, one needs to compute first the variations of the Bergman kernel and of the recursion kernel.


\subsection{Variations of the Bergman kernel}
Rauch variational formula may be used to compute the variations of the Bergmann kernel. In our context, assuming that $\Omega$ has not pole at the ramification points, it states that (see \cite{EO07}) for deformations fixing $x$:
\bea \label{Rauch} \delta_{\Omega}B(p,q)&=&\sum_{i=1}^N\frac{\Omega(p_i)}{dy(p_i)}\Res_{r\to p_i}\frac{B(r,p)B(r,q)}{dx(r)}\cr
&=&\sum_{i=1}^N\Res_{r\to p_i}\frac{\Omega(r)B(r,p)B(r,q)}{dx(r)dy(r)}\cr
&=&-2\sum_{i=1}^N\Res_{r\to p_i}\frac{\Omega(r)dE_{i,r}(p)B(r,q)}{(y(r)-y(\sigma_i(r))dx(r)}=-2\sum_{i=1}^N\Res_{r\to p_i}\frac{\Omega(r)dE_{i,r}(p)B(r,q)}{\omega_i(r)}\cr
&=&\sum_{i=1}^N\Res_{r\to p_i}\frac{dE_{i,r}(p)}{\omega_i(r)}\left[\Omega(r)B(\sigma_i(r),q)+\Omega(\sigma_i(r))B(r,q)\right]
\eea
Note that $dy(p_i)\neq 0$ from \autoref{MainAssumption} (the zero loci of $dy$ and $dx$ are assumed disjoint). The third equality comes from the fact that the ramification points are simple and:
\beqq \frac{dE_{i,r}(p)}{y(r)-y(\sigma_i(r))}=\frac{-B(p_i,p)+O(r-p_i)}{2y'(p_i) +O(r-p_i)}=\left(\frac{-B(r,p)}{2dy(r)}\right)_{r=p_i}+O(r-p_i)
\eeqq Moreover, $\Omega(r)$ has no pole at the ramification points for any of our infinitesimal transformations. We recall that $dE_{i,r}$ and $\omega_i$ are defined locally around the ramification point $p_i$ by \autoref{DefVertexPropagator}.

\subsection{Variations of the recursion kernel}
By integrating once with respect to $q$ the variation $\delta_{\Omega}B(p,q)$, near each ramification point we get for all $j\in \llbracket 1,N\rrbracket$:
\bea \label{VariationdE} \delta_{\Omega}[dE_{j,q}(p)]
&=&\sum_{i=1}^N\Res_{r\to p_i}\frac{dE_{i,r}(p)}{\omega_i(r)}\left[\Omega(r)dE_{j,q}(\sigma_i(r))+\Omega(\sigma_i(r))dE_{j,q}(r)\right]
\eea

This implies that one may compute the variation of the recursion kernel (See \cite[Lemma $5.1$]{EO07}).
\begin{lemma}[Variations of the recursion kernel]\label{LemmaVariationsRecursionKernel} For any symmetric bilinear form $f(q,p)=f(p,q)$ on $\Sigma^2$ and any variation $\delta_\Omega$ considered in this section, one has
\bea
\delta_\Omega\left[
\sum_{k=1}^N\Res_{q\to p_k}\frac{dE_{k,q}(p)}{\omega_k(q)}\, f(q,\sigma_k(q))
\right]
&=&
2\sum_{j=1}^N\sum_{i=1}^N
\Res_{q\to p_j}\Res_{r\to p_i}
\frac{dE_{j,q}(p)}{\omega_j(q)}\,\Omega(q)\,
\frac{dE_{i,r}(q)}{\omega_i(r)}\, f(r,\sigma_j(r))
\cr
&&
+ \sum_{j=1}^N
\Res_{q\to p_j}
\frac{dE_{j,q}(p)}{\omega_j(q)}\,
\delta_{\Omega}[f(q,\sigma_j(q))] 
\eea
\begin{proof}
    The computation is similar to \cite[Appendix B]{EO07} but for completeness, we repeat the computation in the setting of LogTR and show that nothing changes.
    

Variations of $(\omega_{0,m})_{m\geq 3}$ with respect to the irregular times, monodromies and filling fractions are identical to the standard TR case since the recursion is the same in LogTR and in TR for these correlators and that \autoref{LemmaVariationsRecursionKernel} holds. Indeed, the computations of the correlators $(\omega_{0,m})_{m\geq 3}$ is the same in TR and LogTR and the diagrammatic proof of \cite{EO07} holds immediately as soon as \autoref{LemmaVariationsRecursionKernel} is verified and that variations can be rewritten as integrals of the Bergman kernel on a contour that is away from the ramification points. This is the case for these variations so that :
\beq \label{Variationh0}\forall\, m\geq 2\,:\, \delta_{\Omega}[ \omega_{0,m}(z_1,\dots,z_m)]=\int_{\partial_\Omega}\omega_{0,m+1}(z_1,\dots,z_m,q)\Lambda_\Omega(q)
\eeq
where $\Lambda_\Omega$ is defined in \autoref{SubsectionLambda}.

Application of the chain rule and of \eqref{VariationdE} together with the definition of the differential form $\Omega$ \eqref{DefOmegaVariations} provides:
\small{\bea &&\delta_{\Omega}\left[
\sum_{j=1}^N\Res_{q\to p_i}\frac{dE_{j,q}(p)}{\omega_j(q)}\, f(q,\sigma_j(q))
\right]=\sum_{j=1}^N\Res_{q\to p_j}\frac{dE_{j,q}(p)}{\omega_j(q)}\, \delta_{\Omega}[f(q,\sigma_j(q))] \cr&&- \sum_{j=1}^N\Res_{q\to p_j}\frac{dE_{j,q}(p)}{\omega_j(q)^2}(\Omega(q)-\Omega(\sigma_j(q))) f(q,\sigma_j(q))+ 2 \sum_{j=1}^N\Res_{q\to p_j}\sum_{i=1}^N\Res_{r\to p_i}\frac{dE_{i,r}(p)}{\omega_i(r)}\Omega(r) \frac{dE_{j,q}(r)}{\omega_j(q)}\, f(q,\sigma_j(q))\cr&&
=\sum_{j=1}^N\Res_{q\to p_j}\frac{dE_{j,q}(p)}{\omega_j(q)}\, \delta_{\Omega}[f(q,\sigma_j(q))] -2 \sum_{j=1}^N\Res_{q\to p_j}\frac{dE_{j,q}(p)}{\omega_j(q)^2}\Omega(q) f(q,\sigma_j(q))\cr&&+ 2 \sum_{j=1}^N\Res_{q\to p_j}\sum_{i=1}^N\Res_{r\to p_i}\frac{dE_{i,r}(p)}{\omega_i(r)}\Omega(r) \frac{dE_{j,q}(r)}{\omega_j(q)}\, f(q,\sigma_j(q))\cr&&
=\sum_{j=1}^N\Res_{q\to p_j}\frac{dE_{j,q}(p)}{\omega_j(q)}\, \delta_{\Omega}[f(q,\sigma_j(q))] -2 \sum_{j=1}^N\Res_{q\to p_j}\Res_{r\to q}\frac{dE_{j,q}(r)dE_{j,r}(p)}{\omega_j(q)\omega_j(r)}\Omega(r) f(q,\sigma_j(q))\cr&&+ 2 \sum_{j=1}^N\Res_{q\to p_j}\sum_{i=1}^N\Res_{r\to p_i}\frac{dE_{i,r}(p)}{\omega_i(r)}\Omega(r) \frac{dE_{j,q}(r)}{\omega_j(q)}\, f(q,\sigma_j(q))\cr&&
=\sum_{j=1}^N\Res_{q\to p_j}\frac{dE_{j,q}(p)}{\omega_j(q)}\, \delta_{\Omega}[f(q,\sigma_j(q))] +2 \sum_{i=1}^N\Res_{r\to p_i}\sum_{j=1}^N\Res_{q\to p_j}\frac{dE_{i,r}(p)}{\omega_{i}(r)} \Omega(r)\frac{dE_{j,q}(r)}{\omega_j(q)}f(q,\sigma_j(q))\cr&&
\eea
}
\normalsize{where} we have used $\Omega(\sigma_j(q))=-\Omega(q)$ in the first equality from \autoref{PropLocalSymmetryOmegaBranchpoints}. Moreover, for the second equality, the residue at $r\to q$ makes sense because $q$ is around the ramification point $p_j$ so $r$ is too. Finally the last equality comes from the inversion of the residues at the ramification points that are given by
\beqq \sum_{i=1}^N\Res_{r\to p_i}\sum_{j=1}^N\Res_{q\to p_j}= \sum_{j=1}^N\Res_{q\to p_j}\sum_{i=1}^N\Res_{r\to p_i}-\sum_{j=1}^N\Res_{q\to p_j}\Res_{r\to q}\eeqq

\end{proof}
\end{lemma}

\section{Proof of \autoref{TheoVariationalFormulas}}\label{AppendixVariationProof}
We will proceed by induction on $2h+m$. Let us first mention that the initialization $(h,m)=(0,2)$ is valid from the knowledge of the variations of $(\omega_{0,m})_{m\geq 2}$. Consider $h\geq 1$ and $m\geq 0$ such that $(h,m)\neq(0,1)$. Let us consider a variation with respect to irregular times, monodromies of filling fractions of $\td{y}dx$ (i.e. fixed log-times) with the associated $\Lambda_\Omega$ given in \autoref{SubsectionLambda} and assume that the formulas in \autoref{TheoVariationalFormulas} hold for the first correlators up to $2h'+m'<2h+m$. 
Applying the variation to the definition of LogTR \eqref{LogTRDef} and using \autoref{LemmaVariationsRecursionKernel} with the symmetric two form: $f_{h,m}(q,p):=\omega_{h-1,m+1}(p,q,z_2,\dots,z_m)
    +\underset{\substack{h_1+h_2=h\\I_1\sqcup I_2=\{2,\dots, m\} \\(h_i,|I_i|)\neq (0,0)}}{\sum} \omega_{h_1,|I_1|+1}(q,z_{I_1}) \omega_{h_2,|I_2|+1}(p,z_{I_2})$ gives
\bea &&\delta_\Omega[\omega_{h,m}(z_1,\dots,z_m) ]=-\sum_{i=1}^N\Res_{z\to p_i}\frac{dE_{i,z}(z_1)}{\omega_i(z)}\delta_{\Omega}\Big[\omega_{h-1,m+1}(z,\sigma_i(z),z_2,\dots,z_m)\cr
    &&+\sum_{\substack{h_1+h_2=h\\I_1\sqcup I_2=\{2,\dots, m\} \\(h_i,|I_i|)\neq (0,0)}} \omega_{h_1,|I_1|+1}(z,z_{I_1}) \omega_{h_2,|I_2|+1}(\sigma_i(z),z_{I_2}) \Big]\cr
    &&+2\sum_{j=1}^N\sum_{i=1}^N
\Res_{q\to p_j}\Res_{z\to p_i}
\frac{dE_{j,q}(z_1)}{\omega_j(q)}\,\Omega(q)\,
\frac{dE_{i,z}(q)}{\omega_i(z)}\Big(\omega_{h-1,m+1}(z,\sigma_i(z),z_2,\dots,z_m)\cr
    &&+\sum_{\substack{h_1+h_2=h\\I_1\sqcup I_2=\{2,\dots, m\} \\(h_i,|I_i|)\neq (0,0)}} \omega_{h_1,|I_1|+1}(z,z_{I_1}) \omega_{h_2,|I_2|+1}(\sigma_i(z),z_{I_2}) \Big)\cr
    &&-\delta_{m,1}\sum_{s=1}^M\Res_{z\to a_s}\delta_{\Omega}[dS_{a_s,z}(z_1)]dx(z)[\hbar^{2h}]\left(\frac{y_{a_s}}{\mathcal{S}(y_{a_s}^{-1}\hbar \partial_x)}\ln(z-a_s) \right)
    \eea  
Applying the induction assumption, using the relation between $\Omega(q)$ and $\Lambda_\Omega$ and finally computing the variation of $dS_{a_s,z}(z_1)$ using \eqref{Rauch}, we get
 
\bea \label{InductionVarStep1} &&\delta_\Omega[\omega_{h,m}(z_1,\dots,z_m) ]=-\int_{\partial_\Omega} \Lambda_\Omega(s) \sum_{j=1}^N\Res_{q\to p_j}\frac{dE_{j,q}(z_1)}{\omega_j(q)}\Big[\omega_{h-1,m+2}(q,\sigma_j(q),z_2,\dots,z_m,s)\cr
    &&+\sum_{\substack{h_1+h_2=h\\I_1\sqcup I_2=\{2,\dots, m\} \\(h_i,|I_i|)\neq (0,0)}} \omega_{h_1,|I_1|+2}(q,z_{I_1},s) \omega_{h_2,|I_2|+1}(\sigma_j(q),z_{I_2})+\omega_{h_1,|I_1|+1}(q,z_{I_1}) \omega_{h_2,|I_2|+2}(\sigma_j(q),z_{I_2},s) \Big]\cr
    &&+2\int_{\partial_\Omega} \Lambda_\Omega(s)\sum_{j=1}^N\sum_{i=1}^N
\Res_{q\to p_j}\frac{dE_{j,q}(z_1)}{\omega_j(q)}\,B(q,s)\Res_{z\to p_i}
\,
\frac{dE_{i,z}(q)}{\omega_i(z)}\Big(\omega_{h-1,m+1}(z,\sigma_i(z),z_2,\dots,z_m)\cr
    &&+\sum_{\substack{h_1+h_2=h\\I_1\sqcup I_2=\{2,\dots, m\} \\(h_i,|I_i|)\neq (0,0)}} \omega_{h_1,|I_1|+1}(z,z_{I_1}) \omega_{h_2,|I_2|+1}(\sigma_i(z),z_{I_2}) \Big)\cr
    &&+2\delta_{m,1}\sum_{k=1}^M\Res_{z\to a_k}\sum_{j=1}^N\Res_{q\to p_j}\frac{dE_{j,q}(z_1)}{\omega_j(q)}\Omega(q)dS_{a_k,z}(q)dx(z)[\hbar^{2h}]\left(\frac{y_{a_k}}{\mathcal{S}(y_{a_k}^{-1}\hbar \partial_x)}\ln(z-a_k) \right)
    \eea  
Replacing $\Omega(q)$ as a integral involving $\Lambda_\Omega$ in the last line and regrouping with the third line, we get:
\bea \label{InductionVarStep2} &&\delta_\Omega[\omega_{h,m}(z_1,\dots,z_m) ]=-\int_{\partial_\Omega} \Lambda_\Omega(s) \sum_{j=1}^N\Res_{q\to p_j}\frac{dE_{j,q}(z_1)}{\omega_j(q)}\Big[\omega_{h-1,m+2}(q,\sigma_j(q),z_2,\dots,z_m,s)\cr
    &&+\sum_{\substack{h_1+h_2=h\\I_1\sqcup I_2=\{2,\dots, m\} \\(h_i,|I_i|)\neq (0,0)}} \omega_{h_1,|I_1|+2}(q,z_{I_1},s) \omega_{h_2,|I_2|+1}(\sigma_j(q),z_{I_2})+\omega_{h_1,|I_1|+1}(q,z_{I_1}) \omega_{h_2,|I_2|+2}(\sigma_j(q),z_{I_2},s) \Big]\cr
    &&+2\int_{\partial_\Omega} \Lambda_\Omega(s)\sum_{j=1}^N
\Res_{q\to p_j}\frac{dE_{j,q}(z_1)}{\omega_j(q)}\,B(q,s)\Big[ \sum_{i=1}^N\Res_{z\to p_i}
\,
\frac{dE_{i,z}(q)}{\omega_i(z)}\Big(\omega_{h-1,m+1}(z,\sigma_i(z),z_2,\dots,z_m)\cr
    &&+\sum_{\substack{h_1+h_2=h\\I_1\sqcup I_2=\{2,\dots, m\} \\(h_i,|I_i|)\neq (0,0)}} \omega_{h_1,|I_1|+1}(z,z_{I_1}) \omega_{h_2,|I_2|+1}(\sigma_i(z),z_{I_2}) \Big) \cr&&
    +\delta_{m,1}\sum_{k=1}^M\Res_{z\to a_k}dS_{a_k,z}(q)dx(z)[\hbar^{2h}]\left(\frac{y_{a_k}}{\mathcal{S}(y_{a_k}^{-1}\hbar \partial_x)}\ln(z-a_k) \right)\Big]
    \eea 
The quantity in bracket is precisely the definition of $\omega_{h,m}(q,z_2,\dots,z_m)$ (given by \eqref{LogTRDef} with $z_1$ replaced by $q$) even when $m=1$ so that we end up with:
\bea \label{InductionVarStep3} &&\delta_\Omega[\omega_{h,m}(z_1,\dots,z_m) ]=-\int_{\partial_\Omega} \Lambda_\Omega(s) \sum_{j=1}^N\Res_{q\to p_j}\frac{dE_{j,q}(z_1)}{\omega_j(q)}\Big[\omega_{h-1,m+2}(q,\sigma_j(q),z_2,\dots,z_m,s)\cr
    &&+\sum_{\substack{h_1+h_2=h\\I_1\sqcup I_2=\{2,\dots, m\} \\(h_i,|I_i|)\neq (0,0)}} \omega_{h_1,|I_1|+2}(q,z_{I_1},s) \omega_{h_2,|I_2|+1}(\sigma_j(q),z_{I_2})+\omega_{h_1,|I_1|+1}(q,z_{I_1}) \omega_{h_2,|I_2|+2}(\sigma_j(q),z_{I_2},s) \Big]\cr
    &&+2\int_{\partial_\Omega} \Lambda_\Omega(s)\sum_{j=1}^N
\Res_{q\to p_j}\frac{dE_{j,q}(z_1)}{\omega_j(q)}\,B(q,s)\omega_{h,m}(q,z_2,\dots,z_m)
\eea 
Using the symmetry around ramification points $\Omega(\sigma_j(q))=-\Omega(q)$ and also the linear loop equation, it is equivalent to
\bea \label{InductionVarStep4} &&\delta_\Omega[\omega_{h,m}(z_1,\dots,z_m) ]=-\int_{\partial_\Omega} \Lambda_\Omega(s) \sum_{j=1}^N\Res_{q\to p_j}\frac{dE_{j,q}(z_1)}{\omega_j(q)}\Big[\omega_{h-1,m+2}(q,\sigma_j(q),z_2,\dots,z_m,s)\cr
    &&+\sum_{\substack{h_1+h_2=h\\I_1\sqcup I_2=\{2,\dots, m\} \\(h_i,|I_i|)\neq (0,0)}} \omega_{h_1,|I_1|+2}(q,z_{I_1},s) \omega_{h_2,|I_2|+1}(\sigma_j(q),z_{I_2})+\omega_{h_1,|I_1|+1}(q,z_{I_1}) \omega_{h_2,|I_2|+2}(\sigma_j(q),z_{I_2},s) \Big]\cr
    &&-\int_{\partial_\Omega} \Lambda_\Omega(s)\sum_{j=1}^N
\Res_{q\to p_j}\frac{dE_{j,q}(z_1)}{\omega_j(q)}\,(B(q,s)\omega_{h,m}(\sigma_j(q),z_2,\dots,z_m)+B(\sigma_j(q),s)\omega_{h,m}(q,z_2,\dots,z_m))\cr&& 
\eea 

Let us now observe that by definition for all $m\geq 1$ ($m+1\geq 2$ so there are no special LogTR terms).
\normalsize{Thus}, combining with \eqref{InductionVarStep4}, we get
\beq \delta_\Omega[\omega_{h,m}(z_1,\dots,z_m) ]=\int_{\partial_\Omega} \Lambda_\Omega(s) \omega_{h,m+1}(z_1,\dots,z_m,s)\eeq
ending the induction process.

\section{Compatibility of dilaton equations with variational formulas for standard times}\label{AppendixCompatibilityDilatonVar}

In this appendix, we prove that the dilaton equations for the correlators \autoref{TheoremDilatonEquation} are compatible with the variational formulas of the correlators \autoref{TheoVariationalFormulas}. 
\medskip
The dilaton equations for $h\geq 0, k\geq 1$ and $(h,k)\neq (0,1)$ are
\bea -(2h+k-2)\omega_{h,k}(z_1,\dots,z_k)&=&\sum_{i=1}^N\Res_{z\to p_i} \Phi(z)\omega_{h,k+1}(z,z_1,\dots,z_k) \cr&&
-\sum_{j=1}^M\Res_{z\to a_j}\frac{x(z)-x(a_j)}{dx(z)}
        \overset{h}{\underset{h_1=1}{\sum}}\omega_{h_1,1} (z)\omega_{h-h_1,k+1} (z,z_1,\dots,z_k)\cr&&
\eea
Applying $\delta_\Omega$, using \autoref{TheoVariationalFormulas} and the fact that
\beq \label{VarPhi} \delta_\Omega[\Phi(z)]= \int_o^z \delta_{\Omega} [ydx]= \int_o^z \Omega= \int_{\partial_\Omega} \Lambda_\Omega(s) dS_{z,o}(s)\eeq
together with the fact that the contour $\partial_\Omega$ is away from ramification points and $\mathcal{S}_y$ gives that
\small{\bea\label{Compat1} &&-(2h+k-2)\int_{\partial_\Omega}\Lambda_\Omega(s) \omega_{h,k+1}(z_1,\dots,z_k,s)=\int_{\partial_{\Omega}}\Lambda_\Omega(s) \sum_{i=1}^N\Res_{z\to p_i} dS_{z,o}(s)\omega_{h,k+1}(z,z_1,\dots,z_k)\cr
&&+ \int_{\partial_{\Omega}}\Lambda_\Omega(s)\sum_{i=1}^N\Res_{z\to p_i}\Phi(z)\omega_{h,k+2}(z,z_1,\dots,z_k,s)\cr&&
-\int_{\partial_{\Omega}}\Lambda_\Omega(s)\sum_{j=1}^M\Res_{z\to a_j}\frac{x(z)-x(a_j)}{dx(z)}\overset{h}{\underset{h_1=1}{\sum}}\omega_{h_1,2} (z,s)\omega_{h-h_1,k+1} (z,z_1,\dots,z_k)\cr&&
-\int_{\partial_{\Omega}}\Lambda_\Omega(s)\sum_{j=1}^M\Res_{z\to a_j}\frac{x(z)-x(a_j)}{dx(z)}\overset{h}{\underset{h_1=1}{\sum}}\omega_{h_1,1} (z)\omega_{h-h_1,k+2} (z,z_1,\dots,z_k,s)
\eea}
\normalsize{From} the logarithmic projection property and $k\geq1$, we get that
\bea  &&\sum_{i=1}^N\Res_{z\to p_i} dS_{z,o}(s)\omega_{h,k+1}(z,z_1,\dots,z_k)= \omega_{h,k+1}(s,z_1,\dots,z_k).
\eea
Similarly, since $k\geq 1$ we have that $(\omega_{h',r})_{h'\geq 0,r\geq 2}$ are regular at the $(a_j)_{1\leq j\leq M}$, hence
\beq \sum_{j=1}^M\Res_{z\to a_j}\frac{x(z)-x(a_j)}{dx(z)}\overset{h}{\underset{h_1=1}{\sum}}\omega_{h_1,2} (z,s)\omega_{h-h_1,k+1} (z,z_1,\dots,z_k)=0\eeq
Thus, \eqref{Compat1} is equivalent to
\small{\bea\label{Compat2} &&-(2h+k-1)\int_{\partial_\Omega}\Lambda_\Omega(s) \omega_{h,k+1}(z_1,\dots,z_k,s)=
 \int_{\partial_{\Omega}}\Lambda_\Omega(s)\Big[\sum_{i=1}^N\Res_{z\to p_i}\Phi(z)\omega_{h,k+2}(z,z_1,\dots,z_k,s) \cr&&
 -\sum_{j=1}^M\Res_{z\to a_j}\frac{x(z)-x(a_j)}{dx(z)}\overset{h}{\underset{h_1=1}{\sum}}\omega_{h_1,1} (z)\omega_{h-h_1,k+2} (z,z_1,\dots,z_k,s)\Big]
\eea}
\normalsize{This} equation is consistent with the dilaton equation for $\omega_{h,k+1}(z_1,\dots,z_k,s)$, 
such that \eqref{Compat2} is consistent.

\section{Proof of variational formulas for the free energies with respect to standard times}
\subsection{Proof of \autoref{TheoVarFreenergiesStandardtimes}: Variational formulas for $\omega_{h\geq 2,0}$}\label{AppendixCompatibilityDilatonVarFreEnergies}
Let $h\geq 2$ and recall that the free energies $\omega_{h,0}$ are defined by \autoref{DefFreeEnergies}:
\small{\beq (2-2h)\omega_{h,0}=\sum_{i=1}^N\Res_{z\to p_i} \Phi(z)\omega_{h,1}(z)
-\sum_{j=1}^M\Res_{z\to a_j}\big(x(z)-x(a_j)\big)\left(\frac{1}{2}
        \overset{h-1}{\underset{h_1=1}{\sum}}\frac{\omega_{h_1,1} (z)\omega_{h-h_1,1} (z)}{dx(z)}  -dy(z)\int_o^z\omega_{h,1}\right)
\eeq}
\normalsize{and} recall that
\beq \int_{\partial_\Omega} \Lambda_\Omega(s) B(s,.)= \Omega=\delta_{\Omega}[ydx].\eeq
Applying $\delta_\Omega$ and using \autoref{TheoVariationalFormulas} yields to 
\bea \label{VarFreeEnergies1} &&(2-2h)\delta_{\Omega}[\omega_{h,0}]= \int_{\partial_\Omega}\Lambda_\Omega(s) \sum_{i=1}^N \Res_{z\to p_i}  dS_{o,z}(s) \omega_{h,1}(z)+\Res_{z\to p_i} \int_{\partial_\Omega}\Lambda_\Omega(s)  \sum_{i=1}^N \Res_{z\to p_i}\Phi(z) \omega_{h,2}(z,s)
\cr&&
-\int_{\partial_\Omega}\Lambda_\Omega(s)\sum_{j=1}^M\Res_{z\to a_j}\frac{\big(x(z)-x(a_j)\big)}{dx(z)}\frac{1}{2}
        \overset{h-1}{\underset{h_1=1}{\sum}}\left(\omega_{h_1,2} (z,s)\omega_{h-h_1,1} (z)+  \omega_{h_1,1} (z)\omega_{h-h_1,2} (z,s)\right)\cr&& 
+\sum_{j=1}^M\Res_{z\to a_j}\big(x(z)-x(a_j)\big)\delta_{\Omega}\left[dy(z)\int_o^z\omega_{h,1}\right]
\eea
Since $dS_{o,z}(s) \omega_{h,1}(z)$ is a meromorphic one-form in $z$ with only poles at the ramification points, $(a_j)_{1\leq j\leq M}$ and $s$ we have from Riemann bilinear identity (\autoref{RiemannBilinearIdentity}) and the normalization of the Bergmann kernel and of the $(\omega_{h,1})_{h\geq 1}$:
\bea \sum_{i=1}^N \Res_{z\to p_i}  dS_{o,z}(s) \omega_{h,1}(z)
&=&\omega_{h,1}(s)+\sum_{j=1}^M \Res_{z\to a_j}  B(s,z) \left(\int_o^z\omega_{h,1}\right)
\eea
Note that contrary to the verification for the other correlators, the last term is not vanishing. Thus, \eqref{VarFreeEnergies1} is equivalent to (perform $h_1\to h-h_1$ in one of the sum)
\small{\bea \label{VarFreeEnergies2} &&(2-2h)\delta_{\Omega}[\omega_{h,0}]=\int_{\partial_{\Omega}}\Lambda_{\Omega}(s)\omega_{h,1}(s) +\int_{\partial_{\Omega}}\Lambda_{\Omega}(s)\sum_{j=1}^M \Res_{z\to a_j}  B(s,z) \left(\int_o^z\omega_{h,1}\right)\cr&&
+ \int_{\partial_\Omega}\Lambda_\Omega(s)  \sum_{i=1}^N \Res_{z\to p_i}\Phi(z) \omega_{h,2}(z,s)
-\int_{\partial_\Omega}\Lambda_\Omega(s)\sum_{j=1}^M\Res_{z\to a_j}\frac{\big(x(z)-x(a_j)\big)}{dx(z)}
        \overset{h-1}{\underset{h_1=1}{\sum}}\omega_{h_1,2} (z,s)\omega_{h-h_1,1}(z)\cr&&  
+\sum_{j=1}^M\Res_{z\to a_j}\big(x(z)-x(a_j)\big)\delta_{\Omega}\left[dy(z)\int_o^z\omega_{h,1}\right]
\eea}
\normalsize{Inserting} the dilaton equation for $\omega_{h,1}(s)$ 
in the r.h.s. of \eqref{VarFreeEnergies2} provides
\bea \label{VarFreeEnergies3} (2-2h)\delta_{\Omega}[\omega_{h,0}]&=&(2-2h)\int_{\partial_{\Omega}}\Lambda_{\Omega}(s)\omega_{h,1}(s) +\int_{\partial_{\Omega}}\Lambda_{\Omega}(s)\sum_{j=1}^M \Res_{z\to a_j}  B(s,z) \left(\int_o^z\omega_{h,1}\right)\cr&&
+\int_{\partial_{\Omega}}\Lambda_{\Omega}(s)\sum_{j=1}^M\Res_{z\to a_j}\frac{x(z)-x(a_j)}{dx(z)}\omega_{h,1} (z)\omega_{0,2} (z,s)
\cr&&+\sum_{j=1}^M\Res_{z\to a_j}\big(x(z)-x(a_j)\big)\delta_{\Omega}\left[dy(z)\int_o^z\omega_{h,1}\right]\cr
&:=&(2-2h)\int_{\partial_{\Omega}}\Lambda_{\Omega}(s)\omega_{h,1}(s) +(A)
\eea
where we have defined $(A)$ as:
\small{\bea \label{DefAVar} (A)&:=&\int_{\partial_{\Omega}}\Lambda_{\Omega}(s)\sum_{j=1}^M \Res_{z\to a_j}  B(s,z) \left(\int_o^z\omega_{h,1}\right)
+\int_{\partial_{\Omega}}\Lambda_{\Omega}(s)\sum_{j=1}^M\Res_{z\to a_j}\frac{x(z)-x(a_j)}{dx(z)}\omega_{h,1} (z)B(z,s)
\cr&&
+\sum_{j=1}^M\Res_{z\to a_j}\big(x(z)-x(a_j)\big)\delta_{\Omega}\left[dy(z)\int_o^z\omega_{h,1}\right]
\eea}

\normalsize{Let} us study the last term in details:
\small{\bea&& \sum_{j=1}^M\Res_{z\to a_j}\big(x(z)-x(a_j)\big)\delta_{\Omega}\left[dy(z)\int_o^z\omega_{h,1}\right]= \int_{\partial_{\Omega}}\Lambda_{\Omega}(s)\sum_{j=1}^M\Res_{z\to a_j}\big(x(z)-x(a_j)\big)d_z\Big(\frac{B(s,z)}{dx(z)}\Big)\left(\int_o^z\omega_{h,1}\right)\cr&& +\int_{\partial_{\Omega}}\Lambda_{\Omega}(s)\sum_{j=1}^M\Res_{z\to a_j}\frac{\big(x(z)-x(a_j)\big)}{dx(z)}B(s,z)dy(z)\left(\int_o^z\omega_{h,2}(s,.)\right)\cr
&&
=\int_{\partial_{\Omega}}\Lambda_{\Omega}(s)\sum_{j=1}^M\Res_{z\to a_j}\big(x(z)-x(a_j)\big)d_z\Big(\frac{B(s,z)}{dx(z)}\Big)\left(\int_o^z\omega_{h,1}\right)
\eea}
\normalsize{where} we have used for the last equality the fact that $\int_o^z\omega_{h,2}(s,.)$ is regular at $z=a_j$ 
Indeed, $a_j$ is a simple pole of $dy$ but $(x(z)-x(a_j))$ has a simple zero at $z=a_j$ so the integrand is regular.

Thus \eqref{DefAVar} reduces to
\bea \label{DefAVar2} (A)&=&\int_{\partial_{\Omega}}\Lambda_{\Omega}(s)\sum_{j=1}^M \Res_{z\to a_j}  B(s,z) \left(\int_o^z\omega_{h,1}\right)
+\int_{\partial_{\Omega}}\Lambda_{\Omega}(s)\sum_{j=1}^M\Res_{z\to a_j}\frac{x(z)-x(a_j)}{dx(z)}\omega_{h,1} (z)B(z,s)\cr&&
+\int_{\partial_{\Omega}}\Lambda_{\Omega}(s)\sum_{j=1}^M\Res_{z\to a_j}\big(x(z)-x(a_j)\big)d_z\Big(\frac{B(s,z)}{dx(z)}\Big)\left(\int_o^z\omega_{h,1}\right)\cr
&=&\int_{\partial_{\Omega}}\Lambda_{\Omega}(s)\sum_{j=1}^M \Res_{z\to a_j}d_z\left(\frac{\big(x(z)-x(a_j)\big)}{dx(z)}B(z,s)\left(\int_o^z\omega_{h,1}\right)\right)\cr
&=&0
\eea
where the last equality holds because 
is the integrand is a exact one-form.

In the end, inserting \eqref{DefAVar2} into \eqref{VarFreeEnergies3} yields
\beq (2-2h)\delta_{\Omega}[\omega_{h,0}]=(2-2h)\int_{\partial_{\Omega}}\Lambda_{\Omega}(s)\omega_{h,1}(s)\eeq
ending the proof.

\subsection{Proof of \autoref{ThVarFormulaF1StandardTimes}: Variational formula for $\omega_{1,0}$}\label{AppendixVarF1}
Let us first compute the special LogTR term in $\omega_{1,1}$:
\bea\label{SingPartomega11}[\hbar^{2}]\left(\frac{y_{a_s}}{\mathcal{S}(y_{a_s}^{-1}\hbar \partial_x)}\ln(z-a_s) \right)
&=& \frac{1}{24y_{a_s}}\left( \frac{1}{x'(z)^2(z-a_s)^2} +\frac{x''(z)}{x'(z)^3(z-a_s)}\right)
\eea
Thus, we get that for any variation $\delta_\Omega$ with respect to isomonodromic times, monodromies, and filling fractions of $\td{y}dx$:
\bea && -\int_{\partial_\Omega}\Lambda_\Omega(p)\sum_{s=1}^M\Res_{z\to a_s}dS_{a_s,z}(p)dx(z)[\hbar^{2}]\left(\frac{y_{a_s}}{\mathcal{S}(y_{a_s}^{-1}\hbar \partial_x)}\ln(z-a_s) \right)\cr&&= -\int_{\partial_\Omega}\Lambda_\Omega(p)\sum_{s=1}^M\frac{1}{24y_{a_s}}\Res_{z\to a_s}dS_{a_s,z}(p)\left( \frac{1}{x'(z)(z-a_s)^2} +\frac{x''(z)}{x'(z)^2(z-a_s)}\right)dz
\eea
Let us now observe that $dS_{a_s,z}(p)=\int_{a_s}^z B(.,p)$ has a simple zero at $z=a_s$. Moreover, from \autoref{MainAssumption}, $x'(a_s)\neq 0$ (because $a_s$ is not a ramification point since $dy$ has a simple pole at $a_s$ by definition). Thus, $\frac{dS_{a_s,z}(p)x''(z)}{x'(z)^2(z-a_s)}$ is regular at $z=a_s$ and so does not contribute to the residue. In the end, we have
\small{\bea -\int_{\partial_\Omega}\Lambda_\Omega(p)\sum_{s=1}^M\Res_{z\to a_s}dS_{a_s,z}(p)dx(z)[\hbar^{2}]\left(\frac{y_{a_s}}{\mathcal{S}(y_{a_s}^{-1}\hbar \partial_x)}\ln(z-a_s) \right)&=& -\sum_{s=1}^M\frac{1}{24y_{a_s} dx(a_s)}\int_{\partial_\Omega}\Lambda_\Omega(p) B(a_s,p)\cr
&=& -\sum_{s=1}^M\frac{1}{24y_{a_s} dx(a_s)}\Omega(a_s)
\eea}
\normalsize{where} we have used the fact that contours $\partial_\Omega$ are away from the $(a_s)_{1\leq s\leq M}$ and that $\Omega(p)=\int_{\partial_\Omega} B(p,q)\Lambda_{\Omega}(q)$. Variations with respect to isomonodromic times, monodromies or filling fractions are by definition fixing $(y_{a_s})_{1\leq s\leq M}$ and $x(z)$. Moreover, from \eqref{DefOmegaVariations}, we have $\Omega(a_s)=\delta_{\Omega}[\td{y}(a_s)] dx(a_s)$ (note that $\td{y}(a_s)$ makes sense whereas $y(a_s)$ diverges) so that  
\small{\bea\label{AppendixSpecialTermOmega10} -\int_{\partial_\Omega}\Lambda_\Omega(p)\sum_{s=1}^M\Res_{z\to a_s}dS_{a_s,z}(p)dx(z)[\hbar^{2}]\left(\frac{y_{a_s}}{\mathcal{S}(y_{a_s}^{-1}\hbar \partial_x)}\ln(z-a_s) \right)&=&\delta_{\Omega}\left[-\frac{1}{24}\sum_{s=1}^M\frac{1}{y_{a_s} dx(a_s)}\td{y}(a_s)dx(a_s)\right]\cr
&=&\delta_{\Omega}\left[-\frac{1}{24}\sum_{s=1}^M\frac{\td{y}(a_s)}{y_{a_s} }\right]
\eea}
\normalsize{Moreover}:
\beq\delta_{\Omega}\left[\frac{1}{24}\sum_{s=1}^M\bigg(\frac{y(z)}{y_{a_s} }-\log(x(z)-x(a_s))\bigg)\right]=\frac{1}{24}\sum_{s=1}^M\frac{\delta_{\Omega}[y(z)]}{y_{a_s} }=\frac{1}{24}\sum_{s=1}^M\frac{\delta_{\Omega}[\td{y}(z)]}{y_{a_s} }\eeq
so that 
\beq \delta_{\Omega}\left[\frac{1}{24}\sum_{s=1}^M\bigg(\frac{y(z)}{y_{a_s} }-\log(x(z)-x(a_s))\bigg)_{z=a_s}\right]= \frac{1}{24}\sum_{s=1}^M\frac{\delta_{\Omega}[\td{y}(a_s)]}{y_{a_s} }
\eeq
\medskip

\normalsize{Let} us now finish the proof. Recall that
\beq \omega_{1,1}(p)=\sum_{k=1}^N\Res_{q\to p_k}\frac{dE_{k,q}(p)}{\omega_k(q)}B(q,\sigma_k(q))-\sum_{s=1}^M\Res_{z\to a_s}dS_{a_s,z}(p)dx(z)[\hbar^{2}]\left(\frac{y_{a_s}}{\mathcal{S}(y_{a_s}^{-1}\hbar \partial_x)}\ln(z-a_s) \right)\eeq
Thus,
\bea \int_{\partial_\Omega}\omega_{1,1}(p)\Lambda_{\Omega}(p)&=&\sum_{i=1}^N\Res_{q\to p_i}\int_{\partial_\Omega}\frac{dE_{i,q}(p)\Lambda_{\Omega}(p)}{\omega(q)}B(q,\sigma_i(q))\cr
&&-  \int_{\partial_\Omega}\Lambda_{\Omega}(p)\sum_{s=1}^M\Res_{z\to a_s}dS_{a_s,z}(p)dx(z)[\hbar^{2}]\left(\frac{y_{a_s}}{\mathcal{S}(y_{a_s}^{-1}\hbar \partial_x)}\ln(z-a_s) \right)\cr
&=&\sum_{i=1}^N\Res_{q\to p_i}\int_{\partial_\Omega}\frac{dE_{i,q}(p)\Lambda_{\Omega}(p)dz_i(q)dz_i(\sigma_i(q))}{\omega(q)}\left[\frac{1}{(z(q)-z(\sigma_i(q)))^2}+\frac{1}{6}S_B(q)\right]\cr&&
-  \int_{\partial_\Omega}\Lambda_{\Omega}(p)\sum_{s=1}^M\Res_{z\to a_s}dS_{a_s,z}(p)dx(z)[\hbar^{2}]\left(\frac{y_{a_s}}{\mathcal{S}(y_{a_s}^{-1}\hbar \partial_x)}\ln(z-a_s) \right)\cr&&
\eea
where  $z_i$ is a local variable near the ramification points $p_i$ and $S_B$ is the corresponding Bergmann projective connection. Since the last term has a simple pole at the ramification point $p_i$, one can write
\bea \sum_{i=1}^N\Res_{q\to p_i}\int_{\partial_\Omega}\frac{dE_{i,q}(p)\Lambda_{\Omega}(p)dz_i(q)dz_i(\sigma_i(q))}{\omega_i(q)}\frac{1}{6}S_B(q)&=& -\frac{1}{2}\sum_{i=1}^N\frac{\Omega(p_i)}{dy(p_i)}\Res_{q\to p_i}\frac{B(q,\sigma_i(q))}{dx(q)}\cr
&=&-\frac{1}{2}\delta_\Omega[\ln \tau]
\eea

Since $z_i(\sigma_i(q))=-z_i(q)$, we have also
\bea &&\Res_{q\to p_i}\int_{\partial_\Omega}\frac{dE_{i,q}(p)\Lambda_{\Omega}(p)dz_i(q)dz_i(\sigma_i(q))}{\omega_i(q)}\frac{1}{(z(q)-z(\sigma_i(q)))^2}\cr
&&= \int_{\partial_\Omega}\Lambda_{\Omega}(p)\Res_{q\to p_i}\frac{dE_{i,q}(p)dz_i(q)dz_i(\sigma_i(q))}{\omega_i(q)}\frac{1}{(z(q)-z(\sigma_i(q)))^2}=-\frac{1}{24}\frac{\delta_\Omega[y'(p_i)]}{y'(p_i)}\eea
Indeed, recall that $\delta_\Omega[y(z)]dx(z)=\Omega(z)=\int_{\partial_\Omega} \Lambda_\Omega(p)B(p,z)$  so that since the ramification points are not varied (because $x$ is not varied):
\beq -\frac{1}{24}\frac{\delta_{\Omega}[y'(p_i)]}{y'(p_i)}=-\frac{1}{24y'(p_i)}\left[\int_{\partial_\Omega} \Lambda_\Omega(p) \partial_{z}\left(\frac{B(p,z)}{dx(z)}\right)\right]_{|\,z=p_i}\eeq

Eventually, collecting the previous results, we get
\bea \int_{\partial_\Omega}\omega_{1,1}(p)\Lambda_{\Omega}(p)&=&-\frac{1}{24}\sum_{i=1}^N\delta_\Omega[y'(p_i)]-\frac{1}{2}\delta_\Omega[\ln \tau]\cr
&&-  \int_{\partial_\Omega}\Lambda_{\Omega}(p)\sum_{s=1}^M\Res_{z\to a_s}dS_{a_s,z}(p)dx(z)[\hbar^{2}]\left(\frac{y_{a_s}}{\mathcal{S}(y_{a_s}^{-1}\hbar \partial_x)}\ln(z-a_s) \right)\cr&&
\eea 
We end up with
\beq \int_{\partial_\Omega}\omega_{1,1}(p)\Lambda_{\Omega}(p)=-\frac{1}{24}\sum_{i=1}^N\delta_\Omega[y'(p_i)]-\frac{1}{2}\delta_\Omega[\ln \tau]-\delta_{\Omega}\left[\frac{1}{24}\sum_{s=1}^M\bigg(\frac{y(z)}{y_{a_s} }-\log(x(z)-x(a_s))\bigg)_{z=a_s}\right]
\eeq
and thus
\beq \int_{\partial_\Omega}\omega_{1,1}(p)\Lambda_{\Omega}(p)= \delta_{\Omega}\left[-\frac{1}{24}\sum_{i=1}^Ny'(p_i) -\frac{1}{2}\ln \tau-\frac{1}{24}\sum_{s=1}^M\bigg(\frac{y(z)}{y_{a_s} }-\log(x(z)-x(a_s))\bigg)_{z=a_s}\right] \eeq
ending the proof.

\section{Proof of variational formulas with respect to LogTR-vital singularities for correlators}\label{AppendixVariationProofLogPoles}
\subsection{Properties of $\Omega_{a_r}$}
From the global decomposition of $ydx$ and the skew-symmetry of the prime form $E(q,p)=-E(p,q)$, we have that
\beq  ydx(q)=\sum_{a_s} y_{a_s}\ln \frac{E(q,a_s)}{E(q,o)} dx(q) + \td{y}dx(q)\eeq
Thus,
\beq d_{a_r}[ydx(q)]= y_{a_r}\frac{d_{a_r}[E(q,a_r)]}{E(q,a_r)}dx(q)= y_{a_r}d_{a_r}[\ln (E(q,a_r))] dx(q) \eeq
which defines
\beq \Omega_{a_r}(q):=
y_{a_r}d_{a_r}[\ln (E(a_r,q))] dx(q) \eeq
Thus $\Omega_{a_r}$ is a differential with the following properties:
\begin{itemize}
    \item It is holomorphic on $\Sigma$ except a simple pole at $q=a_r$.
    \item It has non-vanishing monodromies along the homology cycles
\end{itemize}
The fact that it has monodromies on the homology cycles implies that one needs to be careful in the computations. Computations make sense locally, hence residues involving $\Omega_{a_r}$ are allowed and can be handled standardly, but identities using global properties (for example the Riemann bilinear identity) should be handled with care. 

\subsection{Proof of \autoref{TheoVariationsLogTRpoles}}

In the formula, the contour of integration $\partial_{\Omega_{a_r}}$ includes a small loop around the ramification points. This case was excluded in the proof of \cite{EO07} and for the other variational formulas. This is why we need a specific proof for these variations and complications appear when exchanging the residues.

\medskip

Let us perform the induction procedure on $2h+n\geq 2$. We first need to initialize the induction for $(h,n)=(0,2)$. The Bergman kernel does not vary when we consider variations with respect to $a_r$, hence the l.h.s. is zero for $(h,n)=(0,2)$. The r.h.s. for $(h,n)=(0,2)$ reduces to
\beq \label{RHS02}\text{RHS}_{0,2}(z_1,z_2):=y_{a_r}\sum_{i=1}^N \Res_{q\to p_i}  
\left(\int_{p_i}^q \Omega_{a_r}\right)
\omega_{0,3}(z_1,z_2,q)\eeq
Moreover, recall that
\beq \omega_{0,3}(z_1,z_2,q)= \frac{1}{2}\sum_{i=1}^N \Res_{z\to p_i} \frac{dS_{z,\sigma_i(z)}(q)}{y(\sigma_i(z))dx(z) -y(z)dx(z)}\left(B(z,z_1)B(\sigma_i(z),z_2)+ B(\sigma_i(z),z_1)B(z,z_2) \right)\eeq
so that $\omega_{0,3}$ has at most double poles at $z_1=p_i$ (because the Bergman kernel has a double pole on the diagonal and no other singularities) and thus by symmetry of $\omega_{0,3}$, it may have at most double poles at $q=p_i$. On the contrary, we have:
\beq \int_{p_i}^q \Omega_{a_r}\overset{q\to p_i}{=} 0+ \frac{\Omega_{a_r}(p_i)}{dp_i}(q-p_i) + O\left((q-p_i)^2\right)=O\left((q-p_i)^2\right)\eeq
because $\Omega_{a_r}(p_i)= d_{a_r}[y(p_i)] dx(p_i)=0$ since $dx(p_i)$ because $p_i$ is a ramification point. Thus the integrand in the residue of \eqref{RHS02} is regular at $q=p_i$ so that the residue is zero and $\text{RHS}_{0,2}(z_1,z_2)=0$ initializing the induction procedure.

\medskip 

Let us now proceed with the induction and assume that the formula holds for any $(h',n')$ such that $2h'+n'< 2h+n$. We now look at $d_{a_r}\omega_{h,m}(z_1,\dots, a_m)$ and we recall that
\beq \label{ExprOmegaar}\Omega_{a_r}(z)da_r:= d_{a_r}[ydx(z)]=y_{a_r}d_{a_r}[\ln (E(a_r,z))] dx(z)\eeq
From \eqref{LogTRDef}, we have by definition:
\bea \label{DefomegahmlogTR}\omega_{h,m}(z_1,\dots,z_m)&:=&-\sum_{i=1}^N\Res_{z\to p_i}\frac{dE_{i,z}(z_1)}{\omega_i(z)}\Big(\omega_{h-1,m+1}(z,\sigma_i(z),z_2,\dots,z_m)\cr
    &&+\sum_{\substack{h_1+h_2=h\\I_1\sqcup I_2=\{2,\dots, m\} \\(h_i,|I_i|)\neq (0,0)}} \omega_{h_1,|I_1|+1}(z,z_{I_1}) \omega_{h_2,|I_2|+1}(\sigma_i(z),z_{I_2}) \Big)\cr
    &&-\delta_{m,1}\sum_{s=1}^M\Res_{z\to a_s}dS_{a_s,z}(z_1)dx(z)[\hbar^{2h}]\left(\frac{y_{a_s}}{\mathcal{S}(y_{a_s}^{-1}\hbar \partial_x)}\ln(z-a_s) \right)
    \eea  
Thus, applying the variation and using the fact that the Bergmann kernel and $x$ do not vary, we get
\beq d_{a_r}[\omega_{h,m}(z_1,\dots,z_m) ]=(A)+(B)+(C)\eeq
with
\footnotesize{\bea (A)&:=&\sum_{i=1}^N\Res_{z\to p_i}\frac{dE_{i,z}(z_1) (\Omega_{a_r}(z) -\Omega_{a_r}(\sigma_i(z))da_r}{\omega_i(z)^2}\cr&&\Big(\omega_{h-1,m+1}(z,\sigma_i(z),z_2,\dots,z_m)+\sum_{\substack{h_1+h_2=h\\I_1\sqcup I_2=\{2,\dots, m\} \\(h_i,|I_i|)\neq (0,0)}} \omega_{h_1,|I_1|+1}(z,z_{I_1}) \omega_{h_2,|I_2|+1}(\sigma_i(z),z_{I_2}) \Big)\cr
(B)&:=& - \sum_{i=1}^N\Res_{z\to p_i}\frac{dE_{i,z}(z_1)}{\omega_i(z)}d_{a_r}\Big(\omega_{h-1,m+1}(z,\sigma_i(z),z_2,\dots,z_m)+\sum_{\substack{h_1+h_2=h\\I_1\sqcup I_2=\{2,\dots, m\} \\(h_i,|I_i|)\neq (0,0)}} \omega_{h_1,|I_1|+1}(z,z_{I_1}) \omega_{h_2,|I_2|+1}(\sigma_i(z),z_{I_2}) \Big)\cr
(C)&:=& \delta_{m,1}\Res_{z\to a_r}B(a_r,z_1)dx(z)[\hbar^{2h}]\left(\frac{y_{a_r}}{\mathcal{S}(y_{a_r}^{-1}\hbar \partial_x)}\ln(z-a_r) \right)\cr&&+ \delta_{m,1}da_r\Res_{z\to a_r}dS_{a_r,z}(z_1)dx(z)[\hbar^{2h}]\left(\frac{y_{a_r}}{\mathcal{S}(y_{a_r}^{-1}\hbar \partial_x)}\frac{1}{z-a_r} \right)
\eea}
\normalsize{}
We now apply the induction results on the terms in $(B)$. We find:
\small{\bea \frac{(B)}{da_r}&=&(B_1)+ (B_2) \,\text{ with }\cr
(B_1)&:=&-y_{a_r}\sum_{i=1}^N\sum_{j=1}^N\Res_{z\to p_i}\Res_{s\to p_j}\frac{dE_{i,z}(z_1)}{\omega_i(z)}
d_{a_r}[\Phi_{p_i}(s)]
\Big[\omega_{h-1,m+2}(z,\sigma_i(z),z_2,\dots,z_m,s)\cr&&+\sum_{\substack{h_1+h_2=h\\I_1\sqcup I_2=\{2,\dots, m\} \\(h_i,|I_i|)\neq (0,0)}} \omega_{h_1,|I_1|+2}(z,z_{I_1},s) \omega_{h_2,|I_2|+1}(\sigma_i(z),z_{I_2}) \cr&& 
+\sum_{\substack{h_1+h_2=h\\I_1\sqcup I_2=\{2,\dots, m\} \\(h_i,|I_i|)\neq (0,0)}} \omega_{h_1,|I_1|+1}(z,z_{I_1}) \omega_{h_2,|I_2|+2}(\sigma_i(z),z_{I_2},s)
\Big]\cr
(B_2)&:=&-y_{a_r}\sum_{i=1}^N\Res_{z\to p_i}\frac{dE_{i,z}(z_1)}{\omega_i(z)}\Big[\cr&&
\sum_{k=1}^{h-1}[\hbar^{2k}]\frac{1}{\mathcal{S}(y_{a_r}^{-1}\hbar)}dx(q)\left(\frac{\partial^{2k-1}}{\partial x(q)^{2k-1}}\frac{\omega_{h-1-k,m+2}(z,\sigma_i(z),z_2,\dots,z_m,q)}{dx(q)}\right)_{|\, q=a_r}\cr&&
+\sum_{\substack{h_1+h_2=h\\I_1\sqcup I_2=\{2,\dots, m\} \\(h_i,|I_i|)\neq (0,0)}} \sum_{k=1}^{h_1}[\hbar^{2k}]\frac{1}{\mathcal{S}(y_{a_r}^{-1}\hbar)}\left(dx(q)\frac{\partial^{2k-1}}{\partial x(q)^{2k-1}}\frac{\omega_{h_1-k,|I_1|+2}(z,z_{I_1},q)}{dx(q)}\right)_{|\, q=a_r} \omega_{h_2,|I_2|+1}(\sigma_i(z),z_{I_2})\cr&&
+ \sum_{\substack{h_1+h_2=h\\I_1\sqcup I_2=\{2,\dots, m\} \\(h_i,|I_i|)\neq (0,0)}} \omega_{h_1,|I_1|+1}(z,z_{I_1})\sum_{k=1}^{h_2}[\hbar^{2k}]\frac{1}{\mathcal{S}(y_{a_r}^{-1}\hbar)}\left(dx(q)\frac{\partial^{2k-1}}{\partial x(q)^{2k-1}}\frac{\omega_{h_2-k,|I_2|+2}(\sigma_{i}(z),z_{I_2},q)}{dx(q)}\right)_{|\, q=a_r} 
\Big]\cr&&
\eea}
\normalsize{Let} us focus first on $(B_1)$. We exchange the residues using \eqref{ResidueExchange}:
\bea (B_1)&=&-y_{a_r}\sum_{j=1}^N\Res_{s\to p_j} 
d_{a_r}[\Phi_{p_i}(s)]\sum_{i=1}^N\Res_{z\to p_i}\frac{dE_{i,z}(z_1)}{\omega_i(z)}\Big[\omega_{h-1,m+2}(z,\sigma_i(z),z_2,\dots,z_m,s)\cr&&+\sum_{\substack{h_1+h_2=h\\I_1\sqcup I_2=\{2,\dots, m\} \\(h_i,|I_i|)\neq (0,0)}} \omega_{h_1,|I_1|+2}(z,z_{I_1},s) \omega_{h_2,|I_2|+1}(\sigma_i(z),z_{I_2}) \cr&& 
+\sum_{\substack{h_1+h_2=h\\I_1\sqcup I_2=\{2,\dots, m\} \\(h_i,|I_i|)\neq (0,0)}} \omega_{h_1,|I_1|+1}(z,z_{I_1}) \omega_{h_2,|I_2|+2}(\sigma_i(z),z_{I_2},s)
\Big]\cr&&
- y_{a_r}\sum_{i=1}^N\Res_{z\to p_i}\Res_{s\to z,\sigma_{i}(z)}\frac{dE_{i,z}(z_1)}{\omega_i(z)}
d_{a_r}[\Phi_{p_i}(s)]\Big[\omega_{h-1,m+2}(z,\sigma_i(z),z_2,\dots,z_m,s)\cr&&+\sum_{\substack{h_1+h_2=h\\I_1\sqcup I_2=\{2,\dots, m\} \\(h_i,|I_i|)\neq (0,0)}} \omega_{h_1,|I_1|+2}(z,z_{I_1},s) \omega_{h_2,|I_2|+1}(\sigma_i(z),z_{I_2}) \cr&& 
+\sum_{\substack{h_1+h_2=h\\I_1\sqcup I_2=\{2,\dots, m\} \\(h_i,|I_i|)\neq (0,0)}} \omega_{h_1,|I_1|+1}(z,z_{I_1}) \omega_{h_2,|I_2|+2}(\sigma_i(z),z_{I_2},s)
\Big]
\eea
Let us now observe that the residues at $s=z$ or $z=\sigma_i(z)$ do not provide any contribution. Indeed, the only cases where they should contribute would be $\omega_{0,2}(z,s)$ or $\omega_{0,2}(\sigma_i(z),s)$. However, these terms never appear in the sum because we exclude $(h_i,|I_i|)=(0,0)$ so that $\omega_{h_i,|I_i|+2}$ can never be equal to $\omega_{0,2}$. Thus, we get:
\bea (B_1)&=&-y_{a_r}\sum_{j=1}^N\Res_{s\to p_j} 
d_{a_r}[\Phi_{p_i}(s)]\sum_{i=1}^N\Res_{z\to p_i}\frac{dE_{i,z}(z_1)}{\omega_i(z)}\Big[\omega_{h-1,m+2}(z,\sigma_i(z),z_2,\dots,z_m,s)\cr&&+\sum_{\substack{h_1+h_2=h\\I_1\sqcup I_2=\{2,\dots, m\} \\(h_i,|I_i|)\neq (0,0)}} \omega_{h_1,|I_1|+2}(z,z_{I_1},s) \omega_{h_2,|I_2|+1}(\sigma_i(z),z_{I_2}) \cr&& 
+\sum_{\substack{h_1+h_2=h\\I_1\sqcup I_2=\{2,\dots, m\} \\(h_i,|I_i|)\neq (0,0)}} \omega_{h_1,|I_1|+1}(z,z_{I_1}) \omega_{h_2,|I_2|+2}(\sigma_i(z),z_{I_2},s)
\Big]
\eea

Since $m+1\geq 2$, we have by definition (there are no special LogTR terms):
\bea &&\omega_{h,m+1}(z_1,\dots,z_m,s):=-\sum_{i=1}^N\Res_{z\to p_i}\frac{dE_{i,z}(z_1)}{\omega_i(z)}\Big(\omega_{h-1,m+2}(z,\sigma_i(z),z_2,\dots,z_m,s)\cr
    &&+\sum_{\substack{h_1+h_2=h\\I_1\sqcup I_2=\{2,\dots, m,s\} \\(h_i,|I_i|)\neq (0,0)}} \omega_{h_1,|I_1|+1}(z,z_{I_1}) \omega_{h_2,|I_2|+1}(\sigma_i(z),z_{I_2}) \Big)\cr
&&= -\sum_{i=1}^N\Res_{z\to p_i}\frac{dE_{i,z}(z_1)}{\omega_i(z)}\Big(\omega_{h-1,m+2}(z,\sigma_i(z),z_2,\dots,z_m,s)\cr
    &&+\sum_{\substack{h_1+h_2=h\\I_1\sqcup I_2=\{2,\dots, m\} \\(h_2,|I_2|)\neq (0,0)}} \omega_{h_1,|I_1|+1}(z,z_{I_1},s) \omega_{h_2,|I_2|+1}(\sigma_i(z),z_{I_2}) \cr
    &&+\sum_{\substack{h_1+h_2=h\\I_1\sqcup I_2=\{2,\dots, m\} \\(h_1,|I_1|)\neq (0,0)}} \omega_{h_1,|I_1|+1}(z,z_{I_1},s) \omega_{h_2,|I_2|+1}(\sigma_i(z),z_{I_2},s)  \Big)
\eea 
\sloppy{Thus in $(B_1)$ we lack the extra-terms $\omega_{0,2}(z,s)\omega_{h,m}(\sigma_i(z),z_2,\dots,z_m)$ and $\omega_{0,2}(\sigma_i(z),s)\omega_{h,m}(z,z_2,\dots,z_m)$.}
\bea (B_1)&=&y_{a_r}\sum_{j=1}^N\Res_{s\to p_j} 
d_{a_r}[\Phi_{p_i}(s)]\omega_{h,m+1}(z_1,\dots,z_m,s) \cr
&&+ y_{a_r}\sum_{j=1}^N\Res_{s\to p_j} 
d_{a_r}[\Phi_{p_i}(s)]\sum_{i=1}^N\Res_{z\to p_i}\frac{dE_{i,z}(z_1)}{\omega_i(z)}\omega_{0,2}(z,s)\omega_{h,m}(\sigma_i(z),z_2,\dots,z_m)\cr&&
+ y_{a_r}\sum_{j=1}^N\Res_{s\to p_j} 
d_{a_r}[\Phi_{p_i}(s)]\sum_{i=1}^N\Res_{z\to p_i}\frac{dE_{i,z}(z_1)}{\omega_i(z)}\omega_{0,2}(\sigma_i(z),s)\omega_{h,m}(z,z_2,\dots,z_m)\cr
&\overset{z\to \sigma_i(z)}{=}&y_{a_r}\sum_{j=1}^N\Res_{s\to p_j} 
d_{a_r}[\Phi_{p_i}(s)]\omega_{h,m+1}(z_1,\dots,z_m,s) \cr
&&+ 2y_{a_r}\sum_{j=1}^N\Res_{s\to p_j} 
d_{a_r}[\Phi_{p_i}(s)]\sum_{i=1}^N\Res_{z\to p_i}\frac{dE_{i,z}(z_1)}{\omega_i(z)}\omega_{0,2}(z,s)\omega_{h,m}(\sigma_i(z),z_2,\dots,z_m)\cr&&
\eea
Let us again exchange the residues in the last term using \eqref{ResidueExchange}. \\
We observe that $\underset{s\to p_j}{\Res}\, \omega_{0,2}(z,s) 
d_{a_r}[\Phi_{p_i}(s)]=0$, because the integrand is regular at $s=p_j$. The residue at $s=z$ simply gives the derivative with respect to $s$ of the integral. Thus, we get:
\bea \label{B1Expression} (B_1)&=&y_{a_r}\sum_{j=1}^N\Res_{s\to p_j} 
d_{a_r}[\Phi_{p_i}(s)]\omega_{h,m+1}(z_1,\dots,z_m,s) \cr
&&+ 2y_{a_r}\sum_{j=1}^N\Res_{z\to p_j} 
d_{a_r}[ydx(z)]\frac{dE_{j,z}(z_1)}{\omega_j(z)}\omega_{h,m}(\sigma_j(z),z_2,\dots,z_m)
\eea

\medskip
Let us now look at term $(A)$. We shall denote
\beq f_{h,m}(z,z',z_2,\dots,z_m):=\omega_{h-1,m+1}(z,z',z_2,\dots,z_m)+\sum_{\substack{h_1+h_2=h\\I_1\sqcup I_2=\{2,\dots, m\} \\(h_i,|I_i|)\neq (0,0)}} \omega_{h_1,|I_1|+1}(z,z_{I_1}) \omega_{h_2,|I_2|+1}(z',z_{I_2})\eeq
the symmetric form appearing on the r.h.s. We have
\bea (A)
&=&\sum_{i=1}^N\Res_{z\to p_i}\frac{dE_{i,z}(z_1) (\Omega_{a_r}(z) -\Omega_{a_r}(\sigma_i(z))da_r}{\omega_i(z)^2} f_{h,m}(z,\sigma_i(z),z_2,\dots,z_m)\cr
&=&2\sum_{i=1}^N\Res_{z\to p_i}\frac{dE_{i,z}(z_1) \Omega_{a_r}(z)da_r}{\omega_i(z)^2} f_{h,m}(z,\sigma_i(z),z_2,\dots,z_m)\cr
&=&2\sum_{j=1}^N\Res_{z\to p_j}\Res_{s\to z}\frac{dE_{j,s}(z_1)\Omega_{a_r}(s)da_r dE_{j,z}(s)}{\omega_j(z)\omega_{j}(s)}f_{h,m}(z,\sigma_j(z),z_2,\dots,z_m)
\eea

Rauch variational formula \eqref{VariationdE} implies that
\bea 0&=&\sum_{i=1}^N \Res_{z\to p_i}\frac{d_{a_r}[dE_{i,z}(z_1)]}{\omega_i(z)}f_{h,m}(z,\sigma_i(z),z_2,\dots,z_m)\cr
&=&\sum_{i=1}^N\Res_{z\to p_i}\sum_{j=1}^N \Res_{s\to p_j} \frac{dE_{j,s}(z_1)}{\omega_{j}(s)\omega_{i}(z)}\left(\Omega_{a_r}(s)dE_{i,z}(\sigma_j(s))+\Omega_{a_r}(\sigma_j(s))dE_{i,z}(s)\right)da_r f_{h,m}(z,\sigma_i(z),z_2,\dots,z_m)\cr
&=&2\sum_{i=1}^N\Res_{z\to p_i}\sum_{j=1}^N \Res_{s\to p_j} \frac{dE_{j,s}(z_1)}{\omega_{j}(s)\omega_{i}(z)}\Omega_{a_r}(s)dE_{i,z}(\sigma_j(s))da_r f_{h,m}(z,\sigma_i(z),z_2,\dots,z_m)
\eea
Exchanging residues using \eqref{ResidueExchange} gives
\bea 0&=& 2\sum_{j=1}^N\Res_{s\to p_j} \sum_{i=1}^N\Res_{z\to p_i}  \frac{dE_{j,s}(z_1)}{\omega_{j}(s)\omega_{i}(z)}\Omega_{a_r}(s)dE_{i,z}(\sigma_j(s))da_r f_{h,m}(z,\sigma_i(z),z_2,\dots,z_m)\cr&&
+2\sum_{j=1}^N\Res_{z\to p_j} \Res_{s\to z} \frac{dE_{j,s}(z_1)}{\omega_{j}(s)\omega_{j}(z)}\Omega_{a_r}(s)dE_{j,z}(\sigma_j(s))da_r f_{h,m}(z,\sigma_i(z),z_2,\dots,z_m)
\eea
so that
\bea (A)&=&2\sum_{j=1}^N\Res_{s\to p_j} \sum_{i=1}^N\Res_{z\to p_i}  \frac{dE_{j,s}(z_1)}{\omega_{j}(s)\omega_{i}(z)}\Omega_{a_r}(s)dE_{i,z}(\sigma_j(s))da_r f_{h,m}(z,\sigma_i(z),z_2,\dots,z_m)\cr
&=&2\sum_{j=1}^N\Res_{s\to p_j}\frac{dE_{j,s}(z_1)\Omega_{a_r}(s)da_r}{\omega_{j}(s)} \sum_{i=1}^N\Res_{z\to p_i}  \frac{dE_{i,z}(\sigma_j(s))}{\omega_{i}(z)} f_{h,m}(z,\sigma_i(z),z_2,\dots,z_m)\cr&&
\eea
By definition, we have
\bea \omega_{h,m}(s,z_2,\dots,z_m)&=&-\sum_{i=1}^N\Res_{z\to p_i}\frac{dE_{i,z}(s)}{\omega_i(z)} f_{h,m}(z,\sigma_i(z),z_2,\dots,z_m)\cr
&&-\delta_{m,1}\sum_{k=1}^M\Res_{z\to a_k}\left(\int_{a_k}^z\omega_{0,2}(s,.)\right)dx(z)[\hbar^{2h}]\left(\frac{y_{a_k}}{\mathcal{S}(y_{a_k}^{-1}\hbar \partial_x)}\ln(z-a_k) \right)\cr&&
\eea
Thus we get
\bea (A)&:=&(A_1)+(A_2)\,\text{ with }\cr
(A_1)&=&-2\sum_{j=1}^N\Res_{s\to p_j}\frac{dE_{j,s}(z_1)\Omega_{a_r}(s)da_r}{\omega_{j}(s)}\omega_{h,m}(\sigma_j(s),z_2,\dots,z_m)\cr
(A_2)&:=& -2\delta_{m,1}\sum_{j=1}^N\Res_{s\to p_j}\frac{dE_{j,s}(z_1)\Omega_{a_r}(s)da_r}{\omega_{j}(s)}\cr&&\left(\sum_{k=1}^M\Res_{z\to a_k}\left(\int_{a_k}^z\omega_{0,2}(\sigma_j(s),.)\right)dx(z)[\hbar^{2h}]\left(\frac{y_{a_k}}{\mathcal{S}(y_{a_k}^{-1}\hbar \partial_x)}\ln(z-a_k) \right)\right)
\eea
Since $\Omega_{a_r}(s)=y_{a_r}d_{a_r}[y(s)] dx(s)$, we have: 
\beq (A_1)=-2y_{a_r}\sum_{j=1}^N\Res_{s\to p_j}\frac{dE_{j,s}(z_1) 
d_{a_r}[y(s)]dx(s)}{\omega_{j}(s)}\omega_{h,m}(\sigma_j(s),z_2,\dots,z_m)
\eeq
Thus, using the expression of $(B_1)$ given by \eqref{B1Expression}, we find:
\beq \label{ExpressionAPlusB1} (A)+(B_1)da_r=y_{a_r}da_r\sum_{j=1}^N\Res_{s\to p_j} 
d_{a_r}[\Phi_{p_i}(s)]
\omega_{h,m+1}(z_1,\dots,z_m,s) \eeq 

\medskip
At this point, the term $(A)+(B_1)da_r$ corresponds to the first part of the induction. In order to prove the induction, we need to prove that
\beq \label{SecondPartProof}(A_2)+(B_2)da_r+ (C)=y_{a_r}da_r\sum_{k=1}^h [\hbar^{2k}]\frac{1}{\mathcal{S}\left(y_{a_r}^{-1}\hbar\right)}\left(dx(q)\frac{\partial^{2k-1}}{\partial x(q)^{2k-1}}\frac{\omega_{h-k,m+1}(z_1,\dots,z_m,q)}{dx(q)}\right)_{|\, q=a_r}\eeq
where
\small{\bea (C)&:=& \delta_{m,1}\Res_{z\to a_r}B(a_r,z_1)dx(z)[\hbar^{2h}]\left(\frac{y_{a_r}}{\mathcal{S}(y_{a_r}^{-1}\hbar \partial_x)}\ln(z-a_r) \right)\cr&&+ \delta_{m,1}da_r\Res_{z\to a_r}dS_{a_r,z}(z_1)dx(z)[\hbar^{2h}]\left(\frac{y_{a_r}}{\mathcal{S}(y_{a_r}^{-1}\hbar \partial_x)}\frac{1}{z-a_r} \right)\cr
(A_2)&:=& -2\delta_{m,1}\sum_{j=1}^N\Res_{s\to p_j}\frac{dE_{j,s}(z_1)\Omega_{a_r}(s)da_r}{\omega_{j}(s)}\cr&&\left(\sum_{k=1}^M\Res_{z\to a_k}\left(\int_{a_k}^z\omega_{0,2}(\sigma_j(s),.)\right)dx(z)[\hbar^{2h}]\left(\frac{y_{a_k}}{\mathcal{S}(y_{a_k}^{-1}\hbar \partial_x)}\ln(z-a_k) \right)\right)\cr
(B_2)&:=&-y_{a_r}\sum_{i=1}^N\Res_{z\to p_i}\frac{dE_{i,z}(z_1)}{\omega_i(z)}\Big[\cr&&
\sum_{k=1}^{h-1}[\hbar^{2k}]\frac{1}{\mathcal{S}(y_{a_r}^{-1}\hbar)}\left(dx(q)\frac{\partial^{2k-1}}{\partial x(q)^{2k-1}}\frac{\omega_{h-1-k,m+2}(z,\sigma_i(z),z_2,\dots,z_m,q)}{dx(q)}\right)_{|\, q=a_r}\cr&&
+\sum_{\substack{h_1+h_2=h\\I_1\sqcup I_2=\{2,\dots, m\} \\(h_i,|I_i|)\neq (0,0)}} \sum_{k=1}^{h_1}[\hbar^{2k}]\frac{1}{\mathcal{S}(y_{a_r}^{-1}\hbar)}\left(dx(q)\frac{\partial^{2k-1}}{\partial x(q)^{2k-1}}\frac{\omega_{h_1-k,|I_1|+2}(z,z_{I_1},q)}{dx(q)}\right)_{|\, q=a_r} \omega_{h_2,|I_2|+1}(\sigma_i(z),z_{I_2})\cr&&
+ \sum_{\substack{h_1+h_2=h\\I_1\sqcup I_2=\{2,\dots, m\} \\(h_i,|I_i|)\neq (0,0)}} \omega_{h_1,|I_1|+1}(z,z_{I_1})\sum_{k=1}^{h_2}[\hbar^{2k}]\frac{1}{\mathcal{S}(y_{a_r}^{-1}\hbar)}\left(dx(q)\frac{\partial^{2k-1}}{\partial x(q)^{2k-1}}\frac{\omega_{h_2-k,|I_2|+2}(\sigma_{i}(z),z_{I_2},q)}{dx(q)}\right)_{|\, q=a_r} 
\Big]\cr&&
\eea
}
\normalsize{}In the third line, note that $(h_1,|I_1|)=(0,0)$  can be added since it does not contribute because the sum over $k$ is empty. Moreover, the case $h_2=h$ provides an empty sum over $k$. Similarly, in the fourth line, the term $(h_2,|I_2|)=(0,0)$ can be added and the case $h_1=h$ also provides an empty sum over $k$. Thus, we can rewrite $(B_2)$ as:
\small{\bea\label{RformulationB2} &&(B_2)=-
y_{a_r}\sum_{i=1}^N\Res_{z\to p_i}\frac{dE_{i,z}(z_1)}{\omega_i(z)}\Big[\cr&&
\sum_{k=1}^{h-1}[\hbar^{2k}]\frac{1}{\mathcal{S}(y_{a_r}^{-1}\hbar)}\left(dx(q)\frac{\partial^{2k-1}}{\partial x(q)^{2k-1}}\frac{\omega_{h-1-k,m+2}(z,\sigma_i(z),z_2,\dots,z_m,q)}{dx(q)}\right)_{|\, q=a_r}\cr&&
+\sum_{\substack{0\leq h_2\leq h-1\\I_1\sqcup I_2=\{2,\dots, m\} \\(h_2,|I_2|)\neq (0,0)}} \sum_{k=1}^{h-h_2}[\hbar^{2k}]\frac{1}{\mathcal{S}(y_{a_r}^{-1}\hbar)}\left(dx(q)\frac{\partial^{2k-1}}{\partial x(q)^{2k-1}}\frac{\omega_{h-h_2-k,|I_1|+2}(z,z_{I_1},q)}{dx(q)}\right)_{|\, q=a_r} \omega_{h_2,|I_2|+1}(\sigma_i(z),z_{I_2})\cr&&
+ \sum_{\substack{0\leq h_1\leq h-1\\I_1\sqcup I_2=\{2,\dots, m\} \\(h_1,|I_1|)\neq (0,0)}} \omega_{h_1,|I_1|+1}(z,z_{I_1})\sum_{k=1}^{h_2}[\hbar^{2k}]\frac{1}{\mathcal{S}(y_{a_r}^{-1}\hbar)}\left(dx(q)\frac{\partial^{2k-1}}{\partial x(q)^{2k-1}}\frac{\omega_{h-h_1-k,|I_2|+2}(\sigma_{i}(z),z_{I_2},q)}{dx(q)}\right)_{|\, q=a_r} 
\Big]\cr&&
\eea}

\normalsize{Let} us start from the expected result. Since $m+1\geq 2$ (there are no special LogTR terms), we have for $(h-k,m+1)\neq (0,2)$:
\small{\bea 
&&\frac{\omega_{h-k,m+1}(z_1,\dots,z_m,q)}{dx(q)}=-\sum_{i=1}^N\Res_{z\to p_i}\frac{dE_{i,z}(z_1)}{\omega_i(z)}\Big(\frac{\omega_{h-k-1,m+2}(z,\sigma_i(z),z_2,\dots,z_m,q)}{dx(q)}\cr
    &&+\sum_{\substack{h_1+h_2=h-k\\I_1\sqcup I_2=\{2,\dots, m\} \\(h_2,|I_2|)\neq (0,0)}} \frac{\omega_{h_1,|I_1|+1}(z,z_{I_1},q)}{dx(q)} \omega_{h_2,|I_2|+1}(\sigma_i(z),z_{I_2})+\sum_{\substack{h_1+h_2=h-k\\I_1\sqcup I_2=\{2,\dots, m\} \\(h_1,|I_1|)\neq (0,0)}} \omega_{h_1,|I_1|+1}(z,z_{I_1}) \frac{\omega_{h_2,|I_2|+1}(\sigma_i(z),z_{I_2},q)}{dx(q)} \Big)\cr&&
\eea}
\normalsize{Thus}, the r.h.s. $R$ of \eqref{SecondPartProof} is for $m\neq 1$
\footnotesize{\bea \label{Casemneq1} &&\frac{R}{da_r}\delta_{m\neq1}=-
\delta_{m\neq1}y_{a_r}\sum_{i=1}^N\Res_{z\to p_i}\frac{dE_{i,z}(z_1)}{\omega_i(z)} \Big[\sum_{k=1}^{h} [\hbar^{2k}]\frac{1}{\mathcal{S}\left(y_{a_r}^{-1}\hbar\right)}\left(dx(q)\frac{\partial^{2k-1}}{\partial x(q)^{2k-1}}\frac{\omega_{h-k-1,m+2}(z,\sigma_i(z),z_2,\dots,z_m,q)}{dx(q)}\right)_{|\, q=a_r}\cr&&
+ \sum_{k=1}^h [\hbar^{2k}]\sum_{\substack{h_1+h_2=h-k\\I_1\sqcup I_2=\{2,\dots, m\} \\(h_2,|I_2|)\neq (0,0)}}\omega_{h_2,|I_2|+1}(\sigma_i(z),z_{I_2})\frac{1}{\mathcal{S}\left(y_{a_r}^{-1}\hbar\right)}\left(dx(q)\frac{\partial^{2k-1}}{\partial x(q)^{2k-1}} \frac{\omega_{h_1,|I_1|+2}(z,z_{I_1},q)}{d(q)}\right)_{|\, q=a_r}\cr&&
+\sum_{k=1}^h [\hbar^{2k}]\sum_{\substack{h_1+h_2=h-k\\I_1\sqcup I_2=\{2,\dots, m\} \\(h_1,|I_1|)\neq (0,0)}} \omega_{h_1,|I_1|+1}(z,z_{I_1}) \frac{1}{\mathcal{S}\left(y_{a_r}^{-1}\hbar\right)}\left(dx(q)\frac{\partial^{2k-1}}{\partial x(q)^{2k-1}}\frac{\omega_{h_2,|I_2|+2}(\sigma_i(z),z_{I_2},q)}{dx(q)}\right)_{|\, q=a_r}\Big]
\eea}
\normalsize{while} it is for $m=1$:
\bea\label{Casemequal1}&&\frac{R}{da_r}\delta_{m=1}= \delta_{m=1}y_{a_r}da_r  [\hbar^{2h}]\frac{1}{\mathcal{S}\left(y_{a_r}^{-1}\hbar\right)}\left(dx(q)\frac{\partial^{2h-1}}{\partial x(q)^{2h-1}}\frac{\omega_{0,2}(z_1,q)}{dx(q)}\right)_{|\, q=a_r}\cr&& 
-y_{a_r}\delta_{m=1}\sum_{i=1}^N\Res_{z\to p_i}\frac{dE_{i,z}(z_1)}{\omega_i(z)} \Big[\sum_{k=1}^{h-1} [\hbar^{2k}]\frac{1}{\mathcal{S}\left(y_{a_r}^{-1}\hbar\right)}\left(dx(q)\frac{\partial^{2k-1}}{\partial x(q)^{2k-1}}\frac{\omega_{h-k-1,3}(z,\sigma_i(z),q)}{dx(q)}\right)_{|\, q=a_r}\cr&&
+ \sum_{k=1}^{h-1} [\hbar^{2k}]\sum_{\substack{h_1+h_2=h-k\\h_2\neq 0}}\omega_{h_2,1}(\sigma_i(z))\frac{1}{\mathcal{S}\left(y_{a_r}^{-1}\hbar\right)}\left(dx(q)\frac{\partial^{2k-1}}{\partial x(q)^{2k-1}} \frac{\omega_{h_1,2}(z,q)}{d(q)}\right)_{|\, q=a_r}\cr&&
+\sum_{k=1}^{h-1} [\hbar^{2k}]\sum_{\substack{h_1+h_2=h-k \\h_1\neq 0}} \omega_{h_1,1}(z) \frac{1}{\mathcal{S}\left(y_{a_r}^{-1}\hbar\right)}\left(dx(q)\frac{\partial^{2k-1}}{\partial x(q)^{2k-1}}\frac{\omega_{h_2,2}(\sigma_i(z),q)}{dx(q)}\right)_{|\, q=a_r}
\Big]
\eea

In the first sum of \eqref{Casemneq1}, the case $k=h$ gives $\omega_{-1,m+2}=0$ so it is null by definition. Thus, we get:
\footnotesize{\bea \label{Rexpression}&&\frac{R}{da_r}\delta_{m\neq 1}=-y_{a_r}\delta_{m\neq 1}\sum_{i=1}^N\Res_{z\to p_i}\frac{dE_{i,z}(z_1)}{\omega_i(z)} \Big[\sum_{k=1}^{h-1} [\hbar^{2k}]\frac{1}{\mathcal{S}\left(y_{a_r}^{-1}\hbar\right)}\left(dx(q)\frac{\partial^{2k-1}}{\partial x(q)^{2k-1}}\frac{\omega_{h-k-1,m+2}(z,\sigma_i(z),z_2,\dots,z_m,q)}{dx(q)}\right)_{|\, q=a_r}\cr&&
+ \sum_{k=1}^h [\hbar^{2k}]\sum_{\substack{0\leq h_2\leq h-k\\I_1\sqcup I_2=\{2,\dots, m\} \\(h_2,|I_2|)\neq (0,0)}}\omega_{h_2,|I_2|+1}(\sigma_i(z),z_{I_2})\frac{1}{\mathcal{S}\left(y_{a_r}^{-1}\hbar\right)}\left(dx(q)\frac{\partial^{2k-1}}{\partial x(q)^{2k-1}} \frac{\omega_{h-k-h_2,|I_1|+2}(z,z_{I_1},q)}{d(q)}\right)_{|\, q=a_r}\cr&&
+\sum_{k=1}^h [\hbar^{2k}]\sum_{\substack{0\leq h_1 \leq h-k\\I_1\sqcup I_2=\{2,\dots, m\} \\(h_1,|I_1|)\neq (0,0)}} \omega_{h_1,|I_1|+1}(z,z_{I_1}) \frac{1}{\mathcal{S}\left(y_{a_r}^{-1}\hbar\right)}\left(dx(q)\frac{\partial^{2k-1}}{\partial x(q)^{2k-1}}\frac{\omega_{h-k-h_1,|I_2|+2}(\sigma_i(z),z_{I_2},q)}{dx(q)}\right)_{|\, q=a_r}
\Big]\cr&&
\eea}
\normalsize{We} can exchange the sums: $\underset{k=1}{\overset{h}{\sum}}\underset{h_2=0}{\overset{h-k}{\sum}}= \underset{h_2=0}{\overset{h-1}{\sum}}\underset{k=1}{\overset{h-h_2}{\sum}}$. For $|I_2|=0$, we have the reduced sum $\underset{k=1}{\overset{h}{\sum}}\underset{h_2=1}{\overset{h-k}{\sum}}=\underset{h_2=1}{\overset{h-1}{\sum}}\underset{k=1}{\overset{h-h_2}{\sum}}$. Similar identities hold for $h_1$. In particular, after this exchange we observe that the expression of $R$ matches with the expression of $(B_2)$ given by \eqref{RformulationB2}.
Similarly, for $m=1$, we have $\underset{k=1}{\overset{h-1}{\sum}}\underset{h_2=1}{\overset{h-k}{\sum}}=\underset{h_2=1}{\overset{h-1}{\sum}}\underset{k=1}{\overset{h-h_2}{\sum}}$. The first double sum appears in \eqref{Casemequal1}  while the second one appear in \eqref{RformulationB2}. Thus we get:
\beq\label{IdentityRB2} \forall \,m\geq 1\,:\, R=(B_2)da_r +\delta_{m=1}y_{a_r}da_r  [\hbar^{2h}]\frac{1}{\mathcal{S}\left(y_{a_r}^{-1}\hbar\right)}\left(dx(q)\frac{\partial^{2h-1}}{\partial x(q)^{2h-1}}\frac{\omega_{0,2}(z_1,q)}{dx(q)}\right)_{|\, q=a_r}\eeq

In the end, from \eqref{SecondPartProof}, the induction is equivalent to prove that
\beq (A_2)+(C)=\delta_{m=1}y_{a_r}da_r  [\hbar^{2h}]\frac{1}{\mathcal{S}\left(y_{a_r}^{-1}\hbar\right)}\left(dx(q)\frac{\partial^{2h-1}}{\partial x(q)^{2h-1}}\frac{\omega_{0,2}(z_1,q)}{dx(q)}\right)_{|\, q=a_r}\eeq
In $(A_2)$, we can exchange the residues because there are located at different points:
\small{\bea (A_2)&:=&- 2\delta_{m,1}\sum_{j=1}^N\Res_{s\to p_j}\frac{dE_{j,s}(z_1)\Omega_{a_r}(s)da_r}{\omega_{j}(s)}\cr&&\left(\sum_{k=1}^M\Res_{z\to a_k}\left(\int_{a_k}^z\omega_{0,2}(\sigma_j(s),.)\right)dx(z)[\hbar^{2h}]\left(\frac{y_{a_k}}{\mathcal{S}(y_{a_k}^{-1}\hbar \partial_x)}\ln(z-a_k) \right)\right)\cr
&=&- 2\delta_{m,1}\sum_{k=1}^M\Res_{z\to a_k}dx(z)[\hbar^{2h}]\left(\frac{y_{a_k}}{\mathcal{S}(y_{a_k}^{-1}\hbar \partial_x)}\ln(z-a_k) \right)\sum_{j=1}^N\Res_{s\to p_j}\frac{dE_{j,s}(z_1)\Omega_{a_r}(s)da_r}{\omega_{j}(s)}dS_{a_k,z}(\sigma_j(s))\cr&&
\eea}
\normalsize{Using} \eqref{ExprOmegaar}, it gives
\small{\beq (A_2)= -2y_{a_r}\delta_{m,1}\sum_{k=1}^M\Res_{z\to a_k}dx(z)[\hbar^{2h}]\left(\frac{y_{a_k}}{\mathcal{S}(y_{a_k}^{-1}\hbar \partial_x)}\ln(z-a_k) \right)\sum_{j=1}^N\Res_{s\to p_j}\frac{dE_{j,s}(z_1)
d_{a_r}[y(s)]dx(s)}{\omega_{j}(s)}dS_{a_k,z}(\sigma_j(s))
\eeq}
\normalsize{Since} $\omega_j(s)=(y(s)-y(\sigma_j(s)))dx(s)$, it has a double zero at $s=p_j$. Similarly $dx(s)$ has a simple zero and $dE_{j,s}(z_1)=\int_{\sigma_j(s)}^s B(.,z_1)$ as a simple zero at $s=p_j$. Finally 
$d_{a_r}[y(s)]$ and $dS_{a_k,z}(\sigma_j(s))$ are regular at $s=p_j$ so that the integrand is regular at $s=p_j$ and we find
\beq (A_2)=0\eeq
The expression for $(C)$ is
\bea (C)&:=&(C_1)+(C_2) \text{ with }\cr
(C_1)&:=&\delta_{m,1}B(a_r,z_1)\Res_{z\to a_r}dx(z)[\hbar^{2h}]\left(\frac{y_{a_r}}{\mathcal{S}(y_{a_r}^{-1}\hbar \partial_x)}\ln(z-a_r) \right)\cr
(C_2)&:=&\delta_{m,1}da_r\Res_{z\to a_r}dS_{a_r,z}(z_1)dx(z)[\hbar^{2h}]\left(\frac{y_{a_r}}{\mathcal{S}(y_{a_r}^{-1}\hbar \partial_x)}\frac{1}{z-a_r} \right)
\eea
The expression for $(C_1)$ is:
\bea (C_1)&=&\delta_{m,1}B(a_r,z_1)\Res_{z\to a_r}dx(z)\left(\frac{\partial^{2h}}{\partial x(z)^{2h}}\ln(z-a_r)\right) [\hbar^{2h}]\left(\frac{y_{a_r}}{\mathcal{S}(y_{a_r}^{-1}\hbar)}\right) \cr
&=& \delta_{m,1}B(a_r,z_1)\Res_{z\to a_r}\, dx(z) \frac{\partial}{\partial x(z)}\left(\frac{\partial^{2h-1}}{\partial x(z)^{2h-1}}\ln(z-a_r)\right)[\hbar^{2h}]\left(\frac{y_{a_r}}{\mathcal{S}(y_{a_r}^{-1}\hbar)}\right) \cr&&
\eea
Since $dx(z)\frac{\partial}{\partial x(z)}= d_z$ we get
\beq (C_1)=\delta_{m,1}B(a_r,z_1)\Res_{z\to a_r}dz\, \frac{\partial}{\partial z}\left(\frac{\partial^{2h-1}}{\partial x(z)^{2h-1}}\ln(z-a_r)[\hbar^{2h}]\left(\frac{y_{a_r}}{\mathcal{S}(y_{a_r}^{-1}\hbar)}\right)\right)=0 \eeq
because we are computing the residue of the total derivative of a meromorphic form at $z=a_r$. ($2h-1\geq 1$ so $\frac{\partial^{(2h-1)}}{\partial x(z)^{2h-1}}\ln(z-a_r)$ is meromorphic at $z=a_r$).

A similar computation for $(C_2)$ gives
\bea (C_2)&=&\delta_{m,1}da_r\Res_{z\to a_r}dS_{a_r,z}(z_1)dx(z)\left(\frac{\partial^{2h}}{\partial x(z)^{2h}}\frac{1}{z-a_r}\right) [\hbar^{2h}]\left(\frac{y_{a_r}}{\mathcal{S}(y_{a_r}^{-1}\hbar)}\right)\cr
&=&\delta_{m,1}da_r\Res_{z\to a_r}dS_{a_r,z}(z_1)dx(z)\frac{\partial}{\partial x(z)}\left(\frac{\partial^{2h-1}}{\partial x(z)^{2h-1}}\frac{1}{z-a_r}\right) [\hbar^{2h}]\left(\frac{y_{a_r}}{\mathcal{S}(y_{a_r}^{-1}\hbar)}\right)\cr
&=&\delta_{m,1}da_r\Res_{z\to a_r}dz \,dS_{a_r,z}(z_1)\frac{\partial}{\partial z}\left(\frac{\partial^{2h-1}}{\partial x(z)^{2h-1}}\frac{1}{z-a_r}[\hbar^{2h}]\left(\frac{y_{a_r}}{\mathcal{S}(y_{a_r}^{-1}\hbar)}\right)\right) \cr
&=&-\delta_{m,1}da_r[\hbar^{2h}]\left(\frac{y_{a_r}}{\mathcal{S}(y_{a_r}^{-1}\hbar)}\right)\Res_{z\to a_r} \,B(z,z_1)\left(\frac{\partial^{2h-1}}{\partial x(z)^{2h-1}}\frac{1}{z-a_r}\right)
\eea
Performing multiple integrations by parts for local meromorphic functions and computing the residue, we get that
\beq\label{AppFCorTerm} (C_2)=\delta_{m=1}da_r  [\hbar^{2h}]\frac{y_{a_r}}{\mathcal{S}\left(y_{a_r}^{-1}\hbar\right)}\left(dx(q)\frac{\partial^{2h-1}}{\partial x(q)^{2h-1}}\frac{\omega_{0,2}(z_1,q)}{dx(q)}\right)_{|\, q=a_r}.\eeq

Let us bring this term into another more compact form.
We can use the identity
\bea &&[\hbar^{2h}]\frac{y_{a_r}}{\mathcal{S}\left(y_{a_r}^{-1}\hbar\right)}\left(dx(q)\frac{\partial^{2h-1}}{\partial x(q)^{2h-1}}\frac{\omega_{h,k+1}(z,z_1,\dots,z_k)}{dx(q)}\right)_{|\, q=a_r}\cr
&&=-[\hbar^{2h}]\left(\frac{y_{a_r}}{\mathcal{S}(y_{a_r}^{-1}\hbar)}\right)\Res_{z\to a_r} \,\omega_{h,k+1}(z,z_1,\dots,z_k)\left(\frac{\partial^{2h-1}}{\partial x(z)^{2h-1}}\frac{1}{z-a_r}\right)\cr
&&=\Res_{z\to a_r} \,dx(z)\left(\int_o^z\omega_{h,k+1}(.,z_1,\dots,z_k)\right)[\hbar^{2h}]\left(\frac{y_{a_r}}{\mathcal{S}(y_{a_r}^{-1}\hbar \partial_x)}\right)\frac{1}{z-a_r}
\eea
to reformulate \eqref{AppFCorTerm} into
\bea &&d_{a_r}[\omega_{h,n}(z_1,\dots,z_n)] 
= y_{a_r}\sum_{i=1}^N \Res_{q\to p_i} \left(\int^{s=q}_{s=p_i} dS_{o,s}(a_r)dx(s)\right) \omega_{h,n+1}(z_1,\dots,z_n,q)\cr\label{Gnewform}
&&+ da_r\sum_{k=1}^h \Res_{z\to a_r} dx(z)\left(\int_o^z \omega_{h-k,n+1}(.,z_1,\dots,z_n)\right)[\hbar^{2k}]\left(\frac{y_{a_r}}{\mathcal{S}(y_{a_r}^{-1}\hbar \partial_x)}\right)\frac{1}{z-a_r}
\eea

Next we have the following lemma.
\begin{lemma}\label{Lem:trick2}
    Let $F(q)$ be a smooth 1-form at $a$.
    We have the following relations (for any $k\in \mathbb{N}$):
   \bea 
dx(a)\Res_{q\to a} F(q)\left(\frac{\partial}{\partial x(q)}\right)^k \log(q-a) &=&-\Res_{q\to a} \left(\int^q_oF(q)\right)dx(q)\left(\frac{\partial}{\partial x(q)}\right)^k \frac{1}{q-a}\cr
&=&\Res_{q\to a} F(q)\left(\frac{\partial}{\partial x(q)}\right)^{k-1} \frac{1}{q-a}
\eea
and 
\beq \Res_{q\to a} (x(q)-x(a))F(q)\left(\frac{\partial}{\partial x(q)}\right)^k \frac{1}{q-a}=-k\Res_{q\to a} F(q)\left(\frac{\partial}{\partial x(q)}\right)^{k-1} \frac{1}{q-a}.
\eeq
\end{lemma}
\begin{proof}
The first relation follows from the fact $dx(a)\frac{\partial}{\partial x(q)}\log(q-a)-\frac{1}{q-a}$ is regular at $q\to a$. Acting on this with $k-1$ derivative wrt $x(q)$ and multiplying it with $F(q)$ gives a 1-form which is regular at $q\to a$ and has therefore a vanishing residue:
\beq\Res_{q\to a} F(q)\frac{\partial^{k-1}}{\partial x(q)^{k-1}}\left(dx(a)\frac{\partial}{\partial x(q)}\log(q-a)-\frac{1}{q-a}\right)=0.\eeq
The last identity is achieved by integration by parts and the fact that $\frac{x(q)-x(a)}{q-a}$ is regular at $q\to a$. Thus, acting with $k$ derivatives wrt $x(q)$ on it and multiplying with $F(q)$ gives a regular 1-form at $q\to a$, thus vanishing residue.
\beq\Res_{q\to a} F(q)\left(\frac{\partial}{\partial x(q)}\right)^k\frac{x(q)-x(a)}{q-a}=0,
\eeq
where one can perform the $x(q)$ derivative $\left(\frac{\partial}{\partial x(q)}\right)^k\frac{x(q)-x(a)}{q-a}=(x(q)-x(a))\left(\frac{\partial}{\partial x(q)}\right)^k\frac{1}{q-a}+k\left(\frac{\partial}{\partial x(q)}\right)^{k-1}\frac{1}{q-a}$.
\end{proof}

Eventually, taking \eqref{Gnewform}, integrating once by parts and using  the first equation of \autoref{Lem:trick2} with $F(q)=\omega_{h-h_1,n+1}(I,q)$ and the remaining part corresponds to the singular part of $\omega_{h_1,1}(q)$ at the LogTR-vital singularity $a_r$ provides the proof.


\section{Compatibility of the variational formulas for a log pole with the dilaton equations}\label{AppendixCompatVarLogTR}
In this section we prove that \autoref{TheoVariationsLogTRpoles} is compatible with \autoref{TheoremDilatonEquation}. Let $k\geq 1$ and $h\geq 0$ with $(k,h)\neq (1,0)$. The dilaton equation is
\bea -(2h+k-2)\omega_{h,k}(z_1,\dots,z_k)&=&\sum_{i=1}^N\Res_{z\to p_i} \Phi_{p_i}(z)\omega_{h,k+1}(z,z_1,\dots,z_k) \cr&&
-\sum_{j=1}^M\Res_{z\to a_j}\frac{x(z)-x(a_j)}{dx(z)}
        \overset{h}{\underset{h_1=1}{\sum}}\omega_{h_1,1} (z)\omega_{h-h_1,k+1} (z,z_1,\dots,z_k)\cr&&
\eea
Applying $d_{a_r}$ on the l.h.s. provides from \autoref{TheoVariationsLogTRpoles}:
\bea \label{LHSToCheck} &&d_{a_r}[\text{LHS}]:=(2-2h-k)\sum_{i=1}^N \Res_{q\to p_i} d_{a_r}[\Phi_{p_i}(q)] \omega_{h,k+1}(z_1,\dots,z_k,q)\cr
&&+  (2-2h-k)\Res_{q\to a_r} \frac{dx(a_r)}{dx(q)}\sum_{h_1=1}^h \omega_{h_1,1}(q)\omega_{h-h_1,k+1}(q,z_1,\dots,z_k)
\eea
Applying $d_{a_r}$ on the r.h.s. provides
\bea d_{a_r}[\text{RHS}]&:=& (A)+(B)+(C)+(D)+(E) \,\text{ with}\cr
(A)&:=&\sum_{i=1}^N\Res_{z\to p_i} d_{a_r}[\Phi_{p_i}(z)]\omega_{h,k+1}(z,z_1,\dots,z_k) \cr
(B)&:=&\sum_{i=1}^N\Res_{z\to p_i} \Phi_{p_i}(z)d_{a_r}[\omega_{h,k+1}(z,z_1,\dots,z_k)] \cr
(C)&:=&dx(a_r)\Res_{z\to a_r} \frac{1}{dx(z)}\overset{h}{\underset{h_1=1}{\sum}}\omega_{h_1,1} (z)\omega_{h-h_1,k+1} (z,z_1,\dots,z_k)\cr
(D)&:=&-\sum_{j=1}^M\Res_{z\to a_j}\frac{x(z)-x(a_j)}{dx(z)}\overset{h}{\underset{h_1=1}{\sum}}d_{a_r}[\omega_{h_1,1} (z)]\omega_{h-h_1,k+1} (z,z_1,\dots,z_k)\cr
(E)&:=&-\sum_{j=1}^M\Res_{z\to a_j}\frac{x(z)-x(a_j)}{dx(z)}
        \overset{h}{\underset{h_1=1}{\sum}}\omega_{h_1,1} (z)d_{a_r}[\omega_{h-h_1,k+1} (z,z_1,\dots,z_k)]
\eea
Using \autoref{TheoVariationsLogTRpoles}, we can split $(B)$ into two terms
\bea (B)&:=&(B_1)+(B_2) \text{ with}\cr
(B_1)&:=&\sum_{i=1}^N\Res_{z\to p_i} \Phi_{p_i}(z)\sum_{j=1}^N \Res_{q\to p_j} d_{a_r}[\Phi_{p_j}(q)] \omega_{h,k+2}(q,z,z_1,\dots,z_k) \cr
(B_2)&:=&\sum_{i=1}^N\Res_{z\to p_i} \Phi_{p_i}(z)\Res_{q\to a_r}\frac{dx(a_r)}{dx(q)}\sum_{h_1=1}^h \omega_{h_1,1}(q)\omega_{h-h_1,k+2}(q,z,z_1,\dots,z_k) 
\eea 
Since $k+2\geq 3$, we have $\omega_{h,k+2}\neq \omega_{0,2}$ so we can exchange the two residues in $(B_1)$ without having a special term:
\bea 
(B_1)&:=&\sum_{j=1}^N \Res_{q\to p_j} d_{a_r}[\Phi_{p_j}(q)]\sum_{i=1}^N\Res_{z\to p_i} \Phi_{p_i}(z)  \omega_{h,k+2}(q,z,z_1,\dots,z_k) \cr
(B_2)&:=&\Res_{q\to a_r}\frac{dx(a_r)}{dx(q)}\sum_{h_1=1}^h\omega_{h_1,1}(q)\sum_{i=1}^N\Res_{z\to p_i} \Phi_{p_i}(z) \omega_{h-h_1,k+2}(q,z,z_1,\dots,z_k)
\eea
Let us now focus on $(D)$. Using \autoref{TheoVariationsLogTRpoles}, it splits into two parts:
\bea (D)&:=&(D_1)+(D_2)\, \text{ with}\cr
(D_1)&:=& -\sum_{j=1}^M\Res_{z\to a_j}\frac{x(z)-x(a_j)}{dx(z)}\overset{h}{\underset{h_1=1}{\sum}}\omega_{h-h_1,k+1} (z,z_1,\dots,z_k)\sum_{i=1}^N \Res_{q\to p_i} d_{a_r}[\Phi_{p_i}(q)]\omega_{h_1,2}(q,z)\cr
(D_2)&:=&-\sum_{j=1}^M\Res_{z\to a_j}\frac{x(z)-x(a_j)}{dx(z)}\overset{h}{\underset{h_1=1}{\sum}}\omega_{h-h_1,k+1} (z,z_1,\dots,z_k)\Res_{q\to a_r} \frac{dx(a_r)}{dx(q)}\sum_{u=1}^{h_1}\omega_{u,1}(q)\omega_{h_1-u,2}(q,z) \cr&&
\eea
Exchanging the residues in $(D_1)$, we have that all terms are regular at $z=a_j$ because we have $k+1\geq 2$ so that all correlation functions are regular at $z=a_j$ and $x(z)$ is regular at $z=a_j$. Thus, we get $(D_1)=0$. The same argument holds for $(D_2)$ but we get an extra term when $a_j=a_r$. From \eqref{ResidueExchange}, we have
\bea (D)&=&-\Res_{q\to a_r} \Res_{z\to q} \frac{x(z)-x(a_r)}{dx(z)}\overset{h}{\underset{h_1=1}{\sum}}\omega_{h-h_1,k+1} (z,z_1,\dots,z_k)\frac{dx(a_r)}{dx(q)}\sum_{u=1}^{h_1}\omega_{u,1}(q)\omega_{h_1-u,2}(q,z) \cr
&=&-\overset{h}{\underset{h_1=1}{\sum}}\Res_{q\to a_r}\omega_{h_1,1}(q)\frac{dx(a_r)}{dx(q)} \Res_{z\to q} \frac{x(z)-x(a_r)}{dx(z)}\omega_{h-h_1,k+1} (z,z_1,\dots,z_k)\omega_{0,2}(q,z) \cr
&=&-\overset{h}{\underset{h_1=1}{\sum}}\Res_{q\to a_r}\omega_{h_1,1}(q)\frac{dx(a_r)}{dx(q)} d_q\left(\frac{x(q)-x(a_r)}{dx(q)}\omega_{h-h_1,k+1} (q,z_1,\dots,z_k)\right)\cr&&
\eea
Let us now decompose $(E)$. Using \autoref{TheoVariationsLogTRpoles}, we have
\bea (E)&:=&(E_1)+(E_2)\,\text{ with}\cr
(E_1)&:=&-\sum_{j=1}^M\Res_{z\to a_j}\frac{x(z)-x(a_j)}{dx(z)}
\overset{h}{\underset{h_1=1}{\sum}}\omega_{h_1,1} (z)\sum_{i=1}^N\Res_{q\to p_i} d_{a_r}[\Phi_{p_i}(q)]\omega_{h-h_1,k+2}(q,z,z_1,\dots,z_k)\cr
(E_2)&:=&-\sum_{j=1}^M\Res_{z\to a_j}\frac{x(z)-x(a_j)}{dx(z)}
\overset{h}{\underset{h_1=1}{\sum}}\omega_{h_1,1} (z)\Res_{q\to a_r}\frac{dx(a_r)}{dx(q)}\sum_{u=1}^{h-h_1}\omega_{u,1}(q)\omega_{h-h_1-u,k+2}(q,z,q_1,\dots,z_k)\cr&&
\eea
Since $k+2\geq 3$, we can exchange the residues in $(E_2)$:
\bea 
(E_1)&=&-\sum_{i=1}^N\Res_{q\to p_i}d_{a_r}[\Phi_{p_i}(q)] \sum_{j=1}^M\Res_{z\to a_j}\frac{x(z)-x(a_j)}{dx(z)}
\overset{h}{\underset{h_1=1}{\sum}}\omega_{h_1,1} (z) \omega_{h-h_1,k+2}(q,z,z_1,\dots,z_k)\cr
(E_2)&=&-\Res_{q\to a_r}\frac{dx(a_r)}{dx(q)}\overset{h}{\underset{h_1=1}{\sum}}\sum_{u=1}^{h-h_1}\sum_{j=1}^M\Res_{z\to a_j}\omega_{h_1,1} (z)\omega_{u,1}(q) \frac{x(z)-x(a_j)}{dx(z)}
\omega_{h-h_1-u,k+2}(q,z,q_1,\dots,z_k)\cr&&
\eea
We exchange the double sum $\underset{h_1=1}{\overset{h}{\sum}}\underset{u=1}{\overset{h-h_1}{\sum}}=\underset{u=1}{\overset{h-1}{\sum}}\underset{h_1=1}{\overset{h-u}{\sum}}$ so that
\bea (E_2)&=&-\Res_{q\to a_r}\frac{dx(a_r)}{dx(q)}\sum_{u=1}^{h-1}\omega_{u,1}(q)\sum_{h_1=1}^{h-u}\sum_{j=1}^M\Res_{z\to a_j}\omega_{h_1,1} (z) \frac{x(z)-x(a_j)}{dx(z)}
\omega_{h-h_1-u,k+2}(q,z,q_1,\dots,z_k)\cr&&
\eea

Let us now regroup $(A)$, $(B_1)$ and $(E_1)$. We have
\bea &&(A)+(B_1)+(E_1)=\sum_{i=1}^N\Res_{q\to p_i}d_{a_r}[\Phi_{p_i}(q)]\Big[\omega_{h,k+1}(q,z_1,\dots,z_k) \cr&&
-\sum_{j=1}^N\Res_{z\to p_j} \Phi_{p_j}(z)  \omega_{h,k+2}(q,z,z_1,\dots,z_k)+\sum_{j=1}^M\Res_{z\to a_j}\frac{x(z)-x(a_j)}{dx(z)}
\overset{h}{\underset{h_1=1}{\sum}}\omega_{h_1,1} (z) \omega_{h-h_1,k+2}(q,z,z_1,\dots,z_k)\Big]\cr&&
\eea
In the second line, we recognize from \autoref{TheoremDilatonEquation} $(2-2h-k-1)\omega_{h,k+1}(q,z_1,\dots,z_k)$ so that
\beq (A)+(B_1)+(E_1)=(2-2h-k)\sum_{i=1}^N \Res_{q\to p_i}d_{a_r}[\Phi_{p_i}(q)]\omega_{h,k+1}(q,z_1,\dots,z_k)\eeq
which is precisely the first part of \eqref{LHSToCheck}. Thus, we are left with
\bea (B_2)&=&\Res_{q\to a_r}\frac{dx(a_r)}{dx(q)}\sum_{h_1=1}^h\omega_{h_1,1}(q) \sum_{i=1}^N\Res_{z\to p_i} \Phi_{p_i}(z)\omega_{h-h_1,k+2}(q,z,z_1,\dots,z_k)  \cr
(C)&=&dx(a_r)\Res_{z\to a_r} \frac{1}{dx(z)}\overset{h}{\underset{h_1=1}{\sum}}\omega_{h_1,1} (z)\omega_{h-h_1,k+1} (z,z_1,\dots,z_k)\cr
(D)&=&-\overset{h}{\underset{h_1=1}{\sum}}\Res_{q\to a_r}\omega_{h_1,1}(q)\frac{dx(a_r)}{dx(q)} d_q\left(\frac{x(q)-x(a_r)}{dx(q)}\omega_{h-h_1,k+1} (q,z_1,\dots,z_k)\right)\cr
(E_2)&=&-\Res_{q\to a_r}\frac{dx(a_r)}{dx(q)}\sum_{u=1}^{h-1}\omega_{u,1}(q)\sum_{h_1=1}^{h-u}\sum_{j=1}^M\Res_{z\to a_j} \frac{x(z)-x(a_j)}{dx(z)}\omega_{h_1,1} (z)
\omega_{h-h_1-u,k+2}(q,z,q_1,\dots,z_k)\cr
&=&-\Res_{q\to a_r}\frac{dx(a_r)}{dx(q)}\sum_{u=1}^{h}\omega_{u,1}(q)\sum_{h_1=1}^{h-u}\sum_{j=1}^M\Res_{z\to a_j} \frac{x(z)-x(a_j)}{dx(z)}\omega_{h_1,1} (z)
\omega_{h-h_1-u,k+2}(q,z,q_1,\dots,z_k)\cr&&
\eea
where we can add $u=h$ in the last equality because it provides an empty sum. Note that we should recover
\beq (2-2h-k)\Res_{q\to a_r} \frac{dx(a_r)}{dx(q)}\sum_{h_1=1}^h \omega_{h_1,1}(q)\omega_{h-h_1,k+1}(q,z_1,\dots,z_k)\eeq
Let us regroup $(B_2)+(E_2)$. The dilaton equation \autoref{TheoremDilatonEquation} gives:
\small{\bea &&(2-2h+2h_1-k-1)\omega_{h-h_1,k+1}(q,z_1,\dots,z_k)=(2-2h-k+2h_1-1)\omega_{h-h_1,k+1}(q,z_1,\dots,z_k)\cr
&&=\sum_{i=1}^N\Res_{z\to p_i} \Phi_{p_i}(z)\omega_{h-h_1,k+2}(q,z,z_1,\dots,z_k)
-\sum_{j=1}^M\Res_{z\to a_j}\frac{x(z)-x(a_j)}{dx(z)}
        \overset{h-h_1}{\underset{u=1}{\sum}}\omega_{u,1} (z)\omega_{h-h_1-u,k+2} (q,z,z_1,\dots,z_k)\cr&&
\eea}
\normalsize{Thus}, we get:
\beq (B_2)+(E_2)=\Res_{q\to a_r}\frac{dx(a_r)}{dx(q)}\sum_{h_1=1}^h(2-2h-k+2h_1-1)\omega_{h_1,1}(q)\omega_{h-h_1,k+1}(q,z_1,\dots,z_k)
\eeq
Thus, the compatibility of the dilaton equations with the variational formulas with respect to the LogTR-vital singularities are equivalent to prove that
\bea &&\sum_{h_1=1}^h\Res_{q\to a_r}(1-2h_1)\frac{dx(a_r)}{dx(q)}\omega_{h_1,1}(q)\omega_{h-h_1,k+1}(q,z_1,\dots,z_k)= (C)+(D)\cr&&
=\sum_{h_1=1}^h\Res_{q\to a_r} \frac{ dx(a_r)}{dx(q)}\omega_{h_1,1} (q)\omega_{h-h_1,k+1} (q,z_1,\dots,z_k)\cr
&&-\sum_{h_1=1}^h \Res_{q\to a_r}\frac{dx(a_r)}{dx(q)}\omega_{h_1,1}(q) d_q\left(\frac{x(q)-x(a_r)}{dx(q)}\omega_{h-h_1,k+1} (q,z_1,\dots,z_k)\right)
\eea
i.e.
\bea &&\sum_{h_1=1}^h\Res_{q\to a_r}2h_1\frac{dx(a_r)}{dx(q)}\omega_{h_1,1}(q)\omega_{h-h_1,k+1}(q,z_1,\dots,z_k)\cr
&&=\sum_{h_1=1}^h \Res_{q\to a_r}\frac{dx(a_r)}{dx(q)}\omega_{h_1,1}(q) d_q\left(\frac{x(q)-x(a_r)}{dx(q)}\omega_{h-h_1,k+1} (q,z_1,\dots,z_k)\right)
\eea
which is precisely the application of \autoref{LemmaIntW02Wg1} with $F(q)=\frac{x(q)-x(a_r)}{dx(q)}\omega_{h-h_1,k+1} (q,z_1,\dots,z_k)$ which is holomorphic in a neighborhood of $q=a_r$.

\section{Proof of \autoref{VariationalFormulaLogTRPole}: Variational formulas with respect to LogTR-vital singularities for free energies}\label{AppendixProofVarationalFormulasLogR}

By definition, we have
\bea (2-2h)\omega_{h,0}&:=&\sum_{i=1}^N\Res_{z\to p_i} \Phi_{p_i}(z)\omega_{h,1}(z)
\cr&&
-\sum_{j=1}^M\Res_{z\to a_j}\big(x(z)-x(a_j)\big)\left(\frac{1}{2}
        \overset{h-1}{\underset{h_1=1}{\sum}}\frac{\omega_{h_1,1}(z)\omega_{h-h_1,1} (z) }{dx(z)} -dy(z)\int_o^z\omega_{h,1}\right)\cr&&
\eea
where $\Phi_{p_i}(z)=\int_{p_i}^z ydx$. Since $x$ is not varied and $d_{a_r}[ydx(z)]=y_{a_r}dS_{o,z}(a_r) dx(z) $, 
we get that
\small{\bea (2-2h)d_{a_r}[\omega_{h,0}]&=&(A)+(B)+(C)+ (D)+ (E)+ (F)\text{ with}\cr
(A)&:=&\sum_{i=1}^N\Res_{z\to p_i} d_{a_r}[\Phi_{p_i}(z)]\omega_{h,1}(z)\cr 
(B)&:=&\sum_{i=1}^N\Res_{z\to p_i} \Phi_{p_i}(z)d_{a_r}[\omega_{h,1}(z)]\cr
(C)&:=&dx(a_r)\Res_{z\to a_r} \left(\frac{1}{2}\overset{h-1}{\underset{h_1=1}{\sum}}\frac{\omega_{h_1,1} (z)\omega_{h-h_1,1} (z) }{dx(z)} -dy(z)\int_o^z\omega_{h,1}\right)\cr
(D)&:=&-\frac{1}{2}\sum_{j=1}^M\Res_{z\to a_j}\big(x(z)-x(a_j)\big)\left(\overset{h-1}{\underset{h_1=1}{\sum}}\frac{d_{a_r}[\omega_{h_1,1} (z)]\omega_{h-h_1,1}(z) +\omega_{h_1,1} (z)d_{a_r}[\omega_{h-h_1,1}(z)]}{dx(z)}\right)\cr
(E)&:=&y_{a_r}\sum_{j=1}^M\Res_{z\to a_j}\big(x(z)-x(a_j)\big)B(z,a_r)\int_o^z\omega_{h,1} \cr
(F)&:=& \sum_{j=1}^M\Res_{z\to a_j}\big(x(z)-x(a_j)\big)dy(z)\int_o^zd_{a_r}[\omega_{h,1}]
\eea}
\normalsize{From} \autoref{TheoVariationsLogTRpoles} we have:
\beq d_{a_r}[\omega_{h,1}(z)] 
= \sum_{i=1}^N \Res_{q\to p_i} d_{a_r}[\Phi_{p_i}(q)] \omega_{h,2}(q,z)+ \Res_{q\to a_r} \frac{dx(a_r)}{dx(q)}\sum_{h_1=1}^h \omega_{h_1,1}(q)\omega_{h-h_1,2}(q,z)
\eeq
Note that $d_{a_r}[\omega_{h,1}(z)]$ has only poles at the ramification points and at $z=a_r$ but not at the other LogTR-vital singularities. Thus we get from the fact that $(x(z)-x(a_j))dy(z)$ is regular at $z=a_j$:
\bea (F)&=&(F_1)+(F_2) \text{ with }\cr
(F_1)&:=&\Res_{z\to a_r}\big(x(z)-x(a_r)\big)dy(z)\sum_{i=1}^N \Res_{q\to p_i} d_{a_r}[\Phi_{p_i}(q)] \int_o^z\omega_{h,2}(q,.)\cr
(F_2)&:=&\Res_{z\to a_r}\big(x(z)-x(a_r)\big)dy(z)\Res_{q\to a_r} \frac{dx(a_r)}{dx(q)}\sum_{h_1=1}^h \omega_{h_1,1}(q)\int_o^z\omega_{h-h_1,2}(q,.)
\eea
Since the ramification points are away from $a_r$, we can swap the residues in $(F_1)$. But since $x(z)-x(a_r)$ has a simple zero, $dy(z)$ has a simple pole at $z=a_r$ and $\int_o^z\omega_{h,2}(.,q)$ is regular at $z=a_r$, we get that the integrand is regular at $z=a_r$ so that the residue is vanishing: $(F_1)=0$. A similar argument holds for $(F_2)$: $\big(x(z)-x(a_r)\big)dy(z)$ is regular at $z=a_r$ and for $h-h_1\neq 0$, $\omega_{h-h_1,2}(z,a_r)$ is regular at $z=a_r$. Thus the integrand is regular at $z=a_r$ and hence the residue is vanishing. Thus, we get from \eqref{ResidueExchange} (we isolate first the case $h_1=h$):
\bea (F)&=&\Res_{z\to a_r}\big(x(z)-x(a_r)\big)dy(z)\Res_{q\to a_r} \frac{dx(a_r)}{dx(q)}\omega_{h,1}(q)\int_o^z\omega_{0,2}(q,.)\cr
&&+\Res_{q\to a_r} \Res_{z\to q}\big(x(z)-x(a_r)\big)dy(z)\frac{dx(a_r)}{dx(q)}\sum_{h_1=1}^{h-1} \omega_{h_1,1}(q)\int_o^z\omega_{h-h_1,2}(q,z)\cr
&=& \Res_{z\to a_r}\big(x(z)-x(a_r)\big)dy(z)\Res_{q\to a_r} \frac{dx(a_r)}{dx(q)}\omega_{h,1}(q)\int_o^z\omega_{0,2}(q,.)
\eea
because for $h_1\neq h$, $\omega_{h-h_1,2}(q,z)$ is regular at $z=q$.

Let us now deal with $(B)$:
\bea (B)&:=&(B_1)+(B_2) \text{ with }\cr
(B_1)&:=&\sum_{i=1}^N\Res_{z\to p_i} \Phi_{p_i}(z)\sum_{j=1}^N \Res_{q\to p_j} d_{a_r}[\Phi_{p_j}(q)] \omega_{h,2}(q,z)\cr
(B_2)&:=&\sum_{i=1}^N\Res_{z\to p_i} \Phi(z)\Res_{q\to a_r} \frac{dx(a_r)}{dx(q)}\sum_{h_1=1}^h \omega_{h_1,1}(q)\omega_{h-h_1,2}(q,z)
\eea
For $h_1=h$, we get that the integrand in $(B_2)$ is regular at $z=p_i$ and hence: 
\beq (B_2)=\Res_{q\to a_r} \frac{dx(a_r)}{dx(q)}\sum_{h_1=1}^{h-1}\omega_{h_1,1}(q)\sum_{i=1}^N\Res_{z\to p_i} \Phi_{p_i}(z) \omega_{h-h_1,2}(q,z)
\eeq
In $(B_1)$ we swap the residue using \eqref{ResidueExchange}:
\beq (B_1)=\sum_{j=1}^N \Res_{q\to p_j} d_{a_r}[\Phi(q)]\sum_{i=1}^N\Res_{z\to p_i} \Phi_{p_i}(z)  \omega_{h,2}(q,z)+\sum_{j=1}^N \Res_{q\to p_j} d_{a_r}[\Phi_{p_j}(q)] \Res_{z\to q,\sigma_j(q)} \Phi(z)  \omega_{h,2}(q,z)
\eeq
Since $\omega_{h,2}(q,z)$ is regular at $z=q$, the last term is vanishing. Thus, we get:
\beq (B_1)=\sum_{j=1}^N \Res_{q\to p_j} d_{a_r}[\Phi_{p_j}(q)]\sum_{i=1}^N\Res_{z\to p_i} \Phi(z)  \omega_{h,2}(q,z)\eeq
Let us now deal with $(D)$ that immediately simplifies into (perform $h_1\to h-h_1$ in one of the sum):
\beq (D)=-\sum_{j=1}^M\Res_{z\to a_j}\frac{\big(x(z)-x(a_j)\big)}{dx(z)}\overset{h-1}{\underset{h_1=1}{\sum}}d_{a_r}[\omega_{h-h_1,1} (z)]\omega_{h_1,1}(z)\eeq
Inserting the expression of $d_{a_r}[\omega_{h_1,1}(z)]$ for $h_1\geq 1$, we find:
\small{\bea (D)&:=&(D_1) +(D_2)\text{ with}\cr
(D_1)&:=&- \sum_{j=1}^M\Res_{z\to a_j}\sum_{i=1}^N \Res_{q\to p_i}\frac{\big(x(z)-x(a_j)\big)}{dx(z)}\overset{h-1}{\underset{h_1=1}{\sum}}d_{a_r}[\Phi_{p_i}(q)] \omega_{h-h_1,2}(z,q)\omega_{h_1,1}(z)\cr
&=&-\sum_{j=1}^N \Res_{q\to p_j}d_{a_r}[\Phi_{p_j}(q)] \sum_{i=1}^M\Res_{z\to a_i}\frac{\big(x(z)-x(a_i)\big)}{dx(z)}\overset{h-1}{\underset{h_1=1}{\sum}} \omega_{h-h_1,2}(z,q)\omega_{h_1,1}(z)\cr
(D_2)&:=&-\sum_{j=1}^M\Res_{z\to a_j}\frac{\big(x(z)-x(a_j)\big)}{dx(z)}\overset{h-1}{\underset{h_1=1}{\sum}} \omega_{h_1,1}(z)\Res_{q\to a_r} \frac{dx(a_r)}{dx(q)}\sum_{u=1}^{h-h_1} \omega_{u,1}(q)\omega_{h-h_1-u,2}(q,z)
\cr
&=&-\sum_{j=1}^M\Res_{z\to a_j}\frac{\big(x(z)-x(a_j)\big)}{dx(z)}\sum_{u=1}^{h-1}\sum_{h_1=1}^{h-u} \omega_{h_1,1}(z)\Res_{q\to a_r} \frac{dx(a_r)}{dx(q)} \omega_{u,1}(q)\omega_{h-h_1-u,2}(q,z)
\eea}
\normalsize{Let us} now write the dilaton equation for $\omega_{h,1}(z_1)$ from \autoref{TheoremDilatonEquation}:
\beq \label{DilatonReduced} (1-2h)\omega_{h,1}(q)=\sum_{i=1}^N\Res_{z\to p_i} \Phi_{p_i}(z)\omega_{h,2}(z,q) -\sum_{i=1}^M\Res_{z\to a_i}\frac{x(z)-x(a_i)}{dx(z)}\overset{h}{\underset{h_1=1}{\sum}}\omega_{h_1,1} (z)\omega_{h-h_1,2} (z,q) \eeq
Thus we get that regrouping $(B_1)+(D_1)$ we have:
\small{\bea&&(B_1)+(D_1)=\sum_{j=1}^N\Res_{q\to p_j} d_{a_r}[\Phi_{p_j}(q)]\Big[\sum_{i=1}^N\Res_{z\to p_i} \Phi(z)  \omega_{h,2}(q,z)
- \sum_{i=1}^M\Res_{z\to a_i}\frac{\big(x(z)-x(a_i)\big)}{dx(z)}\overset{h-1}{\underset{h_1=1}{\sum}} \omega_{h-h_1,2}(z,q)\omega_{h_1,1}(z) \Big]\cr
&&= (1-2h)\sum_{j=1}^N\Res_{q\to p_j} d_{a_r}[\Phi_{p_j}(q)]\omega_{h,1}(q) +\sum_{j=1}^N\Res_{q\to p_j} d_{a_r}[\Phi_{p_j}(q)]\sum_{i=1}^M\Res_{z\to a_i}\frac{\big(x(z)-x(a_i)\big)}{dx(z)}\omega_{0,2}(z,q)\omega_{h,1}(z)\cr&&
\eea}
\normalsize{}
In the last term, we may swap the residue that are located at different points and the residue at $q=p_j$ is vanishing because all the terms in the integrand are regular at $q=p_j$. Thus, we get:
\beq (B_1)+(D_1)=(1-2h)\sum_{j=1}^N\Res_{q\to p_j} d_{a_r}[\Phi_{p_j}(q)]\omega_{h,1}(q) \eeq
Combining with $(A)$, we get
\beq \label{FirstPartTheorem} (A)+(B_1)+(D_1)= (2-2h)\sum_{j=1}^N\Res_{q\to p_j} d_{a_r}[\Phi(q)]\omega_{h,1}(q)\eeq
which is the first term of the r.h.s. of \eqref{EqToProve}.

\medskip

We are left with $(B_2)$, $(C)$, $(D_2)$, $(E)$ and $(F)$ that are given by
\bea (B_2)&=&\Res_{q\to a_r} \frac{dx(a_r)}{dx(q)}\sum_{h_1=1}^{h-1}\omega_{h_1,1}(q)\sum_{i=1}^N\Res_{z\to p_i} \Phi_{p_i}(z) \omega_{h-h_1,2}(q,z)\cr
(C)&=&dx(a_r)\Res_{z\to a_r} \left(\frac{1}{2}\overset{h-1}{\underset{h_1=1}{\sum}}\frac{\omega_{h_1,1} (z)\omega_{h-h_1,1} (z) }{dx(z)} -dy(z)\int_o^z\omega_{h,1}\right)\cr
(D_2)&=&-\sum_{j=1}^M\Res_{z\to a_j}\frac{\big(x(z)-x(a_j)\big)}{dx(z)}\sum_{u=1}^{h-1}\sum_{h_1=1}^{h-u} \omega_{h_1,1}(z)\Res_{q\to a_r} \frac{dx(a_r)}{dx(q)} \omega_{u,1}(q)\omega_{h-h_1-u,2}(q,z)\cr
(E)&=&y_{a_r}\sum_{j=1}^M\Res_{z\to a_j}\big(x(z)-x(a_j)\big)B(z,a_r)\int_o^z\omega_{h,1} \cr
(F)&=&\Res_{z\to a_r}\big(x(z)-x(a_r)\big)dy(z)\Res_{q\to a_r} \frac{dx(a_r)}{dx(q)}\omega_{h,1}(q)\int_o^z\omega_{0,2}(q,.)
\eea
We swap the residues in $(B_2)$ and $(D_2)$. This can be done for $(B_2)$ but for $(D_2)$, it provides an additional term. This term is non-zero except when the Bergman kernel is involved corresponding to $h-h_1-u=0$, i.e. $h_1=h-u$.  From \eqref{ResidueExchange}, we get:
\bea (D_2)&=&-\Res_{q\to a_r}\frac{dx(a_r)}{dx(q)}\sum_{u=1}^{h-1} \omega_{u,1}(q) \sum_{j=1}^M\Res_{z\to a_j}\frac{\big(x(z)-x(a_j)\big)}{dx(z)}\sum_{h_1=1}^{h-u} \omega_{h_1,1}(z) \omega_{h-h_1-u,2}(q,z)\cr&&
-\Res_{q\to a_r}\frac{dx(a_r)}{dx(q)}\sum_{u=1}^{h-1} \omega_{u,1}(q) \Res_{z\to q}\frac{\big(x(z)-x(a_r)\big)}{dx(z)}\sum_{h_1=1}^{h-u} \omega_{h_1,1}(z) \omega_{h-h_1-u,2}(q,z)\cr
&=&-\Res_{q\to a_r}\frac{dx(a_r)}{dx(q)}\sum_{u=1}^{h-1} \omega_{u,1}(q) \sum_{j=1}^M\Res_{z\to a_j}\frac{\big(x(z)-x(a_j)\big)}{dx(z)}\sum_{h_1=1}^{h-u} \omega_{h_1,1}(z) \omega_{h-h_1-u,2}(q,z)\cr&&
-\Res_{q\to a_r}\frac{dx(a_r)}{dx(q)}\sum_{u=1}^{h-1} \omega_{u,1}(q) \Res_{z\to q}\frac{\big(x(z)-x(a_r)\big)}{dx(z)} \omega_{h-u,1}(z) \omega_{0,2}(q,z)\cr
&=&-\Res_{q\to a_r}\frac{dx(a_r)}{dx(q)}\sum_{h_1=1}^{h-1} \omega_{h_1,1}(q) \sum_{j=1}^M\Res_{z\to a_j}\frac{\big(x(z)-x(a_j)\big)}{dx(z)}\sum_{u=1}^{h-h_1} \omega_{u,1}(z) \omega_{h-u-h_1,2}(q,z)\cr&&
-\Res_{q\to a_r}\frac{dx(a_r)}{dx(q)}\sum_{u=1}^{h-1} \omega_{u,1}(q) d_q\left[\frac{\big(x(q)-x(a_r)\big)}{dx(q)} \omega_{h-u,1}(q)\right] 
\eea

We now use the dilaton equation \autoref{TheoremDilatonEquation} for $\omega_{h-h_1,2}$ with $h-h_1\geq 1$:
\small{\beq \label{Dilatwhk} (1-2h+2h_1)\omega_{h-h_1,1}(q)=\sum_{i=1}^N\Res_{z\to p_i} \Phi_{p_i}(z)\omega_{h-h_1,2}(z,q) -\sum_{j=1}^M\Res_{z\to a_j}\frac{x(z)-x(a_j)}{dx(z)}\overset{h-h_1}{\underset{u=1}{\sum}}\omega_{u,1} (z)\omega_{h-h_1-u,2} (z,q)
\eeq}
\normalsize{Thus,} combining $(B_2)$ and $(D_2)$ provides
\bea \label{B2PlusD2} (B_2)+(D_2)
&=& \Res_{q\to a_r}\frac{dx(a_r)}{dx(q)}\sum_{h_1=1}^{h-1}(1-2h+2h_1)\omega_{h_1,1}(q)\omega_{h-h_1,1}(q) \cr&&
-\Res_{q\to a_r}\frac{dx(a_r)}{dx(q)}\sum_{u=1}^{h-1} \omega_{u,1}(q) d_q\left[\frac{\big(x(q)-x(a_r)\big)}{dx(q)} \omega_{h-u,1}(q)\right] 
\eea
Let us denote 
\beq (G):=-\Res_{q\to a_r}\frac{dx(a_r)}{dx(q)}\sum_{h_1=1}^{h-1} \omega_{h_1,1}(q) d_q\left[\frac{\big(x(q)-x(a_r)\big)}{dx(q)} \omega_{h-h_1,1}(q)\right] \eeq

and 
\beq (H):=\Res_{q\to a_r}\frac{dx(a_r)}{dx(q)}\sum_{h_1=1}^{h-1}(1-2h_1)\omega_{h_1,1}(q)\omega_{h-h_1,1}(q)\eeq
Note that
\bea (H)&=& \Res_{q\to a_r}\frac{dx(a_r)}{dx(q)}\sum_{h_1=1}^{h-1}(1-2h_1)\omega_{h_1,1}(q)\omega_{h-h_1,1}(q)\cr
&\overset{h_1\to h-h_1}{=}& \Res_{q\to a_r}\frac{dx(a_r)}{dx(q)}\sum_{h_1=1}^{h-1}(1-2h+2h_1)\omega_{h-h_1,1}(q)\omega_{h_1,1}(q)
\eea
so that
\bea 2(H)&=& \Res_{q\to a_r}\frac{dx(a_r)}{dx(q)}\sum_{h_1=1}^{h-1}(1-2h_1+1-2h+2h_1)\omega_{h_1,1}(q)\omega_{h-h_1,1}(q)\cr
&=&(2-2h)\Res_{q\to a_r}\frac{dx(a_r)}{dx(q)}\sum_{h_1=1}^{h-1}\omega_{h_1,1}(q)\omega_{h-h_1,1}(q)
\eea
Thus:
\beq (H)=\frac{1}{2}(2-2h)\Res_{q\to a_r}\frac{dx(a_r)}{dx(q)}\sum_{h_1=1}^{h-1}\omega_{h_1,1}(q)\omega_{h-h_1,1}(q)\eeq

so that we have that
\beq (B_2)+(D_2) -(G)= \frac{1}{2}(2-2h)\Res_{q\to a_r}\frac{dx(a_r)}{dx(q)}\sum_{h_1=1}^{h-1}\omega_{h_1,1}(q)\omega_{h-h_1,1}(q) \eeq
Hence, $(A)+(B)+(D)-(G)$ provides the first two terms of \eqref{EqToProve}. Let us now prove that we can obtain the second and third lines of the r.h.s. of \eqref{EqToProve} from 
\beq (N):=(C)+(E)+(F)+(G)\eeq

We shall first decompose:
\bea (C)&:=&(C_1)+(C_2)\text{ with}\cr
(C_1)&:=&\frac{1}{2}dx(a_r)\Res_{z\to a_r} \overset{h-1}{\underset{h_1=1}{\sum}}\frac{\omega_{h_1,1} (z)\omega_{h-h_1,1} (z) }{dx(z)}\cr
(C_2)&:=& -dx(a_r)\Res_{z\to a_r} dy(z)\int_o^z \omega_{h,1}
\eea

Let us now focus on $(G)$. We have
\small{\bea \label{weirdId}0&=&dx(a_r)\sum_{h_1=1}^{h-1}\Res_{q\to a_r}d_q\left[ \frac{\omega_{h_1,1}(q)}{dx(q)} \big(x(q)-x(a_r)\big)\frac{\omega_{h-h_1,1}(q)}{dx(q)} \right]\cr
&=&dx(a_r)\sum_{h_1=1}^{h-1}\Res_{q\to a_r}\frac{\omega_{h_1,1}(q) \omega_{h-h_1,1}(q)}{dx(q)} + dx(a_r)\sum_{h_1=1}^{h-1}\Res_{q\to a_r} \frac{\omega_{h_1,1}(q)}{dx(q)} \big(x(q)-x(a_r)\big)d_q\left[\frac{\omega_{h-h_1,1}(q)}{dx(q)} \right] \cr
&&+dx(a_r)\sum_{h_1=1}^{h-1}\Res_{q\to a_r}d_q\left[ \frac{\omega_{h_1,1}(q)}{dx(q)}\right] \big(x(q)-x(a_r)\big)\frac{\omega_{h-h_1,1}(q)}{dx(q)}\cr
&=& dx(a_r)\sum_{h_1=1}^{h-1}\Res_{q\to a_r}\frac{\omega_{h_1,1}(q) \omega_{h-h_1,1}(q)}{dx(q)} + 2dx(a_r)\sum_{h_1=1}^{h-1}\Res_{q\to a_r} \frac{\omega_{h_1,1}(q)}{dx(q)} \big(x(q)-x(a_r)\big)d_q\left[\frac{\omega_{h-h_1,1}(q)}{dx(q)} \right]\cr&&
\eea}
\normalsize{where} we have changed $h_1\to h-h_1$ in the third term. A direct computation provides: 
\small{\bea (G)&=&-\Res_{q\to a_r}\frac{dx(a_r)}{dx(q)}\sum_{h_1=1}^{h-1} \omega_{h_1,1}(q) d_q\left[\frac{\big(x(q)-x(a_r)\big)}{dx(q)} \omega_{h-h_1,1}(q)\right]\cr
&=& -\sum_{h_1=1}^{h-1} \Res_{q\to a_r}\frac{dx(a_r)}{dx(q)}\omega_{h_1,1}(q) \omega_{h-h_1,1}(q)- \Res_{q\to a_r}dx(a_r)\sum_{h_1=1}^{h-1} \frac{\omega_{h_1,1}(q)}{dx(q)} \big(x(q)-x(a_r)\big) d_q\left[\frac{\omega_{h-h_1,1}(q)}{dx(q)}\right]\cr&&
\eea}
\normalsize{Inserting} \eqref{weirdId} into the last identity, we get
\beq (G)=-\sum_{h_1=1}^{h-1} \Res_{q\to a_r}\frac{dx(a_r)}{dx(q)}\omega_{h_1,1}(q) \omega_{h-h_1,1}(q)+\frac{1}{2}dx(a_r)\sum_{h_1=1}^{h-1}\Res_{q\to a_r}\frac{\omega_{h_1,1}(q) \omega_{h-h_1,1}(q)}{dx(q)}
\eeq
so that
\beq\label{C1PlusG} (C_1)+(G)=0
\eeq
and $(N)=(C_2)+(E)+(F)$. Let us now observe that
\beq\label{C2PlusE} (C_2)+(E)=-dx(a_r)\Res_{z\to a_r}dy(z)\int_o^z\omega_{h,1}+y_{a_r}\sum_{j=1}^M\Res_{z\to a_j}\big(x(z)-x(a_j)\big)B(z,a_r)\int_o^z\omega_{h,1}\eeq

Finally, let us turn to
\beq (F)=\Res_{z\to a_r}\big(x(z)-x(a_r)\big)dy(z)\Res_{q\to a_r} \frac{dx(a_r)}{dx(q)}\omega_{h,1}(q)\int_o^z\omega_{0,2}(q,.)\eeq
We exchange the residues using \eqref{ResidueExchange}:
\bea (F)&=&\Res_{q\to a_r}\frac{dx(a_r)}{dx(q)}\omega_{h,1}(q) \Res_{z\to a_r}\big(x(z)-x(a_r)\big)dy(z) \int_o^z\omega_{0,2}(q,.)\cr
&&+ \Res_{q\to a_r}\frac{dx(a_r)}{dx(q)}\omega_{h,1}(q) \Res_{z\to q}\big(x(z)-x(a_r)\big)dy(z) \int_o^z\omega_{0,2}(q,.)
\eea
The first residue is vanishing because the integrand in $z$ is holomorphic at $z=a_r$. Thus, we get:
\bea \label{eqF} (F)&=&\Res_{q\to a_r}\frac{dx(a_r)}{dx(q)}\omega_{h,1}(q) \Res_{z\to q}\big(x(z)-x(a_r)\big)dy(z) \int_o^z\omega_{0,2}(q,.)\cr
&=&-\Res_{q\to a_r}\frac{dx(a_r)}{dx(q)}\omega_{h,1}(q) dy(q) (x(q)-x(a_r))
\eea
Thus, combining \eqref{C2PlusE} and \eqref{eqF} we get that $(C_2)+(E)+(F)$ provides the second and third lines of \eqref{EqToProve} ending the proof for this formulation. 

To obtain the second formulation, we shall us the following lemma.

\begin{lemma}
    We have the following relation 
    \begin{align*}
        &-\frac{1}{2-2h}\Res_{z\to a_r}dx(a_r)dy(z)\int_o^z\omega_{h,1}+\frac{1}{2-2h}y_{a_r}\sum_{j=1}^M\Res_{z\to a_j}\big(x(z)-x(a_j)\big)B(z,a_r)\int_o^z\omega_{h,1}\\
&-\frac{1}{2-2h}\Res_{q\to a_r}\frac{dx(a_r)}{dx(q)}\omega_{h,1}(q) dy(q) (x(q)-x(a_r))\\
=&-\Res_{z\to a_r} dx(a_r) dy(z) \int_o^z\omega_{h,1}-\sum_{j=1}^M\Res_{z\to a_j} d_{a_r}[y(z)] dx(z)\int_o^z\omega_{h,1}.
    \end{align*}
    \begin{proof}
        The first and the third term together with the second at $j=r$ can be written as (neglecting the $1/(2-2h)$ factor and after integration by parts)
        \begin{align*}
            &dx(a_r)\Res_{z\to a_r}\bigg(\int_o^z\omega_{h,1}\bigg) \bigg(d_z\bigg[\frac{dy(z)(x(z)-x(a_r))}{dx(z)}\bigg]- dy(z)+\frac{(x(z)-x(a_r)))d_{a_r}[dy(z)]}{dx(a_r)}\bigg)\\
            =&dx(a_r)\Res_{z\to a_r}\bigg(\int_o^z\omega_{h,1}\bigg) (x(z)-x(a_r))\bigg(d_z\bigg[\frac{dy(z)}{dx(z)}\bigg]+\frac{d_{a_r}[dy(z)]}{dx(a_r)}\bigg).        
        \end{align*}
        Note that $\bigg(d_z\bigg[\frac{dy(z)}{dx(z)}\bigg]+\frac{d_{a_r}[dy(z)]}{dx(a_r)}\bigg)$ is regular at $z=a_r$. Thus, we can integrate by parts again to get $\omega_{h,1}$ (without an integral). Since just the singular part of $\omega_{h,1}$ contributes, we insert the definition of this part from LogTR. Changing the order of the residues gives on the one hand a vanishing contribution and on the other hand  a contribution if $q\to z$. This is easily evaluated and we remain with the following:
        \begin{align*}
            =&-dx(a_r)\Res_{q\to a_r}[\hbar^{2h}]\bigg(\frac{y_{a_r}}{\mathcal{S}(y_{a_r}^{-1}\hbar )}\frac{\partial^{2h}}{\partial x(q)^{2h}}\log(q-a_r)\bigg)dx(q) \int_o^q(x(t)-x(a_r))\bigg(d_t\bigg[\frac{dy(t)}{dx(t)}\bigg]+\frac{d_{a_r}[dy(t)]}{dx(a_r)}\bigg)\\
            =&dx(a_r)\Res_{q\to a_r}[\hbar^{2h}]\bigg(\frac{y_{a_r}}{\mathcal{S}(y_{a_r}^{-1}\hbar)}\frac{\partial^{2h-1}}{\partial x(q)^{2h-1}}\log(q-a_r)\bigg) (x(q)-x(a_r))\bigg(d_q\bigg[\frac{dy(q)}{dx(q)}\bigg]+\frac{d_{a_r}[dy(q)]}{dx(a_r)}\bigg)\\
            =&dx(a_r)(2-2h)\Res_{q\to a_r}[\hbar^{2h}]\bigg(\frac{y_{a_r}}{\mathcal{S}(y_{a_r}^{-1}\hbar)}\frac{\partial^{2h-2}}{\partial x(q)^{2h-2}}\log(q-a_r)\bigg)\bigg(d_q\bigg[\frac{dy(q)}{dx(q)}\bigg]+\frac{d_{a_r}[dy(q)]}{dx(a_r)}\bigg)\\
            =&-dx(a_r)(2-2h)\Res_{q\to a_r}\bigg(\int_o^q\omega_{h,1}\bigg)\bigg(dy(q)+dx(q)\frac{d_{a_r}[y(q)]}{dx(a_r)}\bigg).
        \end{align*}
        In the second last line we used a version of \autoref{LemmaIntW02Wg1} and in the last line integration by parts together with the explicit expression of the LogTR-vital singular term of $\omega_{h,1}$ at  $a_r$. This produces the terms of the lemma for the residues at $a_r$. For the second term with residue at $a_{j}$ with $j\neq r$, the same computation and arguments hold. Bringing everything together, the lemma is proved. 
    \end{proof}
\end{lemma}

\section{Variation of $\omega_{1,0}$ with respect to LogTR-vital singularities}\label{AppendixVarF1LogTR}
Let us study the variations of $F_1=\omega_{1,0}$ with respect to LogTR-vital singularities. Let us recall that $\omega_{1,0}$ is defined in \autoref{DefFreeEnergies} and more precisely by the specific formula:
\beq \omega_{1,0}:= -\frac{1}{2}\ln \tau -\frac{1}{24}\ln\left(\prod_{i=1}^N y'(p_i)\right) -\frac{1}{24}\sum_{s=1}^M\bigg(\frac{y(z)}{y_{a_s} }-\log(x(z)-x(a_s))\bigg)\bigg\vert_{z=a_s}\eeq
where  $y'(p_i)=\frac{dy(p_i)}{dz_i(p_i)}$ with $z_i(q)=\sqrt{x(q)-x(p_i)}$. 

The proof consists of two steps, first computing the variation of \eqref{F1formula} wrt $a_r$, and second comparing this with the RHS of the variational formula \eqref{EqToProve2} for $(h,n)=(1,0)$.

1.) Rauch variational formula (See \cite[Section $5$]{EO07}) implies that for variations at fixed $x$, we have:
\beq \label{Defomega10}\delta_{\Omega}[\tau]=2i\pi \sum_{i=1}^N \Res_{z\to p_i}\frac{\Omega(z)d\mathbf{u}(z)d\mathbf{u}(z)^t}{dx(z)dy(z)}\eeq 
\sloppy{Therefore, for a LogTR-vital singularity $a_r\in \mathcal{S}_y$, we have from $\Omega_{a_r}(z):=d_{a_r}[ydx(z)]=y_{a_r} \frac{\partial_{a_r} E(z,a_r)}{E(z,a_r)} dx(z)$ that $d_{a_r}[\tau]=0$. Indeed, in that case the integrand simplifies into $y_{a_r}\frac{\partial_{a_r} E(z,a_r)}{E(z,a_r)}\frac{d\mathbf{u}(z)d\mathbf{u}(z)^t}{dy(z)}$ which is regular at the ramification points so that the residue is vanishing.}  From \autoref{TheoremGlobalDecompositionydx} and the fact that $a_r$ is a LogTR-vital singularity (i.e. $\td{y}dx$ is regular at $a_r$), we get that 
\beq d_{a_r}[\td{y}(z)]=0\eeq
However, for each $y(z)$ there is a term locally of the form $y_{a_r}\ln \frac{E(z,a_r)}{E(z,o)}=y_{a_r}\log(z-a_r) +O(1)$. Combining with $\log(x(z)-x(a_s))$ yields in the limit a term of the form $-\log(x'(a_r))$.
Thus the variations of the last term in \eqref{Defomega10} give two contributions, either acting on $a_r$ in $y(z)$ or on the argument $z=a_r$. We find
\beq d_{a_r}\left[ \frac{1}{24}\sum_{s=1}^M\bigg(\frac{y(z)}{y_{a_s} }-\log(x(z)-x(a_s))\bigg)\bigg\vert_{z=a_s}\right]
=\frac{1}{24} \sum_{\substack{s=1\\ s\neq r}}^M\frac{y_{a_r}}{y_{a_s}}
\frac{\partial_{a_r} E(a_s,a_r)}{E(a_s,a_r)}
+\frac{1}{24}\left(\frac{dy_r(a_r)}{y_{a_r}}-\frac{x''(a_r)}{x'(a_r)}\right)\eeq
where $y_r(z):=y(z)-y_{a_r}\ln \frac{E(z,a_r)}{E(z,o)} $, is the regular part of $y(z)$ at $z=a_r$.

Finally since $d_{a_r}[y'(z)]=y_{a_r}B(z,a_r)$ and using the fact that the ramification points (defined by $dx(p_i)=0$) are also invariant under deformations at fixed $x$, we get that
\beq d_{a_r}\left[-\frac{1}{24}\ln\left(\prod_{i=1}^N y'(p_i)\right)\right]=-\frac{1}{24}\sum_{i=1}^N\frac{d_{a_r}[y'(p_i)]}{y'(p_i)}=-\frac{1}{24}\sum_{i=1}^N\frac{d_{a_r}[dy(p_i)]}{dy(p_i)}=-\frac{1}{24}\sum_{i=1}^N y_{a_r}\frac{B(p_i,a_r)}{dy(p_i)}\eeq
Thus, we find:
\beq \label{VariationOmega10ar} d_{a_r}[\omega_{1,0}]= -\frac{1}{24}\sum_{i=1}^N y_{a_r}\frac{B(p_i,a_r)}{dy(p_i)}-\frac{1}{24} \sum_{\substack{s=1\\ s\neq r}}^M\frac{y_{a_r}}{y_{a_s}}
\frac{\partial_{a_r} E(a_s,a_r)}{E(a_s,a_r)}
-\frac{1}{24}\left(\frac{dy_r(a_r)}{y_{a_r}}-\frac{x''(a_r)}{x'(a_r)}da_r\right)\eeq

2.) Let us now compare with the r.h.s. of \eqref{EqToProve2}, we have:
\bea \label{RHSF1Proof}&&\Res_{z\to \{p_i\}\cup \{a_j\}} d_{a_r}[\Phi(z)]\omega_{1,1}(z) +\Res_{z\to a_r}dx(a_r)y(z)\omega_{1,1} (z).
\eea
For $\omega_{1,1}(z)$ near $a_s$, we have 
from \autoref{DefLogTR}:
\bea \omega_{1,1}(z) &=& -\frac{1}{24y_{a_s}}\Res_{q\to a_s} dS_{a_s,q}(z)dx(q)\left( \frac{1}{x'(q)^2(q-a_s)^2} +\frac{x''(q)}{x'(q)^3(q-a_s)}\right) +O(1)\cr
&=& -\frac{dz}{24y_{a_s}x'(a_s)(z-a_s)^2} +O(1)
\eea 

In particular, we get that for the residue at $t\to a_s\neq a_r$
\bea 
\Res_{z\to a_s}d_{a_r}[\Phi(z)]\omega_{1,1}(z)
&=&-\frac{1}{24}\frac{y_{a_r} }{y_{a_s}} \frac{\partial_{a_r} E(a_s,a_r)}{E(a_s,a_r)}
\eea

For the residue at a ramification point $p_i$ we compute
\bea
\Res_{z\to p_i}d_{a_r}[\Phi(z)]\omega_{1,1}(z)
&=&\Res_{z\to p_i}d_{a_r}[\Phi(z)]\Res_{q\to p_i}\frac{1}{2}\frac{\int_{q}^{\sigma_i(q)} \omega_{0,2}(z,.)}{\omega_{0,1}(\sigma_i(q))-\omega_{0,1}(q)}\omega_{0,2}(q,\sigma_i(q))\cr
&=& -\Res_{q\to p_i}\Res_{z\to q,\sigma_i(q)}d_{a_r}[\Phi(z)]\frac{1}{2}\frac{\int_{q}^{\sigma_i(q)} \omega_{0,2}(z,.)}{\omega_{0,1}(\sigma_i(q))-\omega_{0,1}(q)}\omega_{0,2}(q,\sigma_i(q))\cr
&=& \Res_{q\to p_i}\frac{1}{2}\frac{d_{a_r}[\Phi(q)]-d_{a_r}[\Phi(\sigma_i(q))]}{\omega_{0,1}(\sigma_i(q))-\omega_{0,1}(q)}\omega_{0,2}(q,\sigma_i(q)).
\eea
Note that $\frac{\omega_{0,2}(q,\sigma_i(q)}{\omega_{0,1}(\sigma_i(q))-\omega_{0,1}(q)}$ is residue-free at $q\to p_i$. This means it can be integrated locally such that we can use integration by parts. Then, $d_qd_{a_r}[\Phi(q)]$ vanishes at $q\to p_i$, which means finally, that we just have to consider the third and fourth order term of $\frac{\omega_{0,2}(q,\sigma_i(q)}{\omega_{0,1}(\sigma_i(q))-\omega_{0,1}(q)}$ in the Laurent expansion at $q\to p_i$, which are of the form
\beq
\frac{\omega_{0,2}(q,\sigma_i(q))}{\omega_{0,1}(\sigma_i(q))-\omega_{0,1}(q)}=-\frac{dq}{8 (q-p_i)^4 x''(p_i) y'(p_i)}+
\frac{x'''(p_i)dq}{24 (q-p_i)^3 x''(q)y'(p_i)}+O((q-p_i)^{-2}dq).
\eeq

Inserting this, it is an easy computation to show that 
\bea
&&\Res_{z\to p_i}d_{a_r}[\Phi(z)]\omega_{1,1}(z)=-\frac{1}{24}\frac{B(p_i,a_r)}{y_{a_r}dy(p_i)}.
\eea

Lastly, we have to compute the residue at $z\to a_r$. First, we observe that $d_{a_r}[\Phi(z)]+dx(a_r)y(z)$ is regular near $z\to a_r$ (it is easier to see after taking its derivative, which does not have a pole $z\to a_r$ concluding that its local integral is regular). Thus, the only contribution for a possible pole of the integrand comes from $\omega_{1,1}(z)$. Integrating by parts again, we find
\beq\Res_{z\to a_r} (d_{a_r}[\Phi(z)]+dx(a_r)y(z))\omega_{1,1}(z) 
= -\Res_{z\to a_r} \left(y_{a_r}\frac{dx(z)da_r}{a_r-z}+dx(a_r)dy(z)\right)\frac{dz}{24y_{a_r}x'(a_r)(z-a_r)}.
\eeq

For the second term in the brackets, we split $dy(z)=dy_r(z)+y_{a_r}dS_{o,a_r}(z)$
. The residue of the regular part $dy_r(z)$ gives the expected result $-\frac{1}{24}\frac{dy_r(a_r)}{y_{a_r}}$. The other term has the residue
\beq-\Res_{z\to a_r} \left(y_{a_r}\frac{dx(z)da_r}{a_r-z}+dx(a_r)y_{a_r}\frac{dz}{z-a_r}\right)\frac{1}{24y_{a_r}x'(a_r)(z-a_r)}
=da_r\frac{x''(a_r)}{24 x'(a_r)}.
\eeq
This ends the proof of the variational formula for the free energy $F_1$ in the presence of LogTR-vital singularities.

\newpage
\bibliographystyle{plain}
\bibliography{Biblio}

@article{EGF19,
    author = "B. Eynard and E. Garcia-Failde",
    title = "{From topological recursion to wave functions and PDEs quantizing hyperelliptic curves}",
    doi = "10.1017/fms.2023.96",
    journal = "Forum Math. Sigma",
    volume = "11",
    year = "2023"
}

@article{MO19_hyper,
	Author = {O.~Marchal and N.~Orantin},
	Fjournal = {Journal of Geometry and Physics},
	Journal = {J. Geom. Phys},
	Title = {Quantization of hyper-elliptic curves from isomonodromic systems and topological recursion},
	Year = {2021}}

@article{Rauch,
	Author = {H.~E.~Rauch},
	Doi = {10.1002/cpa.3160120310},
	Fjournal = {Communications on Pure and Applied Mathematics},
	Issn = {0010-3640},
	Journal = {Comm. Pure Appl. Math.},
	Mrclass = {30.00},
	Mrnumber = {110798},
	Mrreviewer = {T. Klotz},
	Pages = {543--560},
	Title = {Weierstrass points, branch points, and moduli of {R}iemann surfaces},
	Url = {https://doi.org/10.1002/cpa.3160120310},
	Volume = {12},
	Year = {1959}}

@article{BouchardEynard_QC,
	Author = {V.~Bouchard and B.~Eynard},
	Doi = {10.5802/jep.58},
	Fjournal = {Journal de l'\'{E}cole polytechnique. Math\'{e}matiques},
	Issn = {2429-7100},
	Journal = {J. Ec. Polytech. - Math},
	Mrclass = {81R10 (14H70 14H81 30F30)},
	Mrnumber = {3694097},
	Mrreviewer = {Hsian-Hua Tseng},
	Pages = {845--908},
	Title = {Reconstructing {WKB} from topological recursion},
	Url = {https://doi.org/10.5802/jep.58},
	Volume = {4},
	Year = {2017}}

@article{Iwaki-P1,
	Author = {K.~Iwaki},
	Doi = {10.1007/s00220-020-03769-2},
	Fjournal = {Communications in Mathematical Physics},
	Issn = {0010-3616},
	Journal = {Comm. Math. Phys.},
	Mrclass = {34M55 (14H81 33E17 34M60 37J65 81Q20)},
	Number = {2},
	Pages = {1047--1098},
	Title = {2-parameter {$\tau$}-function for the first {P}ainlev\'{e} equation: topological recursion and direct monodromy problem via exact {WKB} analysis},
	Volume = {377},
	Year = {2020}}

@article{MOsl2,
	Author = {O.~Marchal and N.~Orantin},
	Doi = {10.1063/5.0002260},
	Fjournal = {Journal of Mathematical Physics},
	Issn = {0022-2488},
	Journal = {J. Math. Phys.},
	Mrclass = {34M56 (32G34 81Q20)},
	Number = {6},
	Title = {Isomonodromic deformations of a rational differential system and reconstruction with the topological recursion: the {$\mathfrak{sl}_2$} case},
	Volume = {61},
	Year = {2020}}

@article{Quantization_2021,
	Author = {B.~Eynard and E.~Garcia-Failde and O.~Marchal and N.~Orantin},
	journal={Commun. Math. Phys.},
    fjournal={Communications in Mathematical Physics},
	Title = {Quantization of classical spectral curves via topological recursion},
	Year = {2024}
}

@article{Banerjee:2025qgx,
    author = "S. Banerjee and A. Hock and O. Marchal",
    title = "{GW/DT invariants and 5D BPS indices for strips from topological recursion}",
    eprint = "2508.15459",
    archivePrefix = "arXiv",
    primaryClass = "math-ph",
    doi = "10.1007/s11005-026-02046-y",
    journal = "Lett. Math. Phys.",
    volume = "116",
    number = "1",
    pages = "14",
    year = "2026"
}

@article{Borot:2021btb,
    author = {G. Borot and V. Bouchard and N.K. Chidambaram and T. Creutzig},
    title = "{Whittaker vectors for $\mathcal {W}$-algebras from topological recursion}",
    eprint = "2104.04516",
    archivePrefix = "arXiv",
    primaryClass = "math-ph",
    reportNumber = "MPIM-Bonn-2021",
    doi = "10.1007/s00029-024-00921-x",
    journal = "Selecta Math.",
    volume = "30",
    number = "2",
    pages = "33",
    year = "2024"
}

@article{Barbieri:2019yya,
    author = "A. Barbieri and T. Bridgeland and J. Stoppa",
    title = "{A Quantized Riemann{\textendash}Hilbert Problem in Donaldson{\textendash}Thomas Theory}",
    eprint = "1905.00748",
    archivePrefix = "arXiv",
    primaryClass = "math.AG",
    doi = "10.1093/imrn/rnaa294",
    journal = "Int. Math. Res. Not.",
    volume = "2022",
    number = "5",
    pages = "3417--3456",
    year = "2022"
}

@article{Eynard:2012nj,
    author = "B. Eynard and N. Orantin",
    title = "{Computation of Open Gromov{\textendash}Witten Invariants for Toric Calabi{\textendash}Yau 3-Folds by Topological Recursion, a Proof of the BKMP Conjecture}",
    eprint = "1205.1103",
    archivePrefix = "arXiv",
    primaryClass = "math-ph",
    reportNumber = "IPHT-T12-030",
    doi = "10.1007/s00220-015-2361-5",
    journal = "Commun. Math. Phys.",
    volume = "337",
    number = "2",
    pages = "483--567",
    year = "2015"
}

@article{Fang:2013dna,
    author = "B. Fang and C.C.M. Liu and Z. Zong",
    editor = "Bouchard, Vincent and M{\'e}ndez-Diez, Stefan and Quigley, Callum and Doran, Charles",
    title = "{All Genus Open-Closed Mirror Symmetry for Affine Toric Calabi-Yau 3-Orbifolds}",
    eprint = "1310.4818",
    archivePrefix = "arXiv",
    primaryClass = "math.AG",
    journal = "Proc. Symp. Pure Math.",
    volume = "93",
    pages = "1",
    year = "2015"
}

@incollection {Wittenrspin,
    AUTHOR = {E. Witten},
     TITLE = {Algebraic geometry associated with matrix models of
              two-dimensional gravity},
 BOOKTITLE = {Topological methods in modern mathematics ({S}tony {B}rook,
              {NY}, 1991)},
     PAGES = {235--269},
 PUBLISHER = {Publish or Perish, Houston, TX},
      YEAR = {1993},
   MRCLASS = {32G15 (14H15 81T40)},
  MRNUMBER = {1215968},
MRREVIEWER = {Claude\ Itzykson},
}

@article{Bouchard:2007hi,
    author = "V. Bouchard and M. Mari{\~n}o",
    title = "{Hurwitz numbers, matrix models and enumerative geometry}",
    eprint = "0709.1458",
    archivePrefix = "arXiv",
    primaryClass = "math.AG",
    reportNumber = "CERN-PH-TH-2007-152",
    doi = "10.1090/pspum/078/2483754",
    journal = "Proc. Symp. Pure Math.",
    volume = "78",
    pages = "263--283",
    year = "2008"
}

@misc{Eynard:2023dha,
    author = "Eynard, B.",
    title = "{Generalized cycles on Spectral Curves}",
    eprint = "2311.15450",
    archivePrefix = "arXiv",
    primaryClass = "math-ph",
    reportNumber = "IPHT2023",
    note={arXiv:2311.15450, 2023}
}

@article{Osuga:2023kgw,
    author = "K. Osuga",
    title = "{Refined Topological Recursion Revisited: Properties and Conjectures}",
    eprint = "2305.02494",
    archivePrefix = "arXiv",
    primaryClass = "math-ph",
    doi = "10.1007/s00220-024-05169-2",
    journal = "Commun. Math. Phys.",
    volume = "405",
    number = "12",
    pages = "296",
    year = "2024"
}

@article{Kidwai:2023fxs,
    author = {O. Kidwai and K. Osuga},
    title = {{Refined BPS Structures and Topological Recursion—the Weber and Whittaker Curves}},
    journal = {Int. Math. Res. Not.},
    volume = {2025},
    number = {10},
    year = {2025},
    month = {05},
    issn = {1073-7928},
    doi = {10.1093/imrn/rnaf116},
    url = {https://doi.org/10.1093/imrn/rnaf116},
    eprint = {https://academic.oup.com/imrn/article-pdf/2025/10/rnaf116/63255915/rnaf116.pdf},
}

@article{Borot:2024uos,
     author = {G. Borot and N.K. Chidambaram and G. Umer},
     title = {Whittaker vectors at finite energy scale, topological recursion and {Hurwitz} numbers},
     journal = {J. Ec. Polytech. - Math},
     pages = {1503--1564},
     year = {2025},
     publisher = {\'Ecole polytechnique},
     volume = {12},
     doi = {10.5802/jep.316},
     language = {en},
     url = {https://jep.centre-mersenne.org/articles/10.5802/jep.316/}
}

@article{Nekrasov:2003rj,
    author = "N. Nekrasov and A. Okounkov",
    title = "{Seiberg-Witten theory and random partitions}",
    eprint = "hep-th/0306238",
    archivePrefix = "arXiv",
    reportNumber = "ITEP-TH-36-03, PUDM-2003, IHES-P-03-43",
    doi = "10.1007/0-8176-4467-9_15",
    journal = "Prog. Math.",
    volume = "244",
    pages = "525--596",
    year = "2006"
}

@article{ALS,
	Author = {A.~Alexandrov and D.~Lewa\'nski and S.~Shadrin},
	Journal = {JHEP},
	Title = {Ramifications of {H}urwitz theory, {K}{P} integrability and quantum curves},
	Volume = {5},
	Year = {2016}}

@article{Hock:2025wlm,
    author = "A. Hock",
    title = "{Symplectic (Non-)invariance of the Free Energy in Topological Recursion}",
    eprint = "2502.18115",
    archivePrefix = "arXiv",
    primaryClass = "math-ph",
    doi = "10.1007/s00220-025-05373-8",
    journal = "Commun. Math. Phys.",
    volume = "406",
    number = "8",
    pages = "192",
    year = "2025"
}

@misc{EOMirzakhani,
	Author = {B.~Eynard and N.~Orantin},
	Note = {arXiv:0705.3600, 2007},
	Title = {{Weil-Petersson volume of moduli spaces, Mirzakhani's recursion and matrix models}},
   }

@incollection{Norbury_survey,
	Author = {P.~Norbury},
	Booktitle = {String-{M}ath 2014},
	Mrclass = {14H81 (05A15 14N10 81S10)},
	Mrnumber = {3524233},
	Mrreviewer = {Hsian-Hua Tseng},
	Pages = {41--65},
	Publisher = {Amer. Math. Soc., Providence, RI},
	Series = {Proc. Sympos. Pure Math.},
	Title = {Quantum curves and topological recursion},
	Volume = {93},
	Year = {2016}}

@article{Hock:2023dno,
    author = "A. Hock",
    title = "{x-y duality in topological recursion for exponential variables via quantum dilogarithm}",
    eprint = "2311.11761",
    archivePrefix = "arXiv",
    primaryClass = "math-ph",
    doi = "10.21468/SciPostPhys.17.2.065",
    journal = "SciPost Phys.",
    volume = "17",
    number = "2",
    pages = "065",
    year = "2024"
}

@article{Bouchard:2011ya,
    author = "V. Bouchard and P. Sulkowski",
    title = "{Topological recursion and mirror curves}",
    eprint = "1105.2052",
    archivePrefix = "arXiv",
    primaryClass = "hep-th",
    reportNumber = "CALT-68-2836",
    doi = "10.4310/ATMP.2012.v16.n5.a3",
    journal = "Adv. Theor. Math. Phys.",
    volume = "16",
    number = "5",
    pages = "1443--1483",
    year = "2012"
}

@article{Gukov:2011qp,
    author = "S. Gukov and P. Sulkowski",
    title = "{A-polynomial, B-model, and Quantization}",
    eprint = "1108.0002",
    archivePrefix = "arXiv",
    primaryClass = "hep-th",
    reportNumber = "CALT-68-2842",
    doi = "10.1007/JHEP02(2012)070",
    journal = "JHEP",
    volume = "02",
    pages = "070",
    year = "2012"
}

@article{Hock:2022wer,
    author = "A. Hock",
    title = "{On the x-y Symmetry of Correlators in Topological Recursion via Loop Insertion Operator}",
    eprint = "2201.05357",
    archivePrefix = "arXiv",
    primaryClass = "math-ph",
    doi = "10.1007/s00220-024-05043-1",
    journal = "Commun. Math. Phys.",
    volume = "405",
    number = "7",
    pages = "166",
    year = "2024"
}

@article{Hock:2022pbw,
    author = "A. Hock",
    title = "{A simple formula for the x-y symplectic transformation in topological recursion}",
    eprint = "2211.08917",
    archivePrefix = "arXiv",
    primaryClass = "math-ph",
    doi = "10.1016/j.geomphys.2023.105027",
    journal = "J. Geom. Phys.",
    volume = "194",
    year = "2023"
}

@article{Alexandrov:2022ydc,
    author = "A. Alexandrov and B. Bychkov and P. Dunin-Barkowski and M. Kazarian and S. Shadrin",
    title = "{A universal formula for the $x-y$ swap in topological recursion}",
    eprint = "2212.00320",
    archivePrefix = "arXiv",
    primaryClass = "math-ph",
    doi = "10.4171/JEMS/1615",
    year = "2025",
    Journal = {J. Eur. Math. Soc.}
}

@article{BEApol,
	Author = {G.~Borot and B.~Eynard},
	Doi = {10.4171/QT/60},
	Fjournal = {Quantum Topology},
	Issn = {1663-487X},
	Journal = {Quantum Topol.},
	Mrclass = {57M27 (15B52 81T70)},
	Mrreviewer = {Daniel D. Moskovich},
	Number = {1},
	Pages = {39--138},
	Title = {All order asymptotics of hyperbolic knot invariants from non-perturbative topological recursion of {A}-polynomials},
	Url = {https://doi.org/10.4171/QT/60},
	Volume = {6},
	Year = {2015}}

@article{BKMP,
	Author = {V.~Bouchard and A.~Klemm and M.~Mari{\~{n}}o and S.~Pasquetti},
	Doi = {10.1007/s00220-008-0620-4},
	Fjournal = {Communications in Mathematical Physics},
	Issn = {0010-3616},
	Journal = {Commun. Math. Phys.},
	Mrclass = {81T45 (32G81 32Q25 81T30)},
	Mrreviewer = {Johannes Walcher},
	Number = {1},
	Pages = {117--178},
	Title = {Remodeling the {B}-model},
	Url = {https://doi.org/10.1007/s00220-008-0620-4},
	Volume = {287},
	Year = {2009}}

@article{Mirzakhani,
	Author = {M.~Mirzakhani},
	Journal = {Invent. Math.},
	Pages = {179--222},
	Title = {Simple geodesics and {W}eil-{P}etersson volumes of moduli spaces of bordered {R}iemann surfaces},
	Volume = {167},
	Year = {2007}}

@article{Witten,
	Author = {E.~Witten},
	Journal = {Surveys in Diff. Geom.},
	Pages = {243-310},
	Title = {Two dimensional gravity and intersection theory on moduli space},
	Volume = {1},
	Year = {1991}}

@article{Kontsevich,
	Author = {M.~Kontsevich},
	Journal = {Commun. Math. Phys.},
	Pages = {1-23},
	Title = {Intersection theory on the moduli space of curves and the matrix {A}iry function},
	Volume = {147},
	Year = {1992}}

@article{BouchardEynard,
	Author = {V.~Bouchard and B.~Eynard},
	Doi = {10.1007/JHEP02(2013)143},
	Fjournal = {Journal of High Energy Physics},
	Issn = {1126-6708},
	Journal = {JHEP},
	Mrclass = {81T45 (30F10 32G15)},
	Mrreviewer = {Lee-Peng Teo},
	Number = {2},
	Title = {Think globally, compute locally},
	Url = {https://doi.org/10.1007/JHEP02(2013)143},
	Volume = {2013},
	Year = {2013}}

@article{HigherRam,
	Author = {V.~Bouchard and J.~Hutchinson and P.~Loliencar and M.~Meiers and M.~Rupert},
	Doi = {10.1007/s00023-013-0233-0},
	Fjournal = {Annales Henri Poincar\'e. A Journal of Theoretical and Mathematical Physics},
	Issn = {1424-0637},
	Journal = {Ann. Henri Poincar\'e},
	Mrclass = {81Q30 (30F30)},
	Number = {1},
	Pages = {143--169},
	Title = {A generalized topological recursion for arbitrary ramification},
	Url = {https://doi.org/10.1007/s00023-013-0233-0},
	Volume = {15},
	Year = {2014}}

@article{DN1,
	Author = {N.~Do and P.~Norbury},
	Journal = {J. Lond. Math. Soc.},
	Publisher = {Wiley Online Library},
	Title = {Topological recursion for irregular spectral curves},
	Year = {2014}}

@article{Borot:2013lpa,
    author = {G. Borot and B. Eynard and N. Orantin},
    title = "{Abstract loop equations, topological recursion and new applications}",
    eprint = "1303.5808",
    archivePrefix = "arXiv",
    primaryClass = "math-ph",
    doi = "10.4310/CNTP.2015.v9.n1.a2",
    journal = "Commun. Num. Theor. Phys.",
    volume = "09",
    pages = "51--187",
    year = "2015"
}

@article{Alexandrov:2024tjo,
    author = "A. Alexandrov and B. Bychkov and P. Dunin-Barkowski and M. Kazarian and S. Shadrin",
    title = "{Degenerate and Irregular Topological Recursion}",
    eprint = "2408.02608",
    archivePrefix = "arXiv",
    primaryClass = "math-ph",
    reportNumber = "MPIM-Bonn-2024",
    doi = "10.1007/s00220-025-05274-w",
    journal = "Commun. Math. Phys.",
    volume = "406",
    number = "5",
    pages = "94",
    year = "2025"
}

@article{EO07,
	Author = {B.~Eynard and N.~Orantin},
	Journal = {Commun. Number Theory Phys.},
	Number = {2},
	Title = {Invariants of algebraic curves and topological expansion},
	Volume = {1},
	Year = {2007}}

@article{Hock:2023qii,
    author = "A. Hock",
    title = "{Laplace transform of the $x-y$ symplectic transformation formula in Topological Recursion}",
    eprint = "2304.03032",
    archivePrefix = "arXiv",
    primaryClass = "math-ph",
    doi = "10.4310/CNTP.2023.v17.n4.a1",
    journal = "Commun. Num. Theor. Phys.",
    volume = "17",
    number = "4",
    pages = "821--845",
    year = "2023"
}

@article{Alexandrov:2023oov,
    author = "A. Alexandrov and B. Bychkov and P. Dunin-Barkowski and M. Kazarian and S. Shadrin",
    title = "{Topological recursion, symplectic duality, and generalized fully simple maps}",
    eprint = "2304.11687",
    archivePrefix = "arXiv",
    primaryClass = "math-ph",
    doi = "10.1016/j.geomphys.2024.105329",
    journal = "J. Geom. Phys.",
    volume = "206",
    pages = "105329",
    year = "2024"
}

@article{Alexandrov:2024hgu,
    author = "A. Alexandrov and B. Bychkov and P. Dunin-Barkowski and M. Kazarian and S. Shadrin",
    title = "{Any Topological Recursion on a Rational Spectral Curve is KP Integrable}",
    eprint = "2406.07391",
    archivePrefix = "arXiv",
    primaryClass = "math-ph",
    doi = "10.1007/s00220-026-05566-9",
    journal = "Commun. Math. Phys.",
    volume = "407",
    number = "4",
    pages = "69",
    year = "2026"
}

@article{Iqbal:2007ii,
    author = "A. Iqbal and C. Kozcaz and C. Vafa",
    title = "{The Refined topological vertex}",
    eprint = "hep-th/0701156",
    archivePrefix = "arXiv",
    doi = "10.1088/1126-6708/2009/10/069",
    journal = "JHEP",
    volume = "10",
    pages = "069",
    year = "2009"
}

@misc{Eynard:2011vs,
    author = "Eynard, B. and Kozcaz, C.",
    title = "{Mirror of the refined topological vertex from a matrix model}",
    eprint = "1107.5181",
    archivePrefix = "arXiv",
    primaryClass = "hep-th",
    reportNumber = "CERN-2011-182",
    note={arXiv:1107.5181 , 2011}
}

@article{Iqbal:2004ne,
    author = "A. Iqbal and A. Kashani-Poor",
    title = "{The Vertex on a strip}",
    eprint = "hep-th/0410174",
    archivePrefix = "arXiv",
    reportNumber = "SMS-0402, SLAC-PUB-10804, SU-ITP-04-39",
    doi = "10.4310/ATMP.2006.v10.n3.a2",
    journal = "Adv. Theor. Math. Phys.",
    volume = "10",
    number = "3",
    pages = "317--343",
    year = "2006"
}

@article {BorotGeometry,
    AUTHOR = {G. Borot and B. Eynard},
     TITLE = {Geometry of spectral curves and all order dispersive
              integrable system},
   JOURNAL = {SIGMA},
  FJOURNAL = {SIGMA. Symmetry, Integrability and Geometry. Methods and
              Applications},
    VOLUME = {8},
      YEAR = {2012},
      ISSN = {1815-0659},
   MRCLASS = {37K10 (14H70 30Fxx 35Q53)},
  MRNUMBER = {3007259},
MRREVIEWER = {Simonetta\ Abenda},
       DOI = {10.3842/SIGMA.2012.100},
       URL = {https://doi.org/10.3842/SIGMA.2012.100},
}

@article{Alexandrov:2023jcj,
    author = "A. Alexandrov and B. Bychkov and P. Dunin-Barkowski and M. Kazarian and S. Shadrin",
    title = "{KP integrability through the $x-y$ swap relation}",
    eprint = "2309.12176",
    archivePrefix = "arXiv",
    primaryClass = "math-ph",
    doi = "10.1007/s00029-025-01035-8",
    journal = "Selecta Math.",
    volume = "31",
    number = "2",
    pages = "42",
    year = "2025"
}

@article{Alexandrov:2023tgl,
    author = "A. Alexandrov and B. Bychkov and P. Dunin-Barkowski and M. Kazarian and S. Shadrin",
    title = "{Log Topological Recursion Through the Prism of x-y Swap}",
    eprint = "2312.16950",
    archivePrefix = "arXiv",
    primaryClass = "math-ph",
    doi = "10.1093/imrn/rnae213",
    journal = "Int. Math. Res. Not.",
    volume = "2024",
    number = "21",
    pages = "13461--13487",
    year = "2024"
}

@article{Brini:2011wi,
    author = "A. Brini and B. Eynard and M. Mari{\~n}o",
    title = "{Torus knots and mirror symmetry}",
    eprint = "1105.2012",
    archivePrefix = "arXiv",
    primaryClass = "hep-th",
    doi = "10.1007/s00023-012-0171-2",
    journal = "	Ann. Henri Poincar\'{e}",
    volume = "13",
    pages = "1873--1910",
    year = "2012"
}

@misc{schuler2026gromovwitteninvariantsmembraneindices,
      title={Gromov-Witten invariants and membrane indices of fivefolds via the topological vertex}, 
      author={Y. Schuler},
      eprint={2603.24256},
      archivePrefix={arXiv},
      primaryClass={math.AG},
      note={arXiv:2603.24256 , 2026}
}

@misc{Brini:2024gtn,
    author = "A. Brini and Y. Schuler",
    title = "{Refined Gromov-Witten invariants}",
    eprint = "2410.00118",
    archivePrefix = "arXiv",
    primaryClass = "math.AG",
    note={arXiv:2410.00118, 2024}
}

@article{Gu:2014yba,
    author = "J. Gu and H. Jockers and A. Klemm and M. Soroush",
    title = "{Knot Invariants from Topological Recursion on Augmentation Varieties}",
    eprint = "1401.5095",
    archivePrefix = "arXiv",
    primaryClass = "hep-th",
    reportNumber = "BONN-TH-2014-03, BRXTH-674",
    doi = "10.1007/s00220-014-2238-z",
    journal = "Commun. Math. Phys.",
    volume = "336",
    number = "2",
    pages = "987--1051",
    year = "2015"
}

@misc{Aganagic:2012jb,
    author = "M. Aganagic and C. Vafa",
    title = "{Large N Duality, Mirror Symmetry, and a Q-deformed A-polynomial for Knots}",
    eprint = "1204.4709",
    archivePrefix = "arXiv",
    primaryClass = "hep-th",
    Note = {arXiv:1204.4709, 2012},
}

@article{Banerjee:2025shz,
    author = "S. Banerjee and A. Hock",
    title = "{Quantum curve for strip geometries, topological recursion and open GW/DT invariants}",
    eprint = "2510.07146",
    archivePrefix = "arXiv",
    primaryClass = "math-ph",
    doi = "10.1007/s11005-026-02059-7",
    journal = "Lett. Math. Phys.",
    volume = "116",
    number = "1",
    pages = "22",
    year = "2026"
}

@article{Marchal:2017ntz,
    author = "O. Marchal",
    title = "{WKB solutions of difference equations and reconstruction by the topological recursion}",
    eprint = "1703.06152",
    archivePrefix = "arXiv",
    primaryClass = "math-ph",
    doi = "10.1088/1361-6544/aa92ed",
    journal = "Nonlinearity",
    volume = "31",
    number = "1",
    pages = "226--262",
    year = "2018"
}

@article{Bouchard:2023yau,
    author = {V. Bouchard and R. Kramer and Q. Weller},
    title = "{Topological recursion on transalgebraic spectral curves and Atlantes Hurwitz numbers}",
    doi = "10.1016/j.geomphys.2024.105306",
    journal = "J. Geom. Phys.",
    volume = "206",
    pages = "105306",
    year = "2024"
}
\end{document}